\documentclass{sig-alternate-org}

\usepackage{setspace}
\usepackage{graphicx}
\usepackage{amsmath}
\usepackage{amsfonts}
\usepackage{amssymb}
\usepackage{mdwmath} 
\usepackage{mdwtab} 
\usepackage{mathrsfs, euscript}
\usepackage[tight,footnotesize]{subfigure}
\usepackage{tabularx}
\usepackage{multirow}

\usepackage{url}
\usepackage{cite}
\usepackage[colorlinks=true,linkcolor=black,citecolor=black]{hyperref}
\usepackage{color}
\usepackage{booktabs}
\usepackage{amssymb}
\usepackage{float}
\usepackage{bbm}
\usepackage{caption}
\usepackage[boxed]{algorithm2e}
\usepackage{float}
\usepackage{setspace}
\floatstyle{boxed}
\usepackage{booktabs}
	\usepackage{bm}
\makeatletter
\def\@copyrightspace{\relax}
\makeatother
\allowdisplaybreaks[2]

\usepackage{lipsum}
 
\newcommand\blfootnote[1]{%
	\begingroup 
	\renewcommand\thefootnote{}\footnote{#1}%
	\addtocounter{footnote}{-1}%
	\endgroup 
}

\begin{document}
	
%
%
%
	\makeatletter
	\newcommand{\removelatexerror}{\let\@latex@error\@gobble}
	\makeatother
	
	\makeatletter
	\setlength{\@fptop}{0pt}
	\setlength{\@fpbot}{0pt plus 1fil}
	\makeatother
	
	\newtheorem{theorem}{Theorem}[section]
	\newtheorem{definition}{Definition}[section]
	\newtheorem{lemma}{Lemma}[section]
	\newtheorem{corollary}{Corollary}[section]
	\newtheorem{proposition}{Proposition}[section]
	\newtheorem{example}{Example}[section]
	\newtheorem{remark}{Remark}[section]
	\newtheorem{assumption}{Assumption}[section]

\title{De-anonymization of Social Networks with Communities: When Quantifications Meet Algorithms\vspace{-8mm}}
\numberofauthors{5}
	\author{
		 \alignauthor
		 Luoyi Fu$^*$\\
		 \affaddr{Shanghai Jiao Tong University}\\
		 \email{yiluofu@sjtu.edu.cn} 
	\alignauthor
	Xinzhe Fu$^*$\\
	       \affaddr{Shanghai Jiao Tong University}\\
	       \email{fxz0114@sjtu.edu.cn}   
	\alignauthor
	Zhongzhao Hu\\
	\affaddr{Shanghai Jiao Tong University}\\
	\email{hzz5611577@sjtu.edu.cn}
	\and
	\alignauthor
	Zhiying Xu\\
	\affaddr{Shanghai Jiao Tong University}\\
	\email{xuzhiying@sjtu.edu.cn}	
		\alignauthor
	Xinbing Wang\\
	\affaddr{Shanghai Jiao Tong University}\\
	\email{xwang8@sjtu.edu.cn}
	    }

\maketitle

\begin{abstract}
A crucial privacy-driven issue nowadays is re-identifying ano-nymized social networks by mapping them to correlated cross-domain auxiliary networks. Prior works are typically based on modeling social networks as random graphs representing users and their relations, and subsequently quantify the quality of mappings through cost functions that are proposed without sufficient rationale. Also, it remains unknown how to algorithmically meet the demand of such quantifications, i.e., to find the minimizer of the cost functions.

We address those concerns in a more realistic social network modeling parameterized by community structures that can be leveraged as side information for de-anonymization. By Maximum A Posteriori (MAP) estimation, our first contribution is new and well justified cost functions, which, when minimized, enjoy superiority to previous ones in finding the correct mapping with the highest probability. The feasibility of the cost functions is then for the first time algorithmically characterized. While proving the general multiplicative inapproximability, we are able to propose two algorithms, which, respectively, enjoy an $\epsilon$-additive approximation and a conditional optimality in carrying out successful user re-identification. Our theoretical findings are empirically validated, with a notable dataset extracted from rare true cross-domain networks that reproduce genuine social network de-anonymization. Both theoretical and empirical observations also manifest the importance of community information in enhancing privacy inferencing.

\end{abstract}
\blfootnote{*: the first two authors contributed equally to this paper.}
\vspace{-2mm}
\section{Introduction}
The proliferation of social networks has led to generation of massive network data. Although users can be anonymized in the released data through removing personal identifiers \cite{cite:SNAP,cite:targeted-advertising}, with their underlying relations preserved, they may still be re-identified by adversaries from correlated cross domain auxiliary networks where user identities are known \cite{cite:social-privacy,cite:de-anonymization, cite:xiaoyang}.
%

Such idea of unveiling hidden users by leveraging their information collected from other domains, or alternatively called social network de-anonymization \cite{cite:de-anonymization}, is a fundamental privacy issue that has received considerable attention. Inspired by Pedarsani and Grossglauser \cite{cite:seedless}, a large body of existing de-anonymization work shares a basic common paradigm: with an underlying network representing social relations between users, both the \emph{published anonymized network} and the \emph{auxiliary un-anonymized network} are generated from that network based on graph sampling that captures their correlation, as observed in many real cross-domain networks. The equivalent node sets they share are corresponded by an unknown correct mapping. With the availability of only structural information, adversaries attempt to re-identify users by establishing a mapping between networks. To quantify such mapping qualities, several global cost functions have been proposed \cite{cite:seedless,cite:allerton,cite:improved-bound} in favor of exploring the conditions under which the correct matching can be unraveled from the mapping that minimizes the cost function.

Despite those dedications to de-anonymization, it is still not entirely understood how the privacy of anonymized social network can be guaranteed given that adversaries have no access to side information other than network structure, primarily for three reasons. First, the widely adopted Erd\H{o}s-R\'{e}nyi graph or Chung-Lu graph \cite{cite:ER-Graph,cite:ChungLu-Graph} for the modeling of underlying social networks \cite{cite:seedless,cite:allerton,cite:improved-bound}, though facilitating analysis, falls short of well capturing the clustering effects that are prevalent in realistic social networks; Second, the cost functions \cite{cite:seedless,cite:allerton} in measuring mapping qualities not only lack sufficient rationale in analytical aspects, but most importantly, it remains unclear whether the feasibility of minimizing such cost functions could be theoretically characterized from an algorithmic aspect \cite{cite:arxiv-community,cite:shouling1}; Last but not least, due to the rarity of true cross-domain datasets, current empirical observations of social network de-anonymization are either based on synthetic data, or real social networks with artificial sampling in construction of correlated published and auxiliary networks, and consequently do not well represent the genuine practical de-anonymization \cite{cite:shouling1,cite:shouling2}. While a thorough understanding of this issue may better inform us on user privacy protection, this paper is particularly concerned about the following question: \textbf{Is it possible to quantify  de-anonymization in a more realistic modeling, and meanwhile algorithmically meet the demand  brought by such quantifications?}

The answer to this question entails appropriate modeling of social networks, well-designed cost functions as metrics of mappings and elaborated algorithms of finding the mapping that is optimal according to the metric, along with data collection that can empirically validate the related claims. To present a more reasonable model of underlying social network that incorporates the clustering effect, we adopt the stochastic block model \cite{cite:blockmodel} where nodes are partitioned into disjoint sets representing different communities \cite{cite:model}. Based on that, we investigate the problem following the paradigm, as noted earlier, where the published and auxiliary networks serve as two sampled subnetworks. Both of them inherit from the underlying network the community structures that can be leveraged as side structural information for adversaries. Similarly, we assume that other than network structure, there is no additional availability of side information to adversaries as it will only further benefit them. Varying the amount of availability of community information, here we classify our de-anonymization problem into two categories, i.e., bilateral case, and its counterpart, unilateral case, literally meaning that adversaries have access to community structure of both or only one network. A more formal definition of the two cases information is deferred to Section \ref{sec:model}. Subsequently, we summarize, built on the model, our results on metrics, algorithms and empirical validations into three aspects answering the question raised.

\textbf{Analytical aspect:} For both cases, our first contribution is to derive the cost functions as metrics quantifying the structural
mismappings between networks based on Maximum A
Posteriori (MAP) estimation. The virtue of MAP estimation ensures the superiority of our metrics to the previous ones in the sense that the minimizers of our cost functions equal to the underlying correct mappings with the highest probability. Also, as we will rigorously prove later, under fairly mild conditions on network density and the closeness between communities, through minimizing the cost function we can perfectly recover the correct mapping.

\textbf{Algorithmic aspect}: Following the derived quantifications, our next significant contribution is to take a first algorithmic look into the demand imposed by the quantifications, i.e., the optimization problems of minimizing such cost functions. We find that opposed to the simplicity of the cost functions in form, the induced optimization problems are computationally intractable and highly inapproximable. Therefore, we circumvent pursing exact or multiplicative approximation algorithms, but instead seek for algorithms with other types of guarantees. However, the issue is still made particularly challenging by the intricate tension among cost function, mappings, network topology as well as the super-exponentially large number of candidate mappings. Our main idea to resolve the tension is converting the problems into equivalent formulations that enable some relaxations, through bounding the influence of which, we demonstrate that the proposed algorithms have their respective performance guarantees. Specifically, one algorithm enjoys an $\epsilon$-additive approximation guarantee in both cases, while the other yields optimal solutions for bilateral de-anonymization when the two sub-networks are highly structurally similar but fails to provide such guarantee for the unilateral case due to its lack of sufficient community information. Further comparisons of algorithmic results between the two cases also manifest the importance of community as side information in privacy inferencing.

\textbf{Experimental aspect:} Finally, we empirically verified all our theoretical findings under both synthetic and real datasets. We remark that one dataset, as never appeared in this context previously, is extracted from true cross-domain co-authorship networks \cite{cite:MAG} serving as published and auxiliary networks. As a result, it leads to no prior work, other than ours, that reproduces genuine scenarios of social network de-anonymization without artificial modeling assumptions. The experimental results demonstrate the effectiveness of our algorithms as they correctly re-identify more than 40\% of users even in the co-authorship networks that possess the largest deviation from our assumptions. Also, it empirically consolidates
our argument that community information can increase
the de-anonymization capability.

The rest of this paper is organized as follows: In Section \ref{sec:relatedworks}, we briefly survey the related works. In Section \ref{sec:model}, we introduce our model for de-anonymization problem of social networks with community structure and characterize the cases of bilateral and unilateral information. In Sections \ref{sec:bilateral1} and \ref{sec:bilateral2}, we present our results on analytical and algorithmic aspects of bilateral de-anonymization. Following the path of bilateral case, we introduce our results on unilateral de-anonymization and make comparisons between the two cases in Section \ref{sec:unilateral}. We present our experiments in Section \ref{sec:experiments} and conclude the paper in Section \ref{sec:conclusion}.
\vspace{-2mm}
\section{Related Works}\label{sec:relatedworks}
The issue of social network de-anonymization, which has received considerable attention, was pioneeringly investigated by Narayanan and Shimatikov \cite{cite:de-anonymization}, who proposed the idea that users in anonymized networks can be re-identified through utilizing auxiliary networks with the same set of users from other domains. In that regard, they designed practical de-anonymization schemes that rely on side information in the form of a seed
set of ``pre-mapped" node pairs, i.e., a subset of nodes that
are identified priorly across the two networks. Then the
mapping is generated incrementally, starting from the seeds
and percolating to the whole node sets.

Following this framework, Pedarsani and Grossglauser developed a succinct modeling that is amiable to theoretical analysis and serves as the paradigm for a family of subsequent related works on social network de-anonymization \cite{cite:seedless}. They assumed that the published and auxiliary networks are two graphs that share the same node sets with the edge sets resulted from independent samples of an underlying social network. Additionally, they studied a more challenging but practical version of de-anonymization that are free of prior seed information.

The two seminal works triggered a flurry of subsequent attempts that all fall into the categories of either seeded or seedless de-anonymization, tuning the model of the underlying social networks. Specifically, in terms of seeded de-anonymization, current literature focuses on designing efficient de-anonymization algorithms that are executed by percolating the mapping to the whole node sets starting from the seed set. Yartseva et al. \cite{cite:percolation-matching}, Kazemi et al. \cite{cite:vldb1}, and later Fabiana et al. \cite{cite:garetto1} proposed percolation graph matching algorithms for de-anonymization on Erd\H{o}s-R\'{e}nyi graph and scale-free network, respectively. Assuming that the underlying social network is generated following the preferential attachment model, Korula and Lattenzi \cite{cite:vldb} designed a correspnding efficient de-anonymizaiton algorithm. Chiasserini et al. \cite{cite:garetto2} characterized the impact that clustering imposes on the performance of seeded de-anonymization. Under the classification of both perfect and imperfect seeded de-anonymization, Ji et al. \cite{cite:shouling2} analyzed the two cases both qualitatively and empirically.

While this type of seed-based de-anonymizing methods
works well in analysis, it is rather difficult to acquire pre-
identified user pairs across different networks as many real
situations limit the access to user profiles. Therefore,
more often we are faced with adversaries without seeds as
side information, which is also the case considered in the present work. A natural alternative, under such circumstance, is to define a global cost function of mappings and unravel the correct mapping through the minimizer of the cost function. For instance, Pedarsani and Grossglauer \cite{cite:seedless} studied the seedless de-anonymization problem where the underlying social network is an Erd\H{o}s-R\'{e}nyi graph, the results of which were further improved by Cullina and Kiyavash \cite{cite:improved-bound}. Ji et al. analyzed perfect and partial de-anonymization on Chung-Lu graph \cite{cite:model}. Kazemi et al. \cite{cite:allerton} focused on the case of de-anonymization problem on Erd\H{o}s-R\'{e}nyi graph where the published network and auxiliary network exhibit partial overlapping. A very recent work that shares the highest correlation with ours, belongs to that of Onaran et al. \cite{cite:arxiv-community}, who study the situation where there are only two communities in networks, a special case that can be embodied in our bilateral de-anonymization case.

\begin{figure}
	\centering
	\includegraphics[width=1\linewidth]{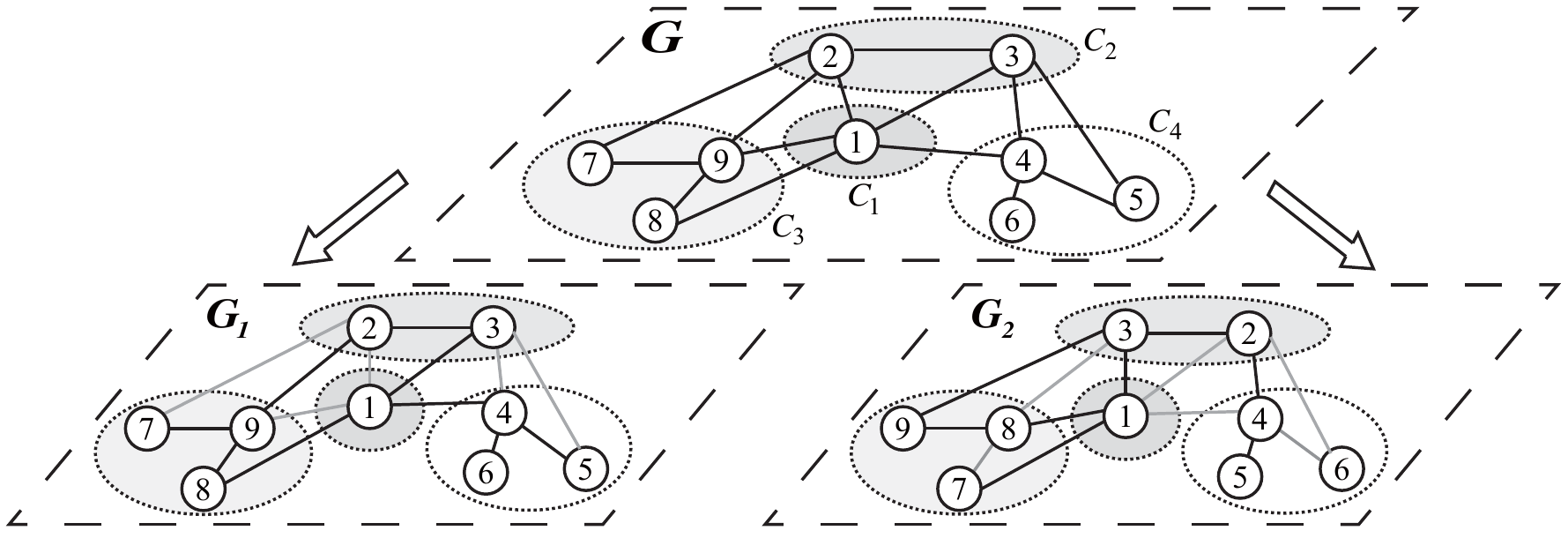}
	\vspace{-6mm}
	\caption{
		\small\bf An example of underlying social network ($G$), the published network ($G_1$) and the auxiliary network ($G_2$) sampled from $G$. $C_1,C_2,C_3,C_4$ represent the four communities in the networks. 
		The correct mapping $\pi_0=\{(1,1),(2,3),(3,2),(4,4),(5,6),(6,5),(7,9),(8,7),(9,8)\}$.
	}
	\label{fig:model}
	\vspace{-6mm}
\end{figure}
\vspace{-1mm}
\section{Models and Definitions}\label{sec:model}
\vspace{-0.5mm}
In this section, we introduce the models and definitions of the social network de-anonymization problem. We first present the network models and then formally define the problem of social network de-anonymization.
\vspace{-1.4mm} 
\subsection{Network Models}\vspace{-1mm}
The network models consist of the underlying social networks $G$, the published network $G_1$ and the auxiliary network $G_2$ as incomplete observations of $G$. In reality, the edges of $G$, for example, might represent the true relationships between a set of people, while $G_1$ and $G_2$ characterize the observable interactions between these people such as communication records in cell phones or ``follow" relationships in online social networks.\vspace{-1.5mm}
\subsubsection{Underlying Social Network}\vspace{-0.7mm}
To elaborate this, let $G=(V,E,\mathbf{M})$\footnote{For a matrix $\mathbf{M}$, we use $\textbf{M}_{ij}$ to denote the element on its $i$th row and $j$th column and $\mathbf{M}_i$ to denote its $i$th row vector.} be the graph representing the underlying social relationships between network nodes, where $V$ is the set of nodes, $E$ is the set of edges and $\mathbf{M}$\footnote{ $\mathbf{M}_{ij}=1$ if $(i,j)\in E$ and $\mathbf{M}_{ij}=0$ otherwise.} denotes the adjacency matrix of $G$. We treat $G$ as an undirected graph and define the number of nodes as $|V|=n$. We assume that $G$ is generated according to the \textit{stochastic block model} \cite{cite:blockmodel}. Specifically, the model is interpreted as follows: the set of nodes in $V$ are partitioned into $\kappa$ disjoint subsets denoted as $C_1,C_2,\ldots,C_{\kappa}$ indicating their communities with $|C_i|=n_i$ and $\sum_in_i=n$. The edges between nodes in different communities are drawn independently at random with certain probabilities. Let $c:V\mapsto \{1\ldots\kappa\}$ be the community assignment function that assigns to each node the label of the community it belongs to, we have
\begin{equation*}
Pr\{(u,v)\in E\}=Pr\{\mathbf{M}_{uv}=1\}=p_{c{(u)}c{(v)}},
\end{equation*}
where affinity values $\{p\}_{ab}$ ($1\le a,b\le\kappa$) are pre-defined parameters that indicate the edge existence probabilities and capture the closeness between communities.
It has been shown that this model well captures the community structures in social networks and can generate graphs with various degree distributions by tuning the values of $\{p\}$ \cite{cite:community}. 

\vspace{-1.6mm}
\subsubsection{Published Network and Auxiliary Network}
\vspace{-1mm}
We define $G_1(V_1,E_1,\mathbf{A})$ as the graph representing the published network and $G_2(V_2,E_2,\mathbf{B})$ as the graph representing the auxiliary network with $E_1,E_2$ denoting their edge sets and $\mathbf{A,B}$ denoting their adjacency matrices respectively. In correspondence to real situations, $G_1$ represents the publicly available anonymized network where user identities are removed for privacy concern.  In contrast, $G_2$ represents the auxiliary cross-domain un-anonymized network where those users' identities are known, and can be collected by the adversary to re-identify the users in $G_1$. Following previous literature \cite{cite:seedless,cite:shouling2}, we assume the node sets in $G_1$ and $G_2$ are equivalent and that the published network and the auxiliary network are independent samples obtained from the underlying social network $G$ with sampling probabilities $s_1$ and $s_2$, respectively. Specifically, for $i=1,2$, we have  
	\[Pr\{(u,v)\in E_i\}=\left\{\begin{array}{ll}
	s_i &\mbox{ if }(u,v)\in E,\\
	0& \mbox{ otherwise}.\\
	\end{array}\right.\]

Technically, $G$, $G_1$ and $G_2$ are defined as the random graph variables for the networks. However, for ease of representation, we will also use them to denote the realizations of the random graph variables without loss of clearance. In the sequel, we will also use $\bm{\theta}$ as a shorthand of the set of parameters including affinity values $\{p\}$ and sampling probabilities $s_1,s_2$ in the models of $G,G_1,G_2$, 
\subsection{Social Network De-anonymization}
Given the published network $G_1$ and the auxiliary network $G_2$, the problem of social network de-anonymization aims to find a bijective mapping $\pi:V_1\mapsto V_2$ that reveals the correct correspondence of the nodes in the two networks. Equivalently, a mapping $\pi$\footnote{In this paper, all  the mappings are assumed to be bijective. Hence, we simply refer to them as mappings for brevity.} can be represented as a permutation matrix $\mathbf{\Pi}$ where $\mathbf{\Pi}_{ij}=1$ if $\pi(i)=j$ and $\mathbf{\Pi}_{ij}=0$ otherwise. We naturally extend the definition of mapping of node set to the mapping of edge set, as $\pi(e=(i,j))=(\pi(i),\pi(j))$. 

We define $\pi_0$ (or equivalently $\mathbf{\Pi}_0$) to be the correct mapping between the node sets of $G_1$ and $G_2$. Note that we do not have access to $\pi_0$ or the generator $G$ of $G_1$ and $G_2$. In other words, although the node sets of $G_1$ and $G_2$ are equivalent, the labeling of the nodes does not reflect their underlying correspondence. We interpret this in the way that the published network $G_1$ has the same node labeling as the underlying network $G$ while the node labeling of $G_2$ is permuted. Following this interpretation, the community assignment function of $G_1$ equals to $c$. However the community assignment function of $G_2$, which we further define as $c'$, may be different. We illustrate an example of our network models in Figure \ref{fig:model}.

The community assignment functions of the two networks may serve as important structural side information for de-anonymization, which naturally divide the social network de-anonymization problem into two types where the adversary possesses different amount of information on the community assignment. In the first type, the adversary possesses the community assignments of both $G_1$ and $G_2$. The corresponding problem is formally defined as follows.
\vspace{-1mm}
\begin{definition}\textbf{(De-anonymization with Bilateral \\Community Information)}\
	Given the published network $G_1$, the auxiliary network $G_2$, the parameters $\bm{\theta}$, as well as the community assignment function $c$ for $G_1$ and $c'$ for $G_2$, the goal is to construct a mapping $\pi$ that satisfies $\forall i, c(i)=c'(\pi(i))$ and is closest to the correct mapping $\pi_0$.
\end{definition}
\vspace{-2mm}
Since in this case, we have the community assignment of $G_2$, we can perform a relabeling on nodes in $G_2$ to make its community assignment equals to that of $G_1$. Hence, without loss of generality, for the case of de-anonymization with bilateral information, we denote $c$ as the community assignment function of both $G_1$ and $G_2$ in the sequel.

The second variant corresponds to the case where the adversary only possesses the community assignment of the published network, which is formally stated as follows.
\vspace{-1.5mm}
\begin{definition}(\textbf{De-anonymization with Unilateral Community Information})\
	Given the published network $G_1$, the auxiliary network $G_2$, parameters $\bm{\theta}$, as well as the community assignment function $c$ for $G_1$, the goal is to construct a mapping that is closest to the correct mapping $\pi_0$.
\end{definition}
\vspace{-1.5mm}
Intuitively, de-anonymization with unilateral information is harder than that with bilateral information due to the lack of side information. We will validate this argument with subsequent theoretical analysis and experiments. In addition, for brevity, we may refer to de-anonymization problem with bilateral community information and with unilateral community information as bilateral de-anonymization and unilateral de-anonymization respectively.

\textbf{Remark: }Till now, we have not given the quantifying metric of the closeness to the correct mapping $\pi_0$. A natural choice would be the mapping accuracy, i.e., percentage of nodes that are mapped identically as in $\pi_0$. However, as we have no knowledge of $\pi_0$, such ground-truth-based metrics do not apply. To tackle this, we leverage the Maximum A Posteriori (MAP) estimator to construct cost functions for measuring the quality of mappings based solely on observable information. The main notations used throughout the paper are summarized in Table~\ref{table:notation}.

\normalsize
				\begin{table}[!tb]
					\setlength{\extrarowheight}{2pt}
					\begin{scriptsize}
						\renewcommand\arraystretch{0.82}
						\caption{\bf Notions and Definitions}
						\vspace{-3mm}
						\centering
						\label{table:notation}
						\resizebox{1.0\columnwidth}{!}{
							
							\begin{tabular}{l|l}\hline\label{table:notation1}
								\textbf{Notation} & \textbf{Definition} \\\hline
								$G$ & Underlying social network \\ 
								$G_1,G_2$ & Published and auxiliary networks \\ 
								$V,V_1,V_2$ & Vertex sets of graphs $G$, $G_1$ and $G_2$ \\ 
								$E,E_{1},E_{2}$ & Edge sets of graphs $G$, $G_1$, $G_2$ \\  
								$s_1,s_2$ & Edge sampling probabilities of graphs $G_1$, $G_2$ \\  
								$\mathbf{M,A,B}$ & Adjacency matrices of graphs $G$, $G_1$, $G_2$ \\  
								$c$ & Community assignment function \\
								$C_i$ & Vertex set of community $i$ \\ 
								$n$ & Total number of vertices \\
								$\kappa$ & Total number of communities \\ 
								$n_i$ & Number of vertices in community $i$ \\ 
								$p_{ab}$ & Affinity value indicating the edge existence\\& probability  between communities $a$ and $b$ \\ 
								$\bm{\theta}$ & Set of parameters in the models\\& of $G$, $G_1$ and $G_2$ \\ 
								$\pi_0$ & Correct mapping between vertices in $G_1$ and $G_2$ \\
								$\pi$ & Mapping between vertices in $G_1$ and $G_2$ \\
								$\mathbf{\Pi}$ & Permutation matrix of mapping $\pi$ \\
								$\Delta_{\pi}$ & Cost function of the mappings\\
								$\{w\}$ & Set of weights in the cost function\\ \hline
							\end{tabular}}
							
							\vspace{-4.6mm}
						\end{scriptsize}
					\end{table}

\vspace{-2mm}
\section{Analytical Aspect of Bilateral De-anonymization}\label{sec:bilateral1}  
First, we investigate the de-anonymization problem with bilateral information, starting with an appropriate metric measuring the quality of mappings. We define our proposed metric in the form of a cost function that derived from Maximum A Posteriori (MAP) estimation.
\vspace{-1mm}
\subsection{MAP-based Cost Function}
According to the definition of MAP estimation, given the published network $G_1$, auxiliary network $G_2$, parameters $\bm{\theta}$ and the community assignment function $c$, the MAP estimate $\hat{\pi}$ of the correct mapping $\pi_0$ is defined as:
\belowdisplayskip=-1pt
\vspace{-1.5mm}
\begin{align}\label{eq:mapestimator}
\hat{\pi}= \arg\max_{\pi\in \Pi}Pr(\pi_0=\pi\mid G_1,G_2,c,\bm{\theta}),
\end{align}
where $\Pi=\{\pi:V_1\mapsto V_2\mid \forall i, c(i)=c(\pi(i))  \}$, i.e. the set of bijective mappings that observe the community assignment. 

From the results in \cite{cite:arxiv-community}, the MAP estimator in Equation (\ref{eq:mapestimator}) can be computed as
\belowdisplayskip=3pt
\belowdisplayshortskip=3pt
\abovedisplayshortskip=3pt
\abovedisplayskip=3pt
\vspace{-2mm}
\begin{small}
\begin{align}
\hat{\pi}&=\arg\min_{\pi\in \Pi}\sum_{i\le j}^{n}w_{ij}\left|\mathbbm{1}\{(i,j)\in E_1\}-\mathbbm{1}\{(\pi(i),\pi(j))\in E_2\}\right|\label{eq:MAP}\\[-2pt]
 &\triangleq\arg\min_{\pi\in\Pi}\Delta_{\pi}, \vspace{-2mm}\nonumber
\end{align}
\end{small}where $w_{ij}=\log\left(\frac{1-p_{c(i)c(j)}(s_1+s_2-s_1s_2)}{p_{c(i)c(j)}(1-s_1)(1-s_2)}\right).$ Based on Equation (\ref{eq:MAP}), we have our cost function $\Delta_{\pi}$ as the metric for the quality of mappings, which can also be interpreted as weighted edge disagreements induced by mappings.
\subsection{Validity of the Cost Function}
Since our cost function $\Delta_{\pi}$ is derived using the MAP estimation,  the minimizer of $\Delta_{\pi}$, being the MAP estimate of $\pi_0$, coincides with the correct mapping with the highest probability \cite{cite:MAP-base}. Aside from this, we proceed to justify the use of MAP estimation in de-anonymization problem from another perspective. Specifically, we prove that if the model parameters satisfy certain conditions, then the MAP estimate $\hat{\pi}$ \textit{asymptotically almost surely}\footnote{An event \textit{asymptotically almost surely} happens if it happens with probability $1-o(1)$.} coincides with the correct mapping $\pi_0$, which means that we can perfectly recover the correct mapping through minimizing $\Delta_{\pi}$.
\vspace{-2mm}
\begin{theorem}\label{theorem:MAP}
	Let $\alpha=\min_{ab}p_{ab}, \beta=\max_{ab}p_{ab}$, $\overline{w}=\max_{ij}w_{ij}$ and $\underline{w}=\min_{ij}w_{ij}$. Assume that $\alpha,\beta\rightarrow 0$, $s_1,s_2$ do not go to 1 as $n\rightarrow \infty$ and $\frac{\log \alpha}{\log \beta}\le \gamma$. Suppose that
	\begin{align*}
	\frac{\alpha(1-\beta)^2s_1^2s_2^2\log(1/\alpha)}{s_1+s_2}=\Omega\left({\frac{\gamma\log^2 n}{n}}\right)+\omega\left(\frac{1}{n}\right)\footnotemark,
	\end{align*}
	then $\hat{\pi}=\pi_0$ holds almost surely as $n\rightarrow \infty$.
\end{theorem}
\vspace{-4mm}
\footnotetext{We use standard Knuth's notations in this paper.}
\begin{proof}
	Due to space limitations, here we only presenting an outline of the proof and defer the details to \textbf{Appendix \ref{sec:MAP-Estimate}}. 
	Recall that for a mapping $\pi$, we define $\Delta_{\pi}=\sum_{i\le j}^{n}w_{ij}$$|\mathbbm{1}\{(i,j)\in E_1\}-$$\mathbbm{1}\{\pi(i),\pi(j)\in E_2\}|$. Also, we denote $\Pi_k$ as the set of mappings that map $k$ nodes incorrectly and $S_k$ as a random variable representing the the number of mappings $\pi\in\Pi_k$ with $\Delta_{\pi}\le \Delta_{\pi_0}$. We then define $S=\sum_{k=2}^nS_k$ as the total number of incorrect mappings $\pi$ with $\Delta_{\pi}\le \Delta_{\pi_0}$ and derive an upper bound on the mean of $S$ as $\mathbb{E}[S]\le \sum_{k=2}^{n}n^k\max_{\pi\in\Pi_k}Pr\{\Delta_\pi-\Delta_{\pi_0}\le 0\}$. We further show that under the conditions stated in the theorem, this upper bound, and consequently $\mathbb{E}[S]$, go to 0 as $n\rightarrow \infty$, which implies that $\pi_0$ is the unique minimizer of $\Delta_{\pi}$ and concludes the proof. \qed

\end{proof}
\vspace{-2mm}
\textbf{Remark: }We now present two further notes regarding Theorem \ref{theorem:MAP}. (i) \textit{Applicability of the Theorem:} 
Recall that for a random Erd\H{o}s-R\'{e}nyi graph $G(n,p)$ to be connected and free of isolated nodes with high probability, it must satisfy $p=\Omega(\frac{\log n}{n})$ \cite{cite:ER-Graph}, and the absence of isolated nodes is necessary for successful de-anonymization since there is no way that we can distinguish the isolated nodes in $G_1$ and $G_2$. Conventionally setting the sampling probabilities $s_1,s_2$ as constants, it is easy to verify that the conditions in Theorem \ref{theorem:MAP} only have constant gap from the graph connectivity conditions even when the expected degree distributions (or equivalently, the closeness between the communities) of $G_1$ and $G_2$ are non-uniform (e.g. power law distribution where $\alpha/\beta =O(n)$ and $\log \alpha/\log\beta =O(\log n)$). From this aspect, the conditions are quite mild and thus make Theorem \ref{theorem:MAP} fairly general; 
(ii) \textit{Extension of the Theorem:} The cost function we design is robust, in the sense that any approximate minimizer $\Delta_{\pi}$ can map most of the nodes correctly. We formally present the claim in Corollary \ref{corollary:accuracy}. 
\vspace{-1.5mm}
\begin{corollary}\label{corollary:accuracy}
	Let $\alpha,\beta,\overline{w},\underline{w}$ be the same parameters defined in Theorem \ref{theorem:MAP}. Assume that $\alpha,\beta,s_1,s_2$ do not go to 0 and $\frac{\log \alpha}{\log\beta}\le \gamma$. Additionally, let $\delta,\epsilon$ be two real numbers with $0\le \delta,\epsilon \le 1$ with $\epsilon=O(\delta-\frac{\delta^2}{2})\alpha(1-\beta)s_1s_2\log(1/\alpha).$~If 
	\vspace{-1mm}
	\abovedisplayskip=2pt
	\begin{small}
	\begin{equation*}
	\frac{\alpha(1-\beta)^2s_1^2s_2^2\log(1/\alpha)}{s_1+s_2}=\Omega\left({\frac{\gamma\log^2 n}{(1-\delta/2)n}}\right)+\omega\left(\frac{1}{n}\right),
	\end{equation*}
	\end{small}
	then for all $\pi^*$ with $\Delta_{\pi^*}-\min_{\pi\in\Pi}\Delta_\pi\le \epsilon n^2$, $\pi^*$ is guaranteed to map at least $(1-\delta)n$ nodes correctly as $n\rightarrow \infty$.
\end{corollary}
\vspace{-3.5mm}
\begin{proof}
	The proof is similar to that of Theorem \ref{theorem:MAP}. Instead of bounding $\sum_{k=2}^n\sum_{\pi\in\Pi_k}Pr\{\Delta_\pi-\Delta_{\pi_0}\le 0\}$, we upper bound $\sum_{k=\delta n}^n\sum_{\pi\in\Pi_k}Pr\{\Delta_\pi-\Delta_{\pi_0}\le \epsilon n^2\}$. Using similar technique as in Theorem \ref{theorem:MAP}, we have that under the conditions stated in the corollary, $\sum_{k=\delta n}^n\sum_{\pi\in\Pi_k}Pr\{\Delta_\pi-\Delta_{\pi_0}\le \epsilon n^2\}\rightarrow 0$ as $n\rightarrow \infty$. Therefore, for a mapping $\pi^*$ with $\Delta_{\pi^*}-\Delta_{\pi_0} \le \epsilon n^2$, it maps at most $k=\delta n$ nodes incorrectly. Since $\Delta_{\pi_0}\ge\arg\min_{\pi\in\Pi}\Delta_{\pi}$, we conclude that all $\pi^*$ with $\Delta_{\pi^*}-\min_{\pi\in\Pi}\Delta_\pi\le \epsilon n^2$ are guaranteed to map at least $(1-\delta)n$ nodes correctly as $n\rightarrow \infty$. 
\end{proof}

\vspace{-5mm}
\section{Algorithmic Aspect of Bilateral De-anonymization}\label{sec:bilateral2}
\vspace{-0.7mm}
The quantification in Section \ref{sec:bilateral1} justified that, under mild conditions, we can unravel the correct mapping through computing its MAP estimate, i.e., the minimizer of $\Delta_{\pi}$. This naturally puts forward the optimization problem of computing the minimizer of $\Delta_{\pi}$, which reasonably serves as the instantiation of the social network de-anonymization problem (Definition 3.1).  To meet the demand of the quantification, in this section, we formally define and investigate this optimization problem, presenting a first look into the algorithmic aspect of social network de-anonymization.
\vspace{-1mm}
\subsection{The Bilateral MAP-ESTIMATE Problem}
\vspace{-0.5mm}
Naturally, with some previously defined notations inherited, the optimization problem induced by the cost function can be formulated as follows.
\vspace{-1.5mm}
\begin{definition}{(\bf{The BI-MAP-ESTIMATE Problem})}\ \label{def:BI-MAP-ESTIMATE}
	Given two graphs $G_1(V_1,E_1,\mathbf{A})$ and $G_2(V_2,E_2,\mathbf{B})$, community assignment function $c$ and a set of weights $\{w\}$, the goal is to compute a mapping $\hat{\pi}:V_1\mapsto V_2$ that satisfies
	\belowdisplayskip=-1pt
	\vspace{-2.7mm}
	\begin{small}
	\begin{equation*}
	\begin{aligned}
	\mathbf{P1:}\ \ \quad
	\hat{\pi}&=\arg\min_{\pi\in \Pi}\sum_{i\le j}^{n}w_{ij}\left|\mathbbm{1}\{(i,j)\in E_1\}-\mathbbm{1}\{\pi(i),\pi(j)\in E_2\}\right|\\[-3pt]
	&\triangleq\arg\min_{\pi\in\Pi}\Delta_{\pi}, 
	\end{aligned}
		\vspace{-2mm}
	\end{equation*} 
	\end{small}
	where $\Pi=\{\pi\mid \forall i,c(i)=c(\pi(i))\}$. 
\end{definition}
\vspace{-2mm}
Note that we require the weights $\{w\}$ to be induced by implicit and well-defined community affinity values and sampling probabilities. Also, the BI-MAP-ESTIMATE Problem denoted as $\mathbf{P1}$ above has several equivalent formulations, which will be presented later.  

The BI-MAP-ESTIMATE seems to be easy at first glance due to the simplicity of its objective function $\Delta_{\pi}$, but as justified by the following proposition, it is not only computationally intractable but also highly inapproximable. 
\vspace{-2mm}
\begin{proposition}\label{proposition:hardness}
	BI-MAP-ESTIMATE problem is NP-hard. And there is no polynomial time (pseudo-polynomial time) approximation algorithm for BI-MAP-ESTIMATE with any multiplicative approximation guarantee unless $GI\in P$ ($GI\in DTIME(n^{\mathrm{polylog}n})$).\footnote{$GI$ denotes the complexity class Graph Isomorphsim.}
\end{proposition}
\vspace{-3mm}
\normalsize
\begin{proof}
The proof can be easily constructed by reduction from the graph isomoprhism problem. The reduction is completed by just setting the two graphs in the instance of the graph isomorphism as $G_1$ and $G_2$, as well as assigning all $w_{ij}=1$ and $c(v)=1$ for all $v\in V_1,V_2$. Obviously, if the two graphs are isomorphic, the value $\Delta_{\hat{\pi}}$ of the optimal mapping $\hat{\pi}$ will be zero. Therefore, in this case, any algorithm with multiplicative approximation guarantee must find a mapping $\pi$ with $\Delta_{\pi}=0$. Furthermore, if $G_1$ and $G_2$ are not isomorphic, then any mapping $\pi$ must induce a $\Delta_{\pi}$ strictly larger than 0. Hence, a polynomial time approximation algorithm for BI-MAP-ESTIMATE with multiplicative guarantee implies a polynomial time algorithm for the graph isomorphism problem. Note that the result can be further extended as there is no pseudo-polynomial time algorithm with multiplicative approximation guarantee unless $GI\in DTIME(n^{\mathrm{polylog}n})$.
\end{proof}
\vspace{-4mm}
\subsection{Approximation Algorithms}
\vspace{-0.7mm}
As demonstrated above, the BI-MAP-ESTIMATE problem bears high computational complexity and approximation hardness. It is thus unrealistic to pursue exact or even multiplicative approximation algorithms. To circumvent this obstacle and still find solutions with provable theoretical properties, we propose two algorithms with their respective advantages: one has an $\epsilon$-additive approximation guarantee and the other has lower time complexity and yields optimal solutions under certain conditions. The main idea behind them is to convert $\textbf{P1}$ to equivalent formulations which are more amenable to relaxation techniques.
\vspace{-2mm}
\subsubsection{Additive Approximation Algorithm}\label{sec:additiveapprox}\vspace{-1mm}
The additive approximation algorithm we propose is based on the following quadratic assignment formulation of the BI-MAP-ESTIMATE Problem which we denote as $\mathbf{P2}$.
\abovedisplayskip=3.5pt
\begin{align}
\mathbf{P2:}\ \ \quad\text{maximize } & \textstyle\sum_{i,j,k,l}q_{ijkl}x_{ik}x_{jl}\\[-2pt]
\text{\textbf{s.t. }}& \textstyle\sum_{i}x_{ij}=1,\quad \forall i\in V_1\\[-2pt]
& \textstyle\sum_{j}x_{ij}=1,\quad \forall j\in V_2\\[-2pt]
& x_{ij}\in\{0,1\}
\end{align}
The coefficients $\{q\}_{ijkl}$ of $\mathbf{P2}$ are defined as:
\belowdisplayskip=3pt
\belowdisplayskip=3pt
\[
q_{ijkl}=
\begin{cases}
w_{ij}, &\text{if }(i,j)\in E_{1}, (k,l)\in E_{2} \mbox{ and }\\[-2pt]& c(i)=c(k),c(j)=c(l),\\
-1 &\text{if $c(i)\neq c(k)$ or $c(j)\neq c(l)$,} \\
0 &\text{otherwise.}
\end{cases}
\]

The solutions to $\mathbf{P2}$ are a set of integers $\{x\}$. We will refer to the value of $\sum_{i,j,k,l}q_{ijkl}x_{ik}x_{kl}$ as the value of $\{x\}$. Based on a solution $\{x\}$, we can construct its equivalent mapping for the BI-MAP-ESTIMATE problem by setting $\pi(i)=j$ iff $x_{ij}=1$. The following proposition shows the correspondence between $\mathbf{P1}$ and $\mathbf{P2}$. 
\vspace{-1.5mm}
\begin{proposition}\label{proposition:QAP}
 	Given $G_1$, $G_2$, $c$ and $\{w\}$, the optimal solutions of $\mathbf{P1}$ and $\mathbf{P2}$ are equivalent.
\end{proposition}
\vspace{-2mm}
\begin{proof}
	We write the equivalent set of integers $\{x\}$ of a mapping $\pi$ as $\{x^\pi\}$. First, we prove that the optimal solution $\{x^*\}$ to $\mathbf{P2}$ must observe the community assignment, i.e., if $x^*_{ij}=1$, then $c(i)=c(j)$. Indeed, for a solution $\{x\}$ having some $x_{i_0i_1}=1$ but $c(i_0)\neq c(i_1)$, we can find a ``cycle of community assignment violations" starting from $i$ with $x_{i_0i_1}=x_{i_1'i_2}=x_{i_2'i_3}=\ldots x_{i_\rho'i_0'}$ and $c(i_0)=c(i_0'),c(i_1)=c(i_1'),\ldots,c(i_\rho)=c(i_\rho')$. Due to the special structure of the coefficients $\{q\}$, this cycle only contributes negative value to the objective function of $\mathbf{P2}$. Therefore, by ``reversing" the cycle, we obtain a new solution $\{x'\}$ from $\{x\}$ with $x'_{i_0i_0'}=x'_{i_1'i_1}=x'_{i_2'i_2}=\ldots=x'_{i_\rho'i_\rho}=1$ and $\sum_{i,j,k,l}q_{ijkl}x'_{ij}x'_{kl}>\sum_{i,j,k,l}q_{ijkl}x_{ij}x_{kl}$.The process of reversing cycles of community assignment violations is demonstrated in Figure 2. If follows that the optimal solution to $\mathbf{P2}$ must observe the community assignment. Then, we proceed to show that the optimal solution to $\mathbf{P1}$ is equivalent to the optimal solution to $\mathbf{P2}$. Notice that for all $\{x^\pi\}$ that observe the community assignment, we have $\sum_{ij}w_{ij}=\sum_{ijkl}q_{ijkl}x^\pi_{ik}x^\pi_{jl}+\Delta_{\pi}$. Therefore, the corresponding $\{x^{\hat\pi}\}$ of the optimal solution $\hat{\pi}$ to $\mathbf{P1}$ is also optimal for $\mathbf{P2}$ and vice versa. 
\end{proof}
\begin{figure}[!tb]
\centering
\includegraphics[width=0.92\linewidth,height=0.11\textheight]{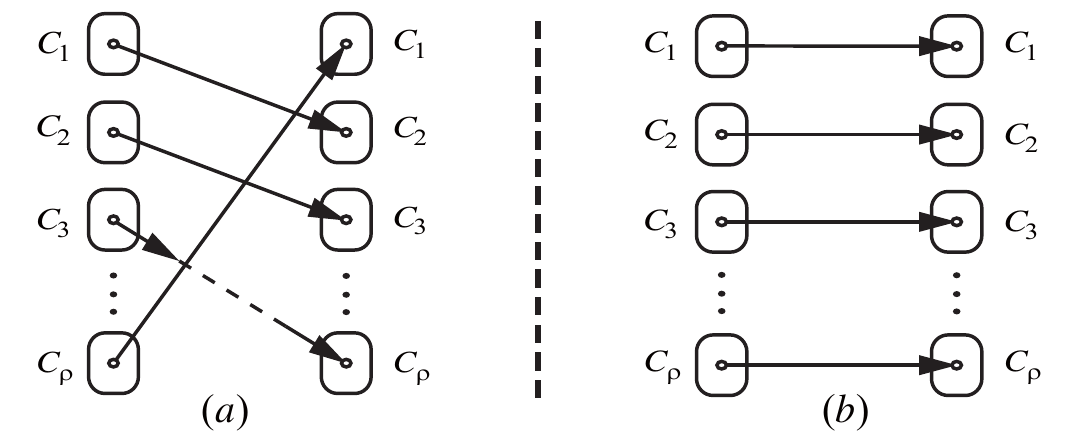}
\vspace{-2mm}
\caption{\small\bf Illustration of the reversal of a cycle of community assignment violations: (a) a cycle of community assignment violations in a mapping; (b) reversal of the cycle of violations.}
\label{fig:reverse}
\vspace{-6mm}
\end{figure}
\normalsize
\vspace{-1.5mm}

The proof of Proposition \ref{proposition:QAP} also provides the two main stages in our additive approximation algorithm: (i) Convert the instance of the BI-MAP-ESTIMATE problem into its corresponding quadratic assignment formulation $\mathbf{P2}$ where the solution is then computed. (ii) Reverse all the ``cycles of community assignment violations" in the solution and construct the desired mapping based on it.

For the first stage, we adopt the relaxing-rounding based algorithm proposed by Arora at al. \cite{cite:QAP} as a sub-procedure referred to as ``QA-Rounding" to solve the converted instances of $\mathbf{P2}$. QA-Rounding has additive approximation guarantee when the instances have coefficients $\{q\}$ that do not scale with the size of the problem \cite{cite:QAP}. Note that the requirement for the coefficients to be independent of the size of the problem is one of the key factors for the seemingly unnatural formulation of $\mathbf{P2}$. For the sake of completeness, we state in the following lemma the related result from \cite{cite:QAP}. 
\vspace{-3mm}
\begin{lemma}\label{lemma:Arora}(Theorem 3 in \cite{cite:QAP})
	Given an instance of $\mathbf{P2}$ with $-C\le q_{ijkl} \le C$ for all $i,j,k,l\in \{1\ldots n\}$ where $C$ is a constant that is independent of $n$, then for any $\epsilon>0$, QA-Rounding finds a solution $\{x\}$ with
	\abovedisplayskip=2pt
	\belowdisplayskip=2pt
	\[
	\sum_{i,j,k,l}q_{ijkl}x_{ik}x_{jl}\ge \sum_{ijkl}q_{ijkl}x^*_{ijkl}-\epsilon n^2
	\] 
	in $n^{O(\log n/\epsilon^2)}$ time, where $\{x^*\}$ is the optimal solution.
\end{lemma}
\vspace{-1.5mm}
The second stage can be completed by repeatedly traversing the solution $\{x\}$ to identify all the cycles of community assignment violations and reversing them. \textbf{Algorithm 1} illustrates a whole diagram of our proposed additive approximation algorithm.

\textbf{Approximation Guarantee: }By Lemma \ref{lemma:Arora}, QA-Rounding yields a solution whose value has a gap of less than $\epsilon n^2$ from the optimal. Combined with the equality $\sum_{i,j}w_{ij}=\Delta_{\pi}+\sum_{i,j,k,l}q_{ijkl}x_{ik}x_{jl}$ and the fact that the reversal of all the cycles of community assignment violations only incurs an increase on the value of the computed solution $\{x\}$, we have that the mapping $\pi$ given by \textbf{Algorithm 1} has an $\epsilon$-additive approximation guarantee and satisfies $c(i)=c(\pi(i))$ for all~$i$. Moreover, by Corollary \ref{corollary:accuracy}, we know that when $\epsilon,\delta$ satisfy the conditions in the corollary, the mappings yielded by \textbf{Algorithm 1} map at least $(1-\delta)n$ nodes correctly.

\vspace{-2mm}
\SetAlCapSkip{0.2em}
\begin{algorithm}[htbp]
	\begin{small}
		\begin{spacing}{0.9}
			\SetAlgoLined
			
			\label{algorithm:addaprox}		
			\KwIn{{Graphs $G_1,G_2$, weights $\{w\}$,}\\
				{$\ \ \quad\qquad$community assignment function $c$.}}
			\KwOut{\text{mapping $\pi$.}}
			\textbf{Initialize: $\pi=\emptyset$, $\forall i,j,k,l\in \{1\ldots n\}, x_{ijkl}=0$, $i',j'=0$}\\
			\text{Compute the set of coefficients $\{q\}_{ijkl}$ and}\\
			\text{form an instance $\mathcal{I}$ of $\mathbf{P2}$.} \\
			\text{$\{x\}:=$QA-Rounding($\mathcal{I}$).}\\
			\For{$i=1$ to $n$}{
				\For{$j=1$ to $n$}{
					\If{$x_{ij}=1$ and $c(i)\neq c(j)$}
					{ $x_{ij}:=0$.\\
						\While{$c(j')\neq c(i)$}{
							\text{Find $i',j'$ with $x_{i'j'}=1$ and $c(i')=c(j)$.}\\
							\text{$x_{i'j'}:=0,x_{i'j}:=1,j:=j'$.}\\
						}
						$x_{ij'}:=1.$\\
					}
				}}
				Construct $\pi$ based on $\{x\}$.\\
				Return {$\pi$}
			\end{spacing}
		\end{small}
		\caption{{The Additive Approximation Algorithm}}
	\end{algorithm}
	\vspace{-2mm}

\textbf{Time Complexity: }The QA-Rounding has a time complexity of $n^{O(\log n/\epsilon ^2)}$. The reversal of all the cycles can be completed in  $O(n^2)$ time when $\{x\}$ is represented in the form of an adjacency list-like structure. Based on those, the time complexity of \textbf{Algorithm 1} is $O(n^{O({\log n/\epsilon^2})}+n^2)$.
\vspace{-1mm}
\subsubsection{Convex Optimization-Based Heuristic}
\vspace{-1mm}
Beside the algorithm that provides additive approximation guarantee under general case, it is also useful to pursue algorithms that have stronger guarantee in special cases. In this section, we present one such algorithm that can find the optimal solution in the cases where the structural similarity between the two networks are higher than certain threshold.

The algorithm is based on convex optimization, which relies on a matrix formulation of the BI-MAP-ESTIMATE problem. The main idea is to first solve a convex-relaxed version of the matrix formulation and then convert the solution back to a legitimate one. Specifically, the matrix formulation of the BI-MAP-ESTIMATE problem, which we denote by P3, is formally stated as follows:
\belowdisplayskip=2pt
\vspace{-1mm}
\begin{align}
\mathbf{P3:}\quad\ \ \text{mininize }  \|\mathbf{W}\circ(\mathbf{A}-&\mathbf{\Pi}^\mathrm{T}\mathbf{B}\mathbf{\Pi})\|_\mathrm{F}^2+\mu\|\mathbf{\Pi m}-\mathbf{m}\|_\mathrm{F}^2\nonumber\\[-2pt]
\text{\textbf{s.t. }}  \forall i\in V_1,\ \textstyle\sum_{i}\mathbf{\Pi}_{ij}&=1\label{P3:constraint1}\\[-2pt]
\forall j\in V_2,\ \textstyle\sum_{j}\mathbf{\Pi}_{ij}&=1 \label{P3:constraint2}\\[-2pt]
\forall i,j,\ \mathbf{\Pi}_{ij}\in&\{0,1\},\label{P3:constraint3}
\end{align}
where $\mathbf{W}$ is a symmetric matrix with $\mathbf{W}_{ij}=\mathbf{W}_{ji}=\sqrt{w_{ij}}$, $\mathbf{m}$ represents the community assignment vector $(c(1),\ldots,c(n))^\mathrm{T}$, $\mu$ is a positive constant that is large enough, $\circ$ denotes the matrix Hadamard product with $(\mathbf{W\circ A})_{ij}=\mathbf{W}_{ij}\cdot \mathbf{A}_{ij}$ and $\|\cdot\|_\mathrm{F}$ represents the Frobenius norm.

Note that $\mathbf{P3}$ is equivalent to $\mathbf{P1}$ from the perspective of the relation between a mapping and its corresponding permutation matrix, as is stated in the following proposition.
\vspace{-3.5mm}
\begin{proposition}
		Given $G_1$, $G_2$, $c$ and $\{w\}$, the optimal solution of $\mathbf{P1}$ and $\mathbf{P3}$ are equivalent.
\end{proposition}
\vspace{-2.5mm}
\begin{proof}
	The proof is similar to that of Proposition \ref{proposition:QAP}. First, due to the existence of the penalty factor $\mu\|\mathbf{\Pi m}-\mathbf{m}\|_\mathrm{F}^2$, we have that the optimal solution of $\mathbf{P3}$ must observe the community assignment. Second, as for all the permutation matrices $\mathbf{\Pi}$'s and their corresponding mappings $\pi$'s that observe the community assignment, it is easy to show that $\Delta_{\pi}=\|\mathbf{W}\circ(\mathbf{A}-\mathbf{\Pi}^\mathrm{T}\mathbf{B}\mathbf{\Pi})\|_\mathrm{F}^2+\mu\|\mathbf{\Pi m}-\mathbf{m}\|_\mathrm{F}^2$ (the second term equals to 0 in this case). Hence, the optimal solution of $\mathbf{P1}$ and $\mathbf{P3}$ are equivalent.
\end{proof}
\vspace{-2mm}
Before introducing the algorithm, we further transform the objective function of $\mathbf{P3}$ into an equivalent but more tractable form. Lemma \ref{lemma:rewrite} gives the main idea of the transformation.
\vspace{-5mm}
\begin{lemma}\label{lemma:rewrite}
	Let $\mathbf{\tilde{A}}=\mathbf{W\circ A}$ and $\mathbf{\tilde{B}}=\mathbf{W\circ B}$ be the weighted adjacency matrices of $G_1$ and $G_2$ respectively, then for all permutation matrices that observe the community assignment\footnote{A permutation matrix $\mathbf{\Pi}$ observes community assignment if for all $\mathbf{\Pi}_{ij}=1$, $c(i)=c(j)$.},\normalsize the following equality holds:
	\begin{equation*}
	\|\mathbf{W\circ(A-\Pi^{\mathrm{T}}B\Pi)}\|_{\mathrm{F}}=\|\mathbf{\Pi\tilde{A}-B\tilde{\Pi}}\|_{\mathrm{F}}.\label{eq:convex}
	\end{equation*}
\end{lemma}
\vspace{-2mm}
\begin{proof}
	We prove the lemma by repeatedly using the symmetry of $\mathbf{A}$ and $\mathbf{B}$ and special properties of $\mathbf{W}$ and~$\mathbf{\Pi}$. The detailed steps are presented as follows:
	\vspace{-1.5mm}
	\begin{small}
	\begin{align}
	\|\mathbf{W}\circ(\mathbf{A}-\mathbf{\Pi}^\mathrm{T}\mathbf{B}\mathbf{\Pi})\|_\mathrm{F}&=\|\mathbf{W}\circ(\mathbf{\Pi}(\mathbf{A}-\mathbf{\Pi}^\mathrm{T}\mathbf{B}\mathbf{\Pi}))\|_\mathrm{F}\label{eq:transform1}\\[-2pt]
	&=\|\mathbf{W}\circ(\mathbf{\Pi}\mathbf{A}-\mathbf{B}\mathbf{\Pi})\|_\mathrm{F}\label{eq:transform2}\\[-2pt]
	&=\| \mathbf{W}\circ(\mathbf{\Pi A})-\mathbf{W}\circ(\mathbf{B}\mathbf{\Pi}) \|_\mathrm{F}\label{eq:transform3}\\[-2pt]
	&=\| \mathbf{\Pi}(\mathbf{W}\circ\mathbf{{A}})-(\mathbf{W}\circ\mathbf{{B}})\mathbf{\Pi} \|_\mathrm{F}\label{eq:transform4}\\[-2pt]
	&=\|(\mathbf{\Pi}\mathbf{\tilde{A}}-\mathbf{\tilde{B}}\mathbf{\Pi})\|_\mathrm{F}.\label{eq:transform5}
	\end{align}
	\end{small}
	Note that Equation (\ref{eq:transform1}) holds because multiplying by a permutation matrix does not change the value of element-wise Frobenius norm. Equations (\ref{eq:transform2}), (\ref{eq:transform3}) and (\ref{eq:transform5}) hold due to the definition of Hadamard product and $\mathbf{\tilde{A},\tilde{B}}$. The validity of Equation (\ref{eq:transform4}) is less straightforward and can be interpreted in the following way: For the weight $w_{ij}$ of a node pair $(i,j)$, it is determined only by $p_{c(i)c(j)},s_1,s_2$. Therefore, if $c(i)=c(j),c(k)=c(l)$ for some nodes $i,j,k,l$, then we have $\mathbf{W}_{ik}=\mathbf{W}_{jl}$, i.e., the weight is invariant within communities. This crucial property, combined with the fact that $\mathbf{\Pi}$ is permutation matrix that observes the community assignment, makes the Hadamard products and normal matrix multiplication in Equation (\ref{eq:transform4}) interchangeable.
\end{proof}
\vspace{-1mm}
Based on Lemma \ref{lemma:rewrite}, we can rewrite the objective function of $\mathbf{P3}$ as $\|(\mathbf{\Pi}\mathbf{\tilde{A}}-\mathbf{\tilde{B}}\mathbf{\Pi})\|^2_\mathrm{F}+\mu\|\mathbf{\Pi m-m}\|^2_\mathrm{F}$. Then, we further relax constraints (\ref{P3:constraint2}) and (\ref{P3:constraint3}) in $\mathbf{P3}$ and obtain the optimization problem $\mathbf{P3'}$ that can be formulated as:
\belowdisplayskip=6pt
\vspace{-1mm}
\begin{equation*}
\begin{aligned}
\mathbf{P3'}\quad\ \ \text{minimize }  \|(\mathbf{\Pi}\mathbf{\tilde{A}}-\mathbf{\tilde{B}}\mathbf{\Pi})\|_\mathrm{F}^2&+\mu\|\mathbf{\Pi m}-\mathbf{m}\|_\mathrm{F}^2\qquad\\[-2pt]
\text{\textbf{s.t. }}  \forall i,\ \textstyle\sum_{i\in V_1}\mathbf{\Pi}_{ij}&=1\\[-2pt]
\end{aligned}
\end{equation*}
Obviously the objective function and the set of feasible solutions are both convex. Immediately we can conclude that $\mathbf{P3'}$ is a convex-relaxed version of $\mathbf{P3}$, which is stated in the following lemma.
\vspace{-1.5mm}
\begin{lemma}
	$\mathbf{P3'}$ is a convex optimization problem.
\end{lemma}
\vspace{-1.5mm}
With all the prerequisites above, we are now ready to present our second convex optimization-based algorithm, which firstly solves for a fractional optimal solution of $\mathbf{P3'}$ and then projects that fractional solution into an integral permutation matrix (and its corresponding mapping). During the projection process, we use an $n$-dimensional array $Mapped$ to record the projected nodes and a set $Legal_i$ for each node $i$ to record the remaining legitimate nodes to which it can be mapped. The details are illustrated in \textbf{Algorithm~2}.

\textbf{Performance Guarantee: }Generally, \textbf{Algorithm 2} can not yield the optimal solution to the BI-MAP-ESTIMATE problem and the gap between its solution and the optimal one may be large. However, we will demonstrate that when the similarity between $G_1$ and $G_2$ are high enough, or equivalently, the difference between the weighted adjacency matrices $\mathbf{\tilde{A}}$ and $\mathbf{\tilde{B}}$ is sufficiently small, \textbf{Algorithm 2} is guaranteed to find the optimal mapping.
\vspace{-1.5mm}

\SetAlCapSkip{0.2em}
\begin{algorithm}[!htbp]
	\renewcommand\arraystretch{0.3}
	\begin{small}
		\begin{spacing}{0.9}
			\SetAlgoLined
			
			\label{algorithm:convexaprox}		
			\KwIn{{Graphs $G_1,G_2$, weights $\{w\}$,}\\
				{$\ \ \quad\qquad$community assignment function $c$.}}
			\KwOut{\text{mapping $\pi$.}}
			\textbf{Initialize: $Mapped[i]=0.$ $Legal_i=\emptyset$ for all $i$,}\\ \text{$\qquad\qquad\quad\pi=\emptyset$, $\mathbf{\Pi}^p,\mathbf{\Pi}^f=\bm{0}$.}\\
			\text{Compute the weight matrix $\mathbf{W}$ and}\\
			\text{form an instance $\mathcal{I}$ of $\mathbf{P3}$.} \\
			\text{Relax $\mathcal{I}$ into an instance $\mathcal{I}'$ of $\mathbf{P3}'$.}
			\text{$\mathbf{\Pi}^f:=$ the optimal (fractional) solution to ($\mathcal{I'}$).}\\
			\For{$i=1$ to $n$}{
				\text{$Legal_i:=\{k\mid Mapped[k]=0\mbox{ and }c(k)=c(i)\}$}
				\text{$j:=\arg\max_{k\in Legal_i}\mathbf{\Pi}^f_{ik}$}.\\
				\text{$\mathbf{\Pi}^p_{ij}:=1$. $Mapped[j]$:=1.}\\
			}
			Construct $\pi$ based on $\mathbf{\Pi}^p$.\\
			Return {$\pi$}
		\end{spacing}
	\end{small}
	\caption{Convex Optimization-Based Algorithm}
\end{algorithm}
\vspace{-5mm}

\begin{theorem}\label{theorem:convexoptimal}
	Let $\mathbf{\tilde{B}'}$ be a symmetric matrix that is related with $\mathbf{\tilde{A}}$ by a unique $\mathbf{\hat{\Pi}}$ that observes the community assignment, i.e., $\mathbf{\tilde{B}'=\hat{\Pi}A\hat{\Pi}^{\mathrm{T}}}$. Denote $\mathbf{\tilde{B}'=U\Lambda U^\mathrm{T}}$ as its unitary eigen-decomposition with $\epsilon_2\le\sum_j|\mathbf{U}_{ij}|\le\epsilon_1$ for all $i$. Define $\lambda_1,\lambda_2,\ldots,\lambda_n$ as the eigenvalues of $\mathbf{\tilde{B}'}$ with $\sigma=\max_i|\lambda_i|$ and $\delta\le |\lambda_i-\lambda_j|$ for all $i,j$.  Assume that there exists a matrix $\mathbf{R}$ that satisfies $\mathbf{\tilde{B}=\tilde{B}'+R}$. We denote $\mathbf{E=URU^\mathrm{T}}$ with $\|\mathbf{E}\|_\mathrm{F}=\xi$ and $\mathbf{M=m^\mathrm{T}m}$ with $\|\mathbf{M}\|_\mathrm{F}=M$. Let $\mathbf{\Pi}^p$ be the solution obtained by \textbf{Algorithm 2} and $\mathbf{\Pi^*}$ be the optimal solution. If
	\vspace{-1mm}
	\belowdisplayskip=0pt
	\abovedisplayskip=0pt
	\begin{small}
		\begin{equation*}
		(\sigma^2+1)\xi^2+\mu^2M^2\le \left[\frac{\delta^2}{(2\sqrt{n}+1)(1+\sqrt{n}\epsilon_1/\epsilon_2)(1+2\epsilon_1/\epsilon_2)}\right]^2,
		\vspace{-2mm}
	\end{equation*}
		\end{small}
		then $\mathbf{\Pi^p=\Pi^*}$.
\end{theorem}
\vspace{-3.2mm}
\begin{proof}
	The proof is divided into three steps: (i) First, similar to the argument in \cite{cite:convex}, by constructing the Lagrangian function of $\mathbf{P3'}$ and setting its gradient to 0, we obtain the necessary conditions that the optimal fractional solution $\mathbf{\Pi}^f$ to $\mathbf{P3'}$ must satisfy; (ii) Then, combining these with the conditions stated in the theorem and the projection from $\mathbf{\Pi}^f$ to $\mathbf{\Pi}^p$, we show that $\mathbf{\Pi}^p=\mathbf{\hat{\Pi}}$; (iii) Finally, we prove that in this case $\mathbf{\hat{\Pi}=\mathbf{\Pi}^*}$, which concludes the~proof.
	
	
		1. Derivation of the Necessary Conditions: We start the first step with rewriting $\mathbf{P3'}$ as an optimization problem with respect to $\mathbf{Q=\Pi\hat{\Pi}^\mathrm{T}}$. Since
		\vspace{-0.7mm}
		\[
		\mathbf{\Pi \tilde{A}}-\mathbf{\tilde{B}\Pi}=(\mathbf{\Pi \hat{\Pi}^{\mathrm{T}}\tilde{B}'}-\mathbf{\tilde{B}\Pi\hat{\Pi}^{\mathrm{T}}})\mathbf{\hat{\Pi}}=(\mathbf{Q\tilde{B}'}-\mathbf{\tilde{B}Q})\mathbf{\hat{\Pi}},
		\]
		and 
		$
		\mathbf{\Pi m}-\mathbf{m}=(\mathbf{Qm}-\mathbf{m})\mathbf{\hat{\Pi}},
		$
		we can reformulate the objective function of $\mathbf{P3}'$ with $\mathbf{Q}$ as variable and divide it by two for ease of further manipulation as
		$
		\frac{1}{2}\|\mathbf{Q\tilde{B}}-\mathbf{BQ}\|^2_{\mathrm{F}}+\frac{\mu}{2}\|\mathbf{Q m}-\mathbf{m}\|_\mathrm{F}^2.
		$
		The constraint $\sum_{j}\mathbf{\Pi}_{ij}$ for all $i$ can be expressed as $\mathbf{Q1=1}$. The solution of the reformulated version can be associated with the original one by $\mathbf{\Pi=Q\hat{\Pi}}$. Next, by introducing multiplier $\bm{\alpha}$ for the equality constraint of $\mathbf{P3}'$, we construct its Lagrangian function as
		\belowdisplayshortskip=6pt
		\belowdisplayskip=6pt
		\vspace{-2mm}
		\begin{equation*}
		L(\mathbf{Q},\mathbf{\bm{\alpha}})=\frac{1}{2}\|\mathbf{Q\tilde{B}}-\mathbf{BQ}\|^2_{\mathrm{F}}+\frac{\mu}{2}\|\mathbf{Q m}-\mathbf{m}\|_\mathrm{F}^2+\mathrm{tr}(\mathbf{Q1}-\mathbf{1})\mathbf{\bm{\alpha}}^\mathrm{T}.	
			\vspace{-1.5mm}
		\end{equation*}
	
		The key element of the proof of the lemma is the sufficient conditions for $\mathbf{Q}$ to be the optimal (fractional) solution to $\mathbf{P3'}$.
		To yield the sufficient conditions, we take the gradient of $L(\mathbf{Q},\mathbf{\bm{\alpha}})$ with respect to $\mathbf{Q}$ and set it as 0. Then we have
		\belowdisplayshortskip=2pt
		\vspace{-1mm}
		\begin{small}
		\begin{equation*}
		\bigtriangledown_{\mathbf{Q}}L(\mathbf{Q},\mathbf{\bm{\alpha}})=\mathbf{QB}^2+\mathbf{\tilde{B}}^2\mathbf{Q}-2\mathbf{\tilde{B}}\mathbf{QB}+\mathbf{\bm{\alpha} 1}^\mathrm{T}+\mu(\mathbf{QM}-\mathbf{M})
		=\mathbf{0}.
		\vspace{-1.5mm}
		\end{equation*}
		
		\end{small}
		Multiplying $\mathbf{U}^\mathrm{T}$ to the left side of $\bigtriangledown_{\mathbf{Q}}L(\mathbf{Q},\mathbf{\bm{\alpha}})$ and $\mathbf{U}$ to the right side we get
		\vspace{-1mm}
		\begin{equation*}
		\begin{aligned}
		(\mathbf{F\Lambda}^2+\mathbf{\Lambda^2F}-2\mathbf{\Lambda F\Lambda})&+(\mathbf{FE\Lambda}+\mathbf{F\Lambda E}-2\mathbf{\Lambda FE})\\[-3pt]&+\mathbf{\bm{\gamma} v}^{\mathrm{T}}+\mathbf{FG}+\mu\mathbf{FM'}-\mu\mathbf{M'}=\mathbf{0},
		\end{aligned}
		\vspace{-1.5mm}
		\end{equation*}
		where $\mathbf{F}=\mathbf{U}^\mathrm{T}\mathbf{Q}\mathbf{U}$, $\mathbf{v}=\mathbf{U}^\mathrm{T}\mathbf{1}$, $\mathbf{\bm{\gamma}=U^\mathrm{T}\bm{\alpha}}$, $\mathbf{G=E}^2$ and $\mathbf{M'=U^\mathrm{T}MU}$. 
		
		Rewriting the equation coordinate-wise, we have
		\vspace{-1mm}
			\begin{equation*}
		\begin{aligned}
		\mathbf{F}_{ij}(&\lambda_i-\lambda_j)^2+\mathbf{v}_j\bm{\gamma}_i-\mu \mathbf{M}'_{ij}\nonumber\\[-2pt]&+\textstyle\sum_{k}\mathbf{F}_{ik}(\mathbf{E}_{kj}(\lambda_j+\lambda_k-2\lambda_i)+\mathbf{G}_{kj}+\mu\mathbf{M'}_{kj})=0\label{eq:element}
		\end{aligned}
		\vspace{-1.5mm}
		\end{equation*}
		Substituting $i=j$ into the above equation and plugging the results back to eliminate variables $\bm{\gamma}_i$'s, it follows that
		\vspace{-1mm}
%
\begin{small}
		\begin{equation*}
		\begin{aligned}
		\mathbf{F}_{ij}\mathbf{v}_i(\lambda_i-&\lambda_j)^2+\textstyle\sum_k\mathbf{F}_{ik}(\mathbf{v}_i\mathbf{G}_{kj}-\mathbf{v}_j\mathbf{G}_{ki}+\mu \mathbf{v}_i\mathbf{M'}_{kj}-\mu \mathbf{v}_j\mathbf{M'}_{ki})\\[-2pt]&+\textstyle\sum_k\mathbf{F}_{ik}(\mathbf{v}_i\,\mathbf{E}_{kj}(\lambda_j+\lambda_k-2\lambda_i)-\mathbf{v}_j\mathbf{E}_{kj}(\lambda_k-\lambda_i))\\[-2pt]&+\mu (\mathbf{v}_j\mathbf{M'}_{ii}-\mathbf{v}_i\mathbf{M'}_{ij})=0
		\end{aligned}
		\end{equation*}
		\vspace{-2mm}
		\end{small}
		
		We further define the following variables
		\vspace{-1mm}
		\begin{small}
		\begin{equation*}
		\begin{aligned}
		r_{ij}&=\frac{\mu}{(\lambda_i-\lambda_j)^2}(\mathbf{v}_j\mathbf{M'}_{ii}-\mathbf{v}_i\mathbf{M'}_{ij})\\[-2pt]
		s_{jk}^i &=\frac{1}{(\lambda_i-\lambda_j)^2}\left(\mathbf{E}_{kj}(\lambda_j+\lambda_k-2\lambda_i)-\frac{\mathbf{v}_j}{\mathbf{v}_i}\mathbf{E}_{ki}(\lambda_k-\lambda_i)\right)\\[-2pt]
		t_{jk}^i&=\frac{1}{(\lambda_i-\lambda_j)^2}\left(\mathbf{G}_{kj}-\frac{\mathbf{v}_j}{\mathbf{v}_i}\mathbf{G}_{ki}\right)\\[-2pt]
		w_{jk}^{i}&=\frac{\mu}{(\lambda_i-\lambda_{j})^2} \left(\mathbf{M'}_{kj}-\frac{\mathbf{v}_j}{\mathbf{v}_i}\mathbf{M'}_{ki}\right),
		\end{aligned}
		\vspace{-1.5mm}
		\end{equation*}
		\end{small}
		for $i\neq j$. And $s_{ik}^j=t_{ik}^j=w_{jk}^i=r_{ij}=0$ for $i=j$.
		Then, we arrive at the following linear system
\vspace{-1mm}
		\begin{align}
		\mathbf{F}_{ij}+\textstyle\sum_{k}\mathbf{F}_{ik}&(s_{jk}^i+t_{jk}^i+w_{jk}^i+\frac{r_{ij}}{n})=0,\quad i\neq j \label{eq:condition1}\\[-2pt]
		\textstyle\sum_{k}\mathbf{F}_{ik}\mathbf{v}_k&=\mathbf{v}_i\label{eq:condition2},
		\end{align}
		where the second set of equations come from the constraint $\mathbf{Q1}=\mathbf{1}$. Equations (\ref{eq:condition1}) and (\ref{eq:condition2}) represent conditions that the optimal solution $\mathbf{Q}$ (or equivalently $\mathbf{F}$) needs to satisfy. 
		
		2. The Equivalence of $\mathbf{\Pi}^p$ and $\mathbf{\hat{\Pi}}$: Based on the conditions above, we move on to the second step. Recall that in this step our goal is to prove that $\mathbf{\Pi}^p$, which is a projection of the optimal fractional solution $\mathbf{\Pi}^f$, equals to $\mathbf{\hat{\Pi}}$.
		We formalize this notion in Lemma \ref{lemma:bounddifference}, the proof of which carries on the main idea of the second~step.
		\vspace{-1.5mm}
		\begin{lemma}\label{lemma:bounddifference}
			Let $\mathbf{\Pi}^p$ be the solution computed by \textbf{Algorithm 2} and $\mathbf{\hat{\Pi}}$ be defined as in Theorem \ref{theorem:convexoptimal}. Under the conditions stated in the theorem, $\mathbf{\Pi}^p=\mathbf{\hat{\Pi}}$.
		\end{lemma}
		\vspace{-3mm}
		\begin{proof}
		As the optimal fractional solution $\mathbf{\Pi}^f=\mathbf{Q\hat{\Pi}}$, we first show that $\mathbf{Q}$ (or $\mathbf{F}$) is sufficiently close to the identity matrix $\mathbf{I}$, from which using the property of the projection process we obtain that $\mathbf{\Pi}^p$ is identical to $\mathbf{\hat{\Pi}}$. We achieve this by treating linear system consisting of Equations (\ref{eq:condition1}) and (\ref{eq:condition2}) as a perturbed version of
		\vspace{-1.5mm}
\begin{align}
\mathbf{F}_{ij}&=0,\quad i\neq j\\[-3pt]
\textstyle\sum_{k}\mathbf{F}_{ik}\mathbf{v}_k&=\mathbf{v}_i,
\end{align}
	the solution of which is clearly $\mathbf{I}$. Then using the results from stability of perturbed linear system \cite{cite:perturb} that is presented in Lemma \ref{lemma:perturb} below and the conditions in Theorem \ref{theorem:convexoptimal}, we can bound the difference between $\mathbf{F}$ and $\mathbf{I}$.
\vspace{-2mm}
\begin{lemma} (Theorem 1 in \cite{cite:perturb})\label{lemma:perturb}
	\ Let $\|\cdot\|$ be any $p$-norm. For two linear systems $\mathbf{Dx=b}$ and $\tilde{\mathbf{D}}\mathbf{x}=\tilde{\mathbf{b}}$, let $\mathbf{x}_0$ and $\mathbf{x}$ be their solutions, if $\|\mathbf{D-\tilde{D}}\|\|\mathbf{D}^{-1}\|<1$, then we have
	\[
	\frac{\|\mathbf{x}-\mathbf{x_0}\|}{\|\mathbf{x}_0\|}\le \frac{\|\mathbf{D}\|\|\mathbf{D}^{-1}\|}{1-\|\mathbf{D}-\tilde{\mathbf{D}}\|\|\mathbf{D}^{-1}\|}\left\lbrace\frac{\|\mathbf{D}-\tilde{\mathbf{D}}\|}{\|\mathbf{D}\|}+\frac{\|\mathbf{b}-\tilde{\mathbf{b}}\|}{\mathbf{b}}\right\rbrace.
	\]
\end{lemma}
		Denoting by $\mathbf{f}=(\mathbf{F}_{11},\ldots,\mathbf{F}_{1n},\ldots,\mathbf{F}_{n1},\ldots,\mathbf{F}_{nn})^\mathrm{T}$ the row stack vector representation of $\mathbf{F}$, we can rewrite the perturbed system as 
		$\mathbf{(D+N)f=b},\label{linear:perturb}$
		and the original unperturbed system as 
	$	\mathbf{Df=b},\label{linear:unperturb}$
		with $\mathbf{D}=\mathrm{diag}\{\mathbf{D_1},\ldots,\mathbf{D}_n\}$ being an $n^2\times n^2$ block-diagonal matrix, where each $\mathbf{D}_i$ is an $n\times n$ block consisting of identity matrix with the $i$th row replaced by vector $\mathbf{v}^\mathrm{T}$. $\mathbf{N}$ is also an $n^2\times n^2$ block-diagonal matrix with the $n\times n$ blocks $\mathbf{N}_i$ being a matrix with elements $(\mathbf{N}_i)_{jk}=s_{jk}^i+t_{jk}^i+w_{jk}^i+{r_{ij}}/{n}$. And $\mathbf{b}$ is an $n^2\times 1$ vector with the $[(i-1)(n+1)+1]$-st element as $v_i$ and other element as $0$. 
		Using Lemma \ref{lemma:perturb} on the perturbed system and the unperturbed one with $\|\cdot\|$ taken as 2-norm (Euclidean norm), we obtain that
		\vspace{-2mm}
			\begin{align}
			\|\mathbf{f-f}_0\|\le 
			\|\mathbf{f}_0\|\frac{\|\mathbf{D}^{-1}\|\|\mathbf{N}\|}{1-\|\mathbf{D}^{-1}\|\|\mathbf{N}\|}, \label{eq:matrixbound}
			\end{align}
		where $\mathbf{f}_0$ is the row stack vector representation of $\mathbf{I}$. Therefore, to derive the upper bound for the difference between $\mathbf{F}$ and $\mathbf{I}$, we need to further upper bound the RHS of Inequality (\ref{eq:matrixbound}). The technique we use here is harnessing the special structure of $\mathbf{D}$ and $\mathbf{N}$ so that we can derive bounds for $\|\mathbf{D}^{-1}\|$ and $\|\mathbf{N}\|$, which are represented in functions of variables $\{s\},\{t\},\{w\}$ and $\{r\}$. By further associating the variables with the spectral parameters $\delta,\epsilon_1,\epsilon_2,$ etc. defined in the theorem, we yield an upper bound for the RHS of Inequality (\ref{eq:matrixbound}) that depends on those spectral parameters. Due to the space limitations, we defer the detailed derivation of the upper bound to \textbf{Appendix \ref{app:upperbound}}.

		Based on the upper bound, we have that if the conditions in the theorem are satisfied, then 
		\begin{align*}
		\|\mathbf{F-I}\|_\mathrm{F}=\|\mathbf{f-f_0}\|\le \frac{1}{2}.
		\end{align*}
		Since $\|\mathbf{\Pi}^f-\mathbf{\hat{\Pi}}\|_{\mathrm{F}}=\|\mathbf{Q\hat{\Pi}-\hat{\Pi}}\|_{\mathrm{F}}=\|\mathbf{(Q-I)\hat{\Pi}}\|_{\mathrm{F}}=\mathbf{\|Q-I\|_{\mathrm{F}}}=\mathbf{\|F-I\|_{\mathrm{F}}}\le 1/2$, the entry-wise difference between $\mathbf{\Pi}^f$ and $\mathbf{\hat{\Pi}}$ is less than $1/2$. Thus, the projection process in \textbf{Algorithm 2} is bound to project $\mathbf{\Pi}_f$ as $\mathbf{\hat{\Pi}}$, which concludes the second step, i.e., the proof of Lemma \ref{lemma:bounddifference}.
		\end{proof}
		\vspace{-2mm}
		
		3. Optimality of $\mathbf{\hat{\Pi}}$: Now we proceed to the final step and prove that $\mathbf{\hat{\Pi}=\Pi}^*$ by contradiction. If there exists some permutation matrix $\mathbf{\Pi}'\neq \mathbf{\hat{\Pi}}$ with $\|\mathbf{\Pi' \tilde{A}}-\mathbf{\Pi'\tilde{B}}\|_\mathrm{F}<\|\mathbf{\hat{\Pi} \tilde{A}}-\mathbf{\hat{\Pi}\tilde{B}}\|_\mathrm{F}$. Then, we consider $\mathbf{\tilde{B}=B_0+R'}$ with $\mathbf{B}_0=\mathbf{\Pi'}^{\mathrm{T}}\mathbf{A\Pi'}$. Obviously, $\mathbf{R'}$ satisfies the conditions in Theorem \ref{theorem:convexoptimal}. 
		Hence, by Lemma \ref{lemma:bounddifference}, we should have that the the solution $\mathbf{\Pi}^p$ computed by \textbf{Algorithm 2} equals to $\mathbf{\Pi}'$. However, we also have $\mathbf{\Pi}^p=\mathbf{\hat{\Pi}}$, which leads to a contradiction. Thus, $\mathbf{\hat{\Pi}}$ is the optimal solution to $\mathbf{P3}$, which finishes the proof of the theorem. \qed
\end{proof}
\vspace{-1.5mm}
\textbf{Time Complexity: }In the first stage of \textbf{Algorithm 2}, we use the primal interior point algorithm proposed in \cite{cite:polysolution} to solve the instance of $\mathbf{P3'}$, which has a time complexity of $O(N^3)=O(n^6)$ where $N=n^2$ is the number of variables in the instance. The projection process of the second stage can be implemented in $O(n^2)$ time. Thus, the total time complexity of \textbf{Algorithm 2} is $O(n^6)$. Note that the result is only the worst case guarantee and the average time complexity of \textbf{Algorithm 2} is much lower \cite{cite:polysolution}.
\vspace{-2mm}
\section{De-anonymization with Unilateral Community Information}\label{sec:unilateral}
In this section, we investigate the de-anonymization problem with unilateral community information, i.e., when the adversary only possesses the community assignment function of the published network $G_1$. Following the path of the bilateral de-anonymization in Sections \ref{sec:bilateral1} and \ref{sec:bilateral2}, we will give the corresponding results we obtain for the unilateral case. Through comparisons of these results and illustration in our later experiments, we demonstrate that de-anonymization with only unilateral community information is harder than that with bilateral community information, which shows the importance of community assignment as side information.

\subsection{MAP-based Cost Function}
We first derive our cost function in the unilateral case. 
Again, according to the definition of MAP estimation, given the published network $G_1$, auxiliary network $G_2$, parameters $\bm{\theta}$ and the community assignment function $c$ of $G_1$, the MAP estimate $\hat{\pi}$ of the correct mapping $\pi_0$ is defined as:
\abovedisplayskip=2pt
\abovedisplayshortskip=2pt
\belowdisplayshortskip=2pt
\belowdisplayskip=2pt
\begin{align}
	\hat{\pi}= \arg\max_{\pi\in \Pi}Pr(\pi_0=\pi\mid G_1,G_2,c,\bm{\theta}),\label{eq:mapestimator1}
\end{align}
where $\Pi$ denotes the set of all bijective mappings from $V_1$ to $V_2$. Note that in the unilateral case we have no prior knowledge of the community assignment of $G_2$. Consequently, we can not restrict $\Pi$ to the set of mappings that observe the community assignment. 

Due to the space limit, we omit the processing of the MAP estimator (\ref{eq:mapestimator1}) and present the detailed steps in \textbf{Appendix \ref{app:MAP}}. After a sequence of manipulations, we arrive at the following equation for calculation of the MAP estimate.
\vspace{-0.5mm}
\begin{equation*}
\begin{aligned}
\hat{\pi}
=&\arg\min_{\pi\in\Pi}\left\lbrace \sum_{i<j}^nw_{ij}(\mathbbm{1}\{(i,j)\notin E_1,(\pi(i),\pi(j)\in E_2)\})\right\rbrace\\[-2pt]
\triangleq&\arg\min_{\pi\in\Pi}\Delta_{\pi},
\end{aligned}
\vspace{-1mm}
\end{equation*}
where $w_{ij}=\log\left(\frac{1-p_{c(i)c(j)}(s_1+s_2-s_1s_2)}{p_{c(i)c(j)}(1-s_1)(1-s_2)}\right).$ Note that different from the bilateral case, the cost function in the unilateral case is equivalent to a single-sided weighted edge disagreement induced by a mapping. This subtle difference has crucial implications to our analysis on the algorithmic aspect of unilateral de-anonymization.


\subsection{Validity of the Cost Function}
Following the same thread of thought, we proceed to justify the MAP estimation used in unilateral de-anonymization. Using similar proof technique, we derive the same result for the cost function in unilateral case as in bilateral one.
\vspace{-1mm}
\begin{theorem}\label{theorem:MAP1}
	Let $\alpha=\min_{ab}p_{ab}, \beta=\max_{ab}p_{ab}$, $\overline{w}=\max_{ij}w_{ij}$ and $\underline{w}=\min_{ij}w_{ij}$. Assume that $\alpha,\beta\rightarrow 0$, $s_1,s_2$ do not go to 1 as $n\rightarrow \infty$ and $\frac{\log \alpha}{\log \beta}\le \gamma$. Furthermore, suppose that
	\begin{align*}
	\frac{\alpha(1-\beta)^2s_1^2s_2^2\log(1/\alpha)}{s_1+s_2}=\Omega\left({\frac{\gamma\log^2 n}{n}}\right)+\omega\left(\frac{1}{n}\right),
	\end{align*}
	then the MAP estimate $\hat{\pi}$ in the unilateral case  almost surely equals to the correct mapping $\pi_0$ as $n\rightarrow \infty$.
\end{theorem}
\vspace{-2mm}
\begin{proof}
	The proof is basically identical to the proof of Theorem \ref{theorem:MAP}. The only difference here is that we redefine $X_{ij}$ as a Bernoulli random variable with mean $p_{ij}s_1(1-p_{\pi(i)\pi(j)}s_2)$ and $Y_{ij}$ as a Bernoulli random variable with mean $p_{ij}s_1(1-s_2)$. Then, by using the same bounding technique for $Pr\{X_{\pi}-Y_{\pi}\le 0\}$, we conclude the same result for the cost function in unilateral case. 
\end{proof}
\vspace{-2mm}
Theorems \ref{theorem:MAP} and \ref{theorem:MAP1} show that the cost function based on MAP estimation is equally effective in de-anonymization with bilateral and unilateral community information. However, as we will show in the sequel, the feasibility of the cost function in unilateral case is weaker than in bilateral case.

	\arrayrulewidth=0.52pt
	\belowrulesep=0pt
	\aboverulesep=0pt
	\begin{table*}[!tb]
		\centering
		
		\begin{tabular}{c|c|c|c|c|c}
			\toprule[0.1em]
			\textbf{Dataset}&\textbf{Degree Distribution}&\textbf{Source} &\textbf{\# of Nodes} &\textbf{\# of Edges}&\textbf{\# of Communities}  \\
			\midrule\midrule
			\multirow{3}{*}{Synthetic Networks}& power law & synthetic & 500-2000 & $\approx$500-100000 & 10-40\\
			\cline{2-6} & Poisson & synthetic & 500-2000 & $\approx$500-100000 & 10-40 \\
			\cline{2-6} & exponential & synthetic & 500-2000 & $\approx$500-100000 & 10-40 \\
			\cmidrule{1-6}\morecmidrules\cmidrule[0.04em]{1-6}
			\multicolumn{2}{c|}{Sampled Social Networks}& SNAP\cite{cite:SNAP} & 500-2000 & $\approx$1000-40000 & 10-40\\
			\cmidrule{1-6}\morecmidrules\cmidrule[0.04em]{1-6}
			\multicolumn{2}{c|}{Co-authorship Networks}& MAG\cite{cite:MAG} & $\approx2000$ & $\approx8000$ & $\approx60$\\
			\bottomrule[0.1em]
		\end{tabular}
		\vspace{-3mm}
		\caption{\bf Summary of datasets in experiments}
		\label{table:summary}
		\vspace{-5mm}
	\end{table*}
\subsection{Algorithmic Aspect}
In this section, we investigate the algorithmic aspect of de-anonymization with unilateral community information and propose corresponding algorithms as in the bilateral case.

\subsubsection{The Unilateral MAP-ESTIMATE Problem}
We first formally introduce the combinatorial optimization problem induced by minimizing the cost function in unilateral de-anonymization.
\vspace{-1mm}
\begin{definition}{(\bf{The UNI-MAP-ESTIMATE Problem})}\ \label{def:BI-MAP-ESTIMATE1}
	Given two graphs $G_1(V,E_1,\mathbf{A})$ and $G_2(V,E_2,\mathbf{B})$, community assignment function $c$ of $G_1$ and weights $\{w\}$, the goal is to compute a mapping $\hat{\pi}:V_1\mapsto V_2$ that satisfies
	\begin{small}
		\vspace{-1mm}
		\begin{align}
		\hat{\pi}&=\arg\min_{\pi\in\Pi}\left\lbrace \sum_{i<j}^nw_{ij}(\mathbbm{1}\{(i,j)\notin E_1,(\pi(i),\pi(j))\in E_2)\}\right\rbrace\nonumber\\[-2pt]
		&\triangleq\arg\min_{\pi\in\Pi}\Delta_{\pi}, \nonumber
		\end{align} 
	\end{small}
	where $\Pi=\{\pi:V_1\mapsto V_2\}$. 
\end{definition}
\vspace{-1mm}
Similar to the bilateral de-anonymization, we require the weights $\{w\}$ to be induced by well-defined community affinity values $\{p\}$, $s_1$ and $s_2$, though the latter ones are not explicitly given. 
Due to the asymmetry of $\Delta_{\pi}$ in unilateral de-anonymization, intuitively, the UNI-MAP-ESTIMATE problem may bear higher approximation hardness than the BI-MAP-ESTIMATE problem in bilateral de-anonymization. The proposition we present below consolidates this intuition.
\vspace{-4mm}
\begin{proposition}\label{proposition:hardness1}
	UNI-MAP-ESTIMATE problem is NP-hard. Moreover, there is no polynomial time (pseudo polynomial time) approximation algorithm for UNI-MAP-ESTIMATE with any multiplicative approximation guarantee unless $P=NP$ ($NP\in DTIME(n^{\mathrm{polylog} n})$).
\end{proposition}
\vspace{-2.5mm}
\begin{proof}
		The proof is done by reduction from $k$-CLIQUE problem. Given a graph $G(V,E)$, the $k$-CLIQUE problem asks whether there exists a clique of size no smaller than $k$ in $G$. The main idea of the reduction is that: Given an instance of $k$-CLIQUE with $G(V,E)$ and $k$, we set $G_1$ as $G$ and $G_2$ as a graph consisting of a  clique of size $k$ and $(|V|-k)$ additional nodes. Setting $w_{ij}=1$ and $c(v)=1$ for all $v$ in $G_1$, we have an instance of UNI-MAP-ESTIMATE. Obviously, if the $G$ contains a clique of size no less than $k$, the value $\Delta_{\hat{\pi}}$ of the optimal mapping $\hat{\pi}$ in UNI-MAP-ESTIMATE will be zero. Therefore, in this case, any algorithm with multiplicative approximation guarantee must find a mapping $\pi$ with $\Delta_{\pi}=0$. Furthermore, if $G$ does not contain a clique of size no smaller than $k$, then any mapping $\pi$ must satisfy $\Delta_{\pi}>0$. Hence, a polynomial (pseudo-polynomial) time approximation algorithm for BI-MAP-ESTIMATE with multiplicative guarantee implies a polynomial (pseudo-polynomial) time algorithm for $k$-CLIQUE. Since $k$-CLIQUE problem is NP-Complete, we justify the approximation hardness of UNI-MAP-ESTIMATE as stated in the proposition. 
\end{proof}
\vspace{-2mm}
Note that the graph isomorphism problem is at least as hard as the problems in $P$, which implies that the approximation hardness result for UNI-MAP-ESTIMATE is stronger than that for BI-MAP-ESTIMATE. 
\vspace{-1mm}
\subsubsection{Additive Approximation Algorithm}
We design a similar approximation algorithm with an $\epsilon$-additive approximation guarantee as in the bilateral case, by formulating the UNI-MAP-ESTIMATE problem in quadratic assignment fashion as follows
\begin{align}
\quad\text{minimize } & \textstyle\sum_{i,j,k,l}q_{ijkl}x_{ik}x_{jl}\\[-2pt]
\text{\textbf{s.t. }}& \textstyle\sum_{i}x_{ij}=1,\quad \forall i\in V_1\\[-2pt]
& \textstyle\sum_{j}x_{ij}=1,\quad \forall j\in V_2\\[-2pt]
& x_{ij}\in\{0,1\}
\end{align}
with the coefficients $\{q\}_{ijkl}$ of the formulation defined as:
\[
q_{ijkl}=
\begin{cases}
w_{ij}, &\text{if }(i,j)\not\in E_{1}, (k,l)\in E_{2}\\
0 &\text{otherwise.}
\end{cases}
\]
Note that due to the absence of community assignment constraints, we can directly formulate the problem as a minimization one and omit the penalty factor as in bilateral de-anonymization. By invoking the same QA-Rounding procedure on the formulated instance and convert the resulting solution $\{x\}$ to its equivalent mapping $\pi$. Using similar analysis technique as in Section \ref{sec:additiveapprox}, we have that the algorithm obtains solutions that have a gap of at most $\epsilon n^2$ to the optimal ones in time $O(n^{O(\log n/\epsilon^2)}+n^2)$.

\subsubsection{Convex Optimization Based Heuristic}
We now proceed to present the heursitc based on convex optimization for the UNI-MAP-ESTIMATE problem, which relies on the following matrix formulation.
	\begin{align}
\text{mininize }  \|\mathbf{W}\circ(\mathbf{\Pi}\mathbf{A}-&\mathbf{B}\mathbf{\Pi})\|_\mathrm{\lfloor F\rfloor}^2\nonumber\\[-2pt]
	\text{\textbf{s.t. }}  \forall i\in V_1,\ \textstyle\sum_{i}\mathbf{\Pi}_{ij}&=1\label{P3:constraint1u}\\[-2pt]
	\forall j\in V_2,\ \textstyle\sum_{j}\mathbf{\Pi}_{ij}&=1 \label{P3:constraint2u}\\[-2pt]
	\forall i,j,\ \mathbf{\Pi}_{ij}\in&\{0,1\},\label{P3:constraint3u}
	\end{align}
where $\mathbf{W}$ and $\circ$ share the same definitions as those in $\mathbf{P3}$ and $\|\cdot\|_\mathrm{\lfloor F\rfloor}$ is defined to be a variant of Frobenius norm. Specifically, $\|\mathbf{M}\|_{\mathrm{\lfloor F\rfloor}}=\sqrt{\sum_{i=1}^n\sum_{j=1}^n(\mathbbm{1}\{\mathbf{M}_{ij}\le0\}\mathbf{M}_{ij}^2)}$ for a matrix $\mathbf{M}$, where only negative elements contribute to the value of the norm\footnote{It is easy to verify that operator $\|\cdot\|_{\mathrm{\lfloor F\rfloor}}$ satisfies the definition of norm.}. By relaxing the integral constraint (\ref{P3:constraint3u}), we again arrive at an optimization problem, which is shown to be convex in \textbf{Appendix \ref{app:convex}}. Our second algorithm for unilateral de-anonymization is to first solve the relaxed version of the matrix formulation of UNI-MAP-ESTIMATE and then project the fractional solution to an integral one. Unfortunately, due to the asymmetry of the operator $\|\cdot\|_{\mathrm{\lfloor F\rfloor}}$, it is difficult to derive closed form expression for the gradient of the Lagrangian function of UNI-MAP-ESTIMATE. Thus, we cannot prove conditional optimality of the algorithm as we did in BI-MAP-ESTIMATE. 

We provide a summary of the differences existing in bilateral and unilateral de-anonymizations from a higher level in \textbf{Appendix \ref{sec:difference}}).

\begin{figure*}[htbp]
	\centering
	\subfigure[]{
		\begin{minipage}[]{0.235\linewidth}
			\centering
			\vspace{-4mm}
			\includegraphics[width=1.02\linewidth]{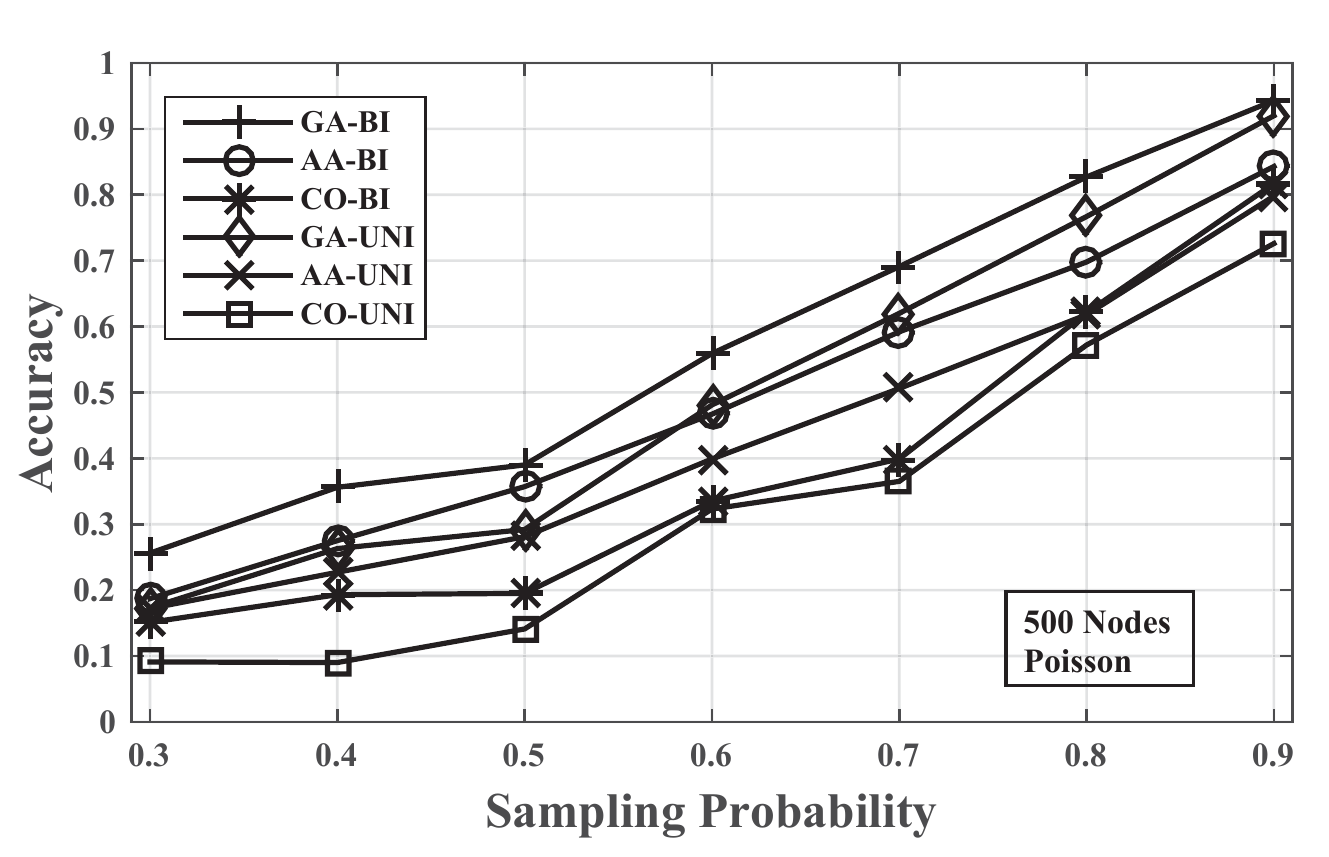}
			\vspace{-4mm}
			\label{fig:500poi}
		\end{minipage}%
		
	}
	\subfigure[]{
		\begin{minipage}[]{0.235\linewidth}
			\centering
			\vspace{-4mm}
			\includegraphics[width=1.02\linewidth]{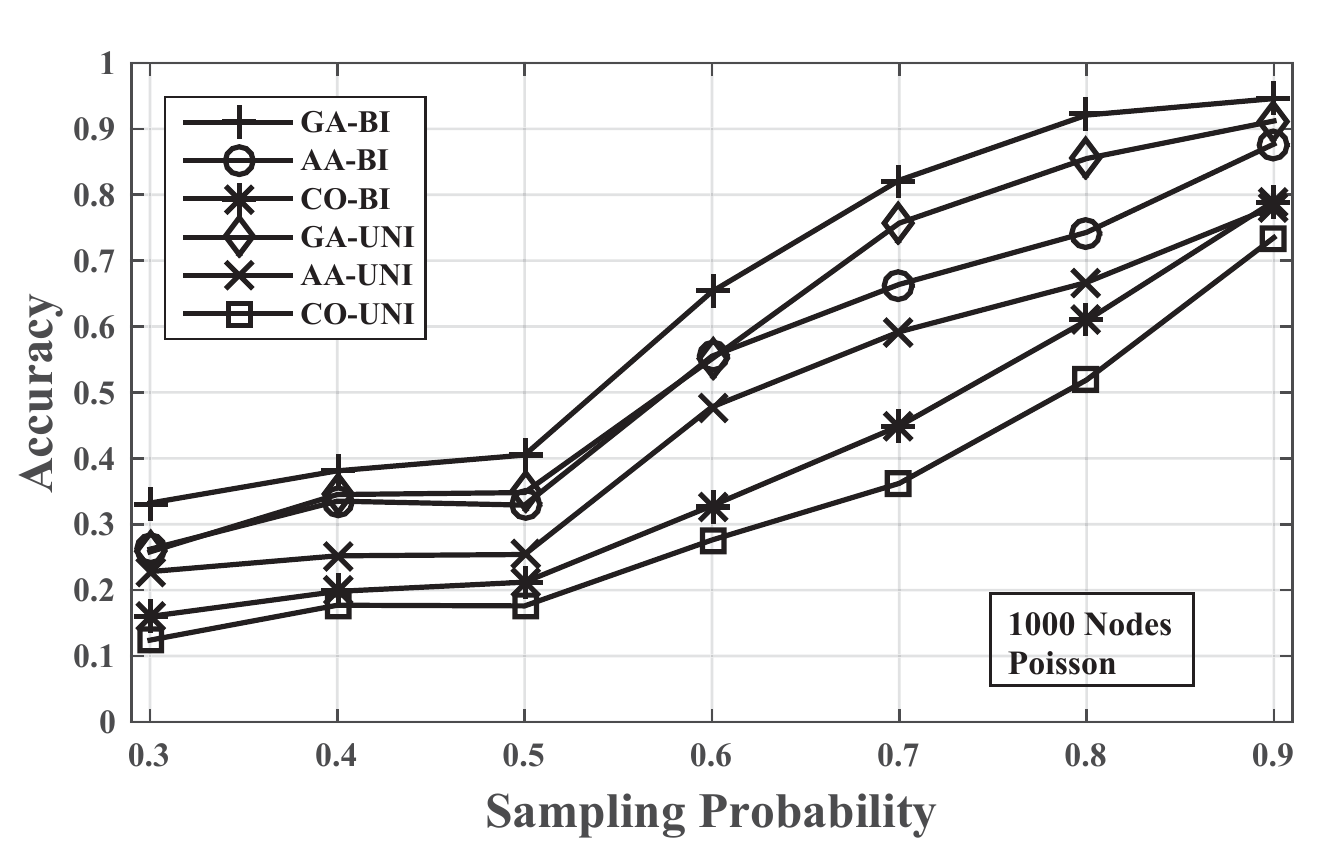}
			\vspace{-4mm}
			\label{fig:1000poi}
		\end{minipage}%
		
	}
	\subfigure[]{
		\begin{minipage}[]{0.235\linewidth}
			\centering
			\vspace{-4mm}
			\includegraphics[width=1.02\linewidth]{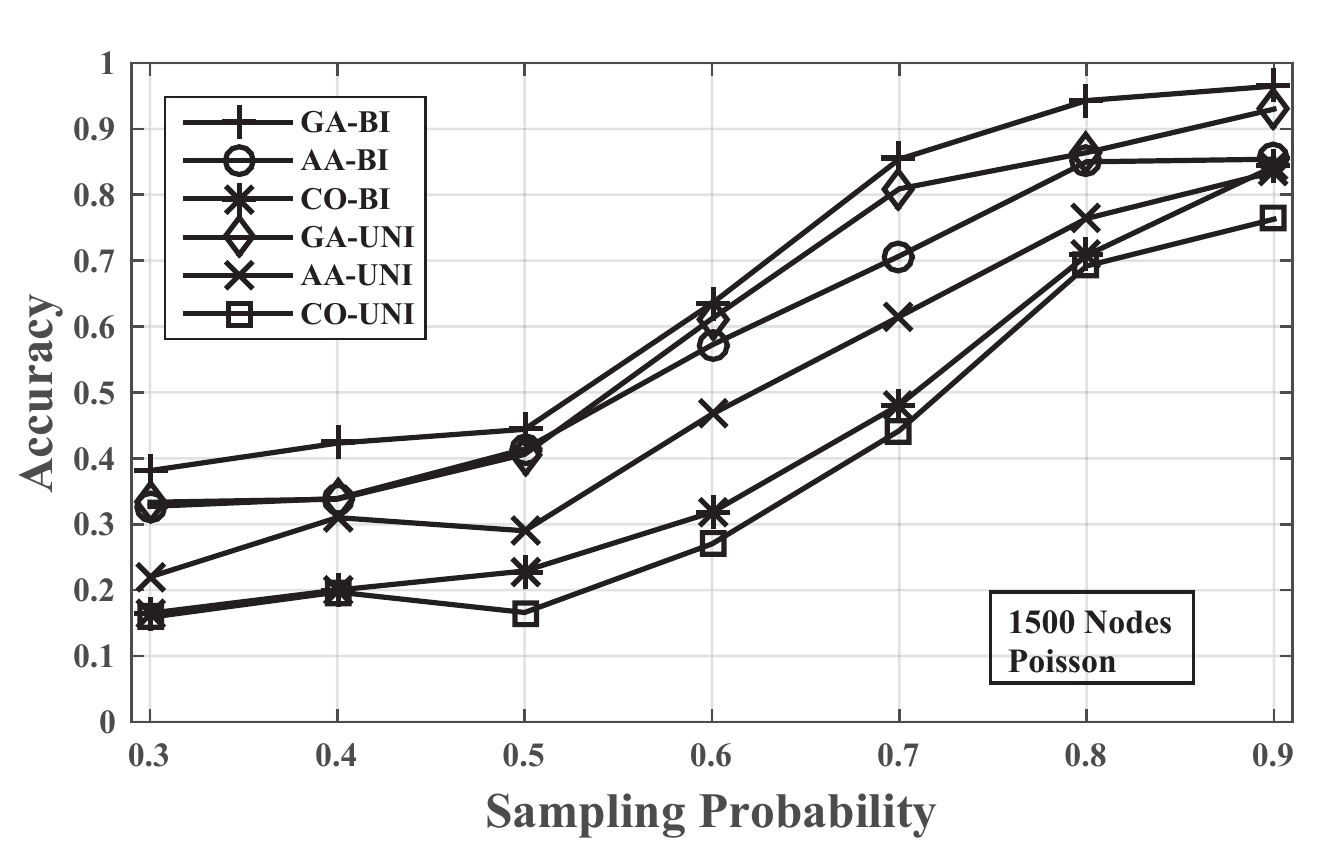}
			\vspace{-4mm}
			\label{fig:1500poi}
		\end{minipage}%
		
	}
	\subfigure[]{
		\begin{minipage}[]{0.235\linewidth}
			\centering
			\vspace{-4mm}
			\includegraphics[width=1.02\linewidth]{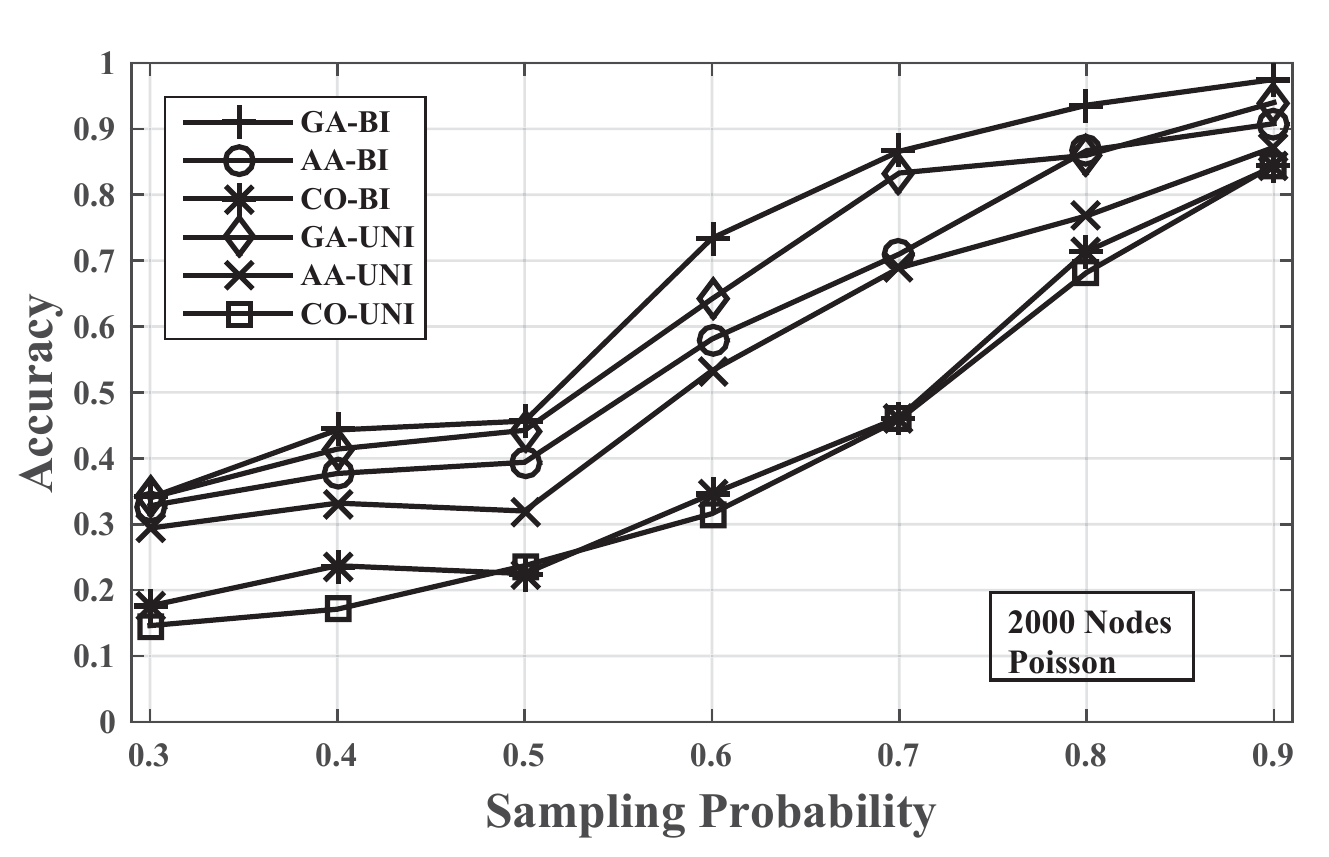}
			\vspace{-4mm}
			\label{fig:2000poi}
		\end{minipage}%
		
	}

	\subfigure[]{
		\begin{minipage}[]{0.235\linewidth}
			\centering
			\vspace{-4mm}
			\includegraphics[width=1.02\linewidth]{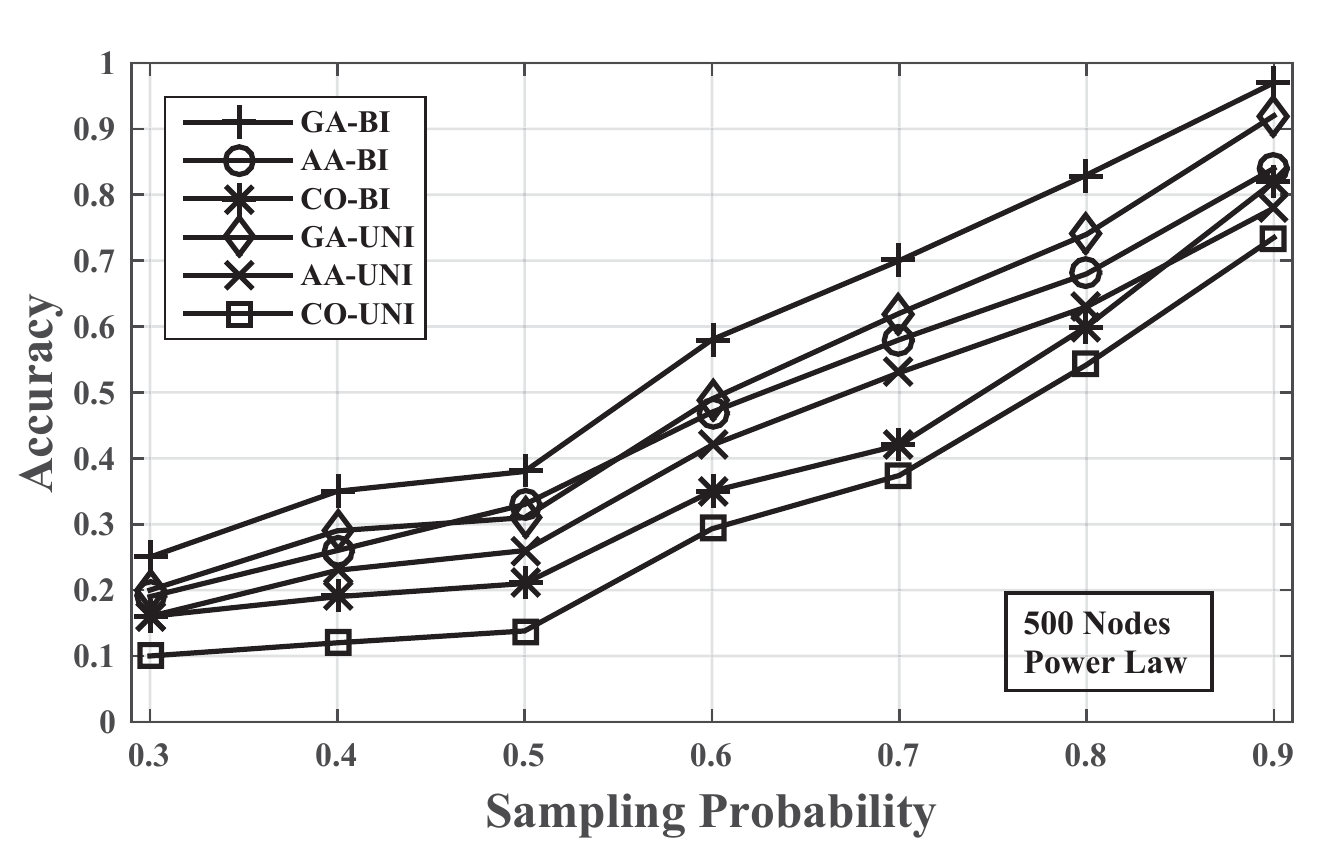}
			\vspace{-4mm}
			\label{fig:500pow}
		\end{minipage}%
		
	}
	\subfigure[]{
		\begin{minipage}[]{0.235\linewidth}
			\centering
			\vspace{-4mm}
			\includegraphics[width=1.02\linewidth]{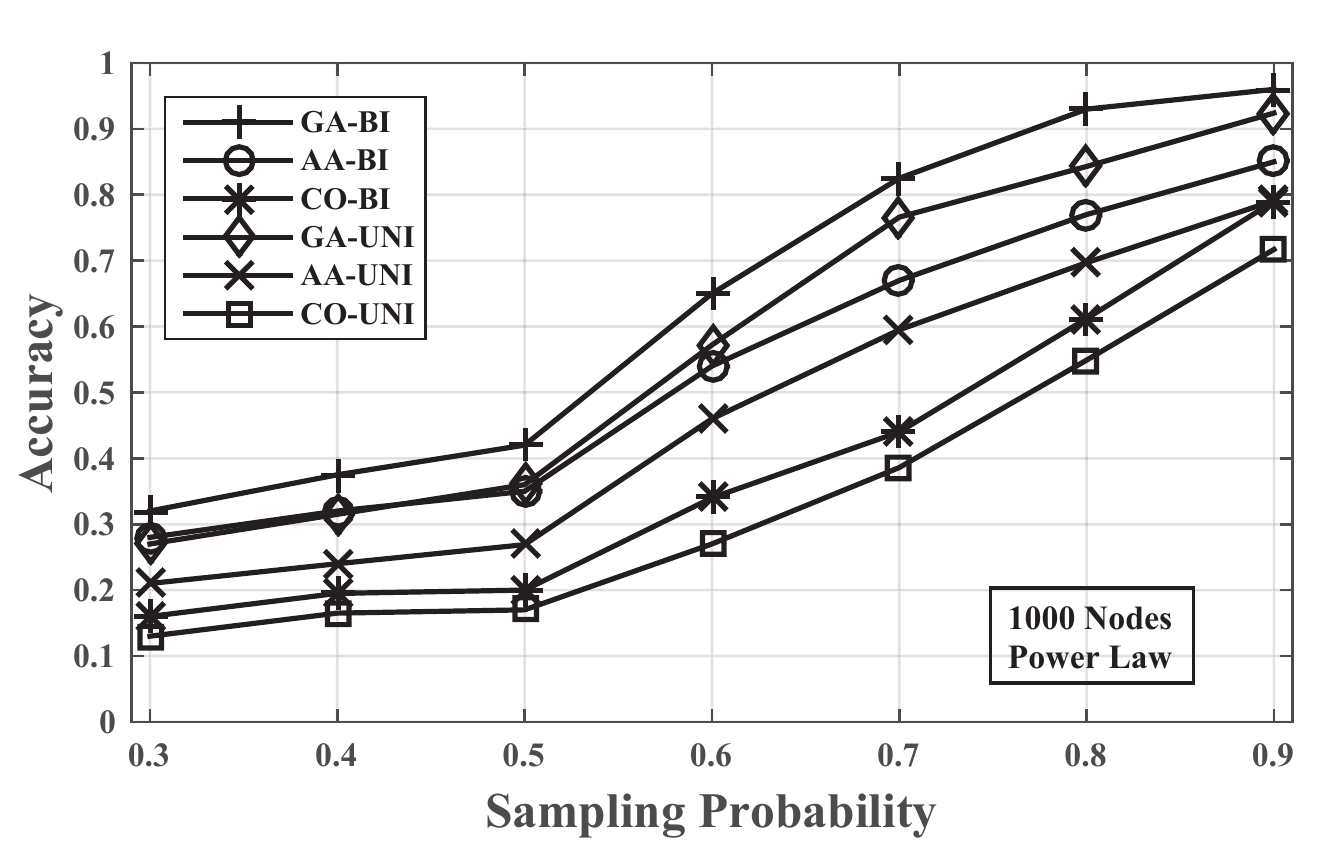}
			\vspace{-4mm}
			\label{fig:1000pow}
		\end{minipage}%
		
	}
	\subfigure[]{
		\begin{minipage}[]{0.235\linewidth}
			\centering
			\vspace{-4mm}
			\includegraphics[width=1.02\linewidth]{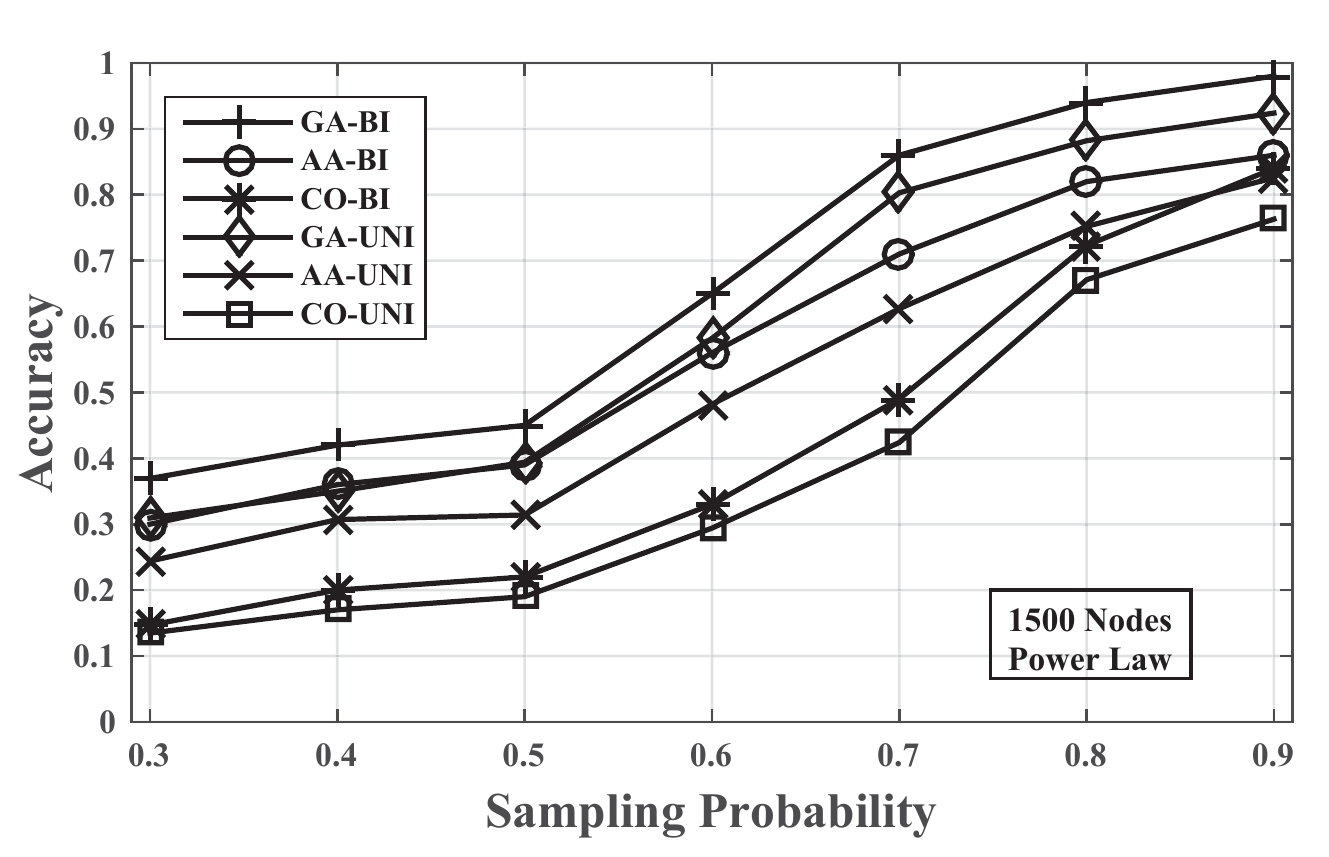}
			\vspace{-4mm}
			\label{fig:1500pow}
		\end{minipage}%
		
	}
	\subfigure[]{
		\begin{minipage}[]{0.235\linewidth}
			\centering
			\vspace{-4mm}
			\includegraphics[width=1.02\linewidth]{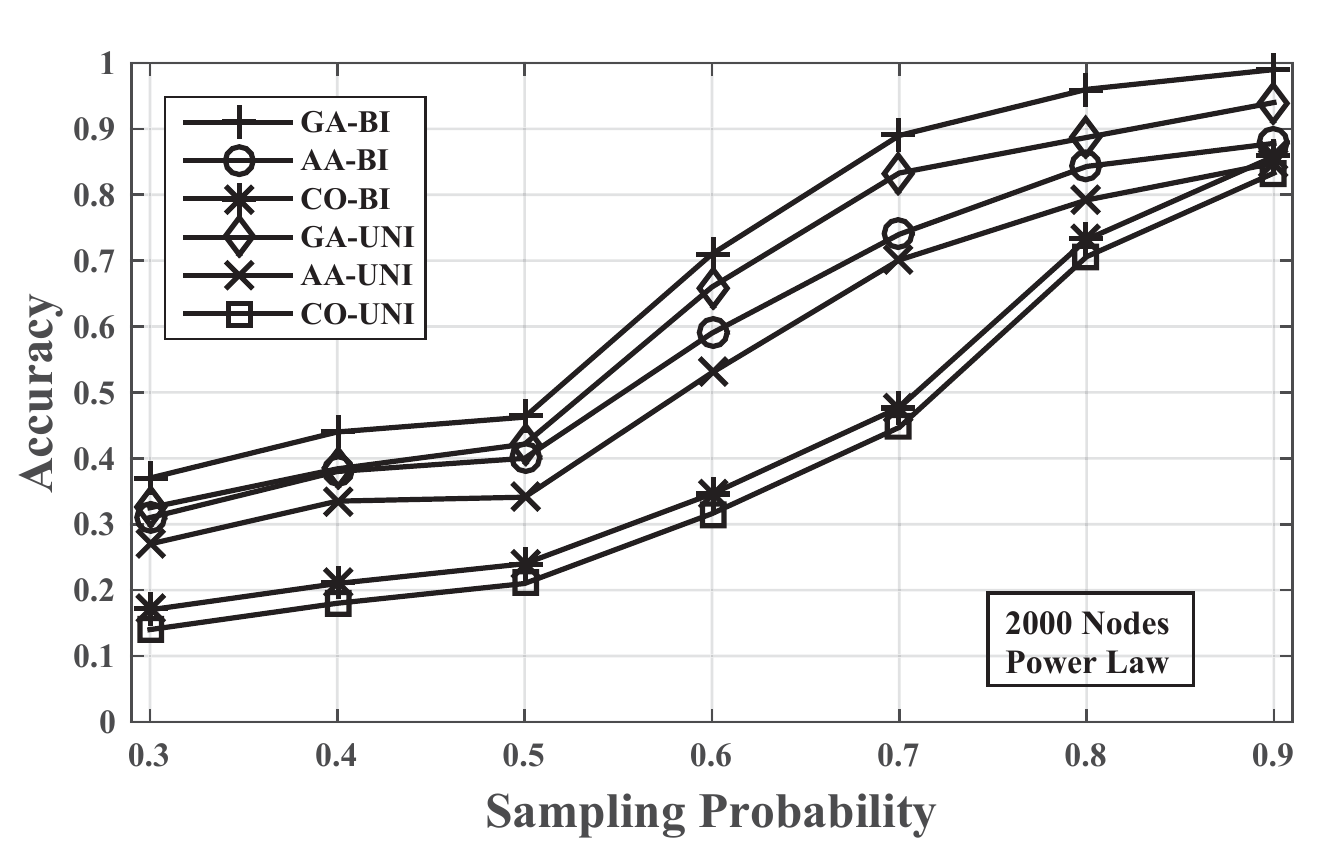}
			\vspace{-4mm}
			\label{fig:2000pow}
		\end{minipage}%
		
	}

	\subfigure[]{
		\begin{minipage}[]{0.235\linewidth}
			\centering
			\vspace{-4mm}
			\includegraphics[width=1.02\linewidth]{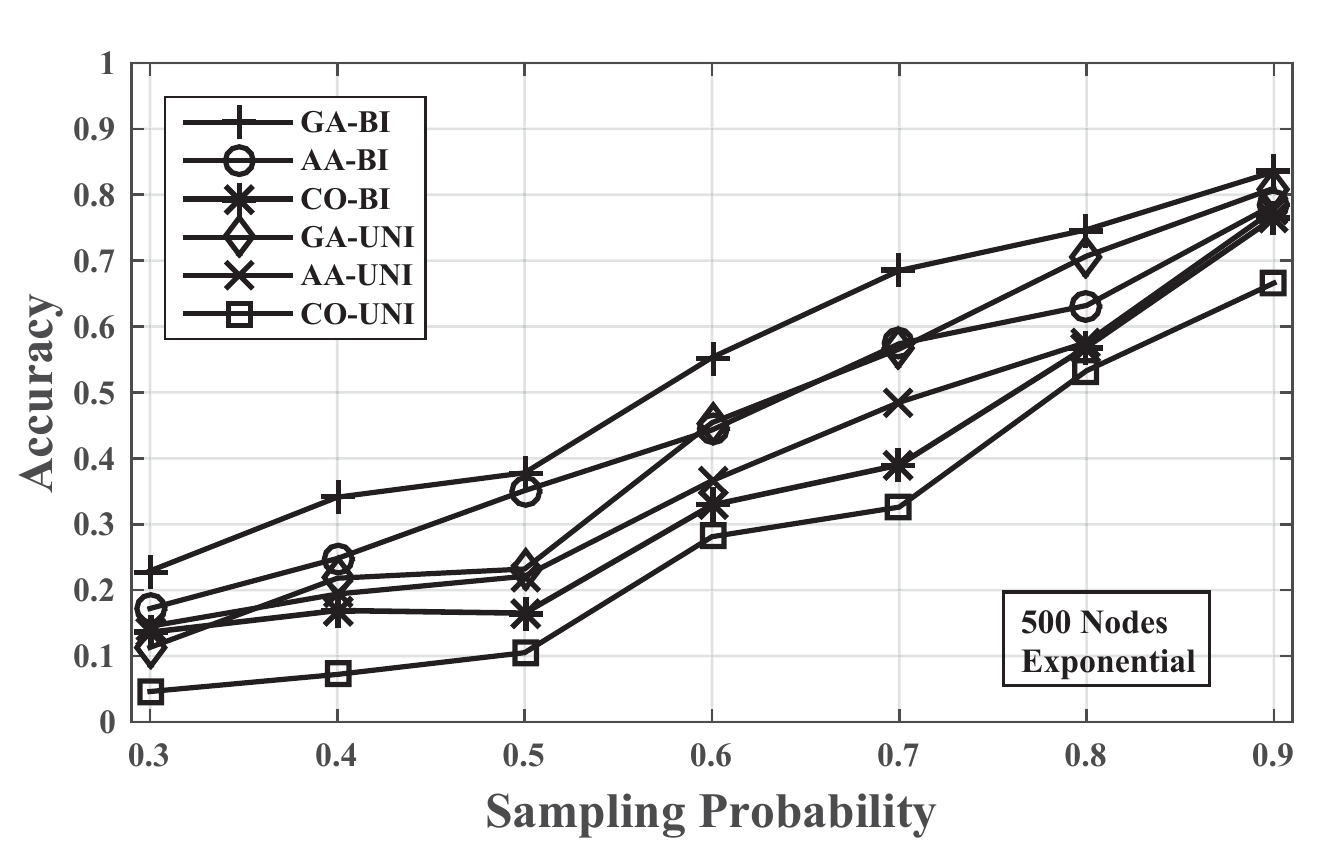}
			\vspace{-4mm}
			\label{fig:500exp}
		\end{minipage}%
		
	}
	\subfigure[]{
		\begin{minipage}[]{0.235\linewidth}
			\centering
			\vspace{-4mm}
			\includegraphics[width=1.02\linewidth]{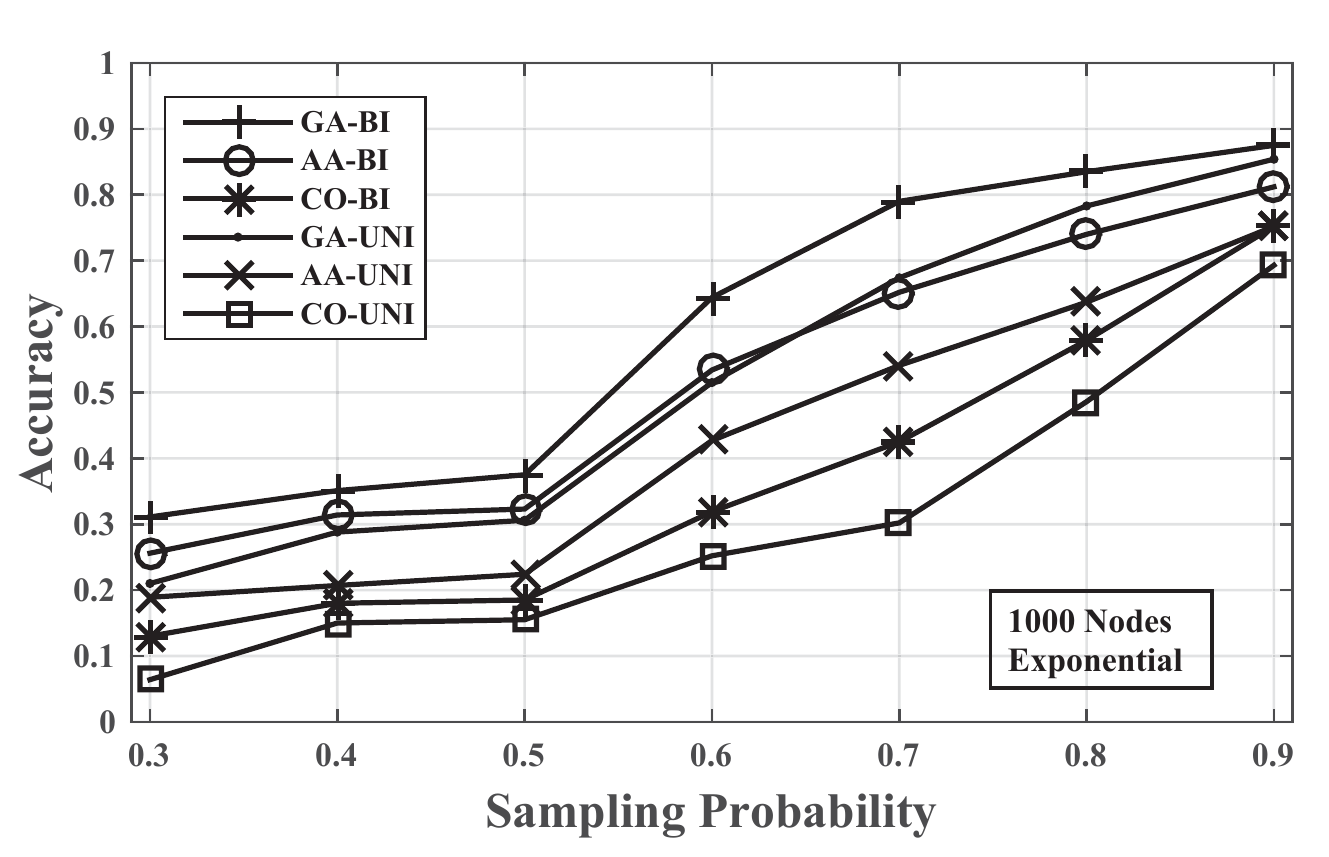}
			\vspace{-4mm}
			\label{fig:1000exp}
		\end{minipage}%
		
	}
	\subfigure[]{
		\begin{minipage}[]{0.235\linewidth}
			\centering
			\vspace{-4mm}
			\includegraphics[width=1.02\linewidth]{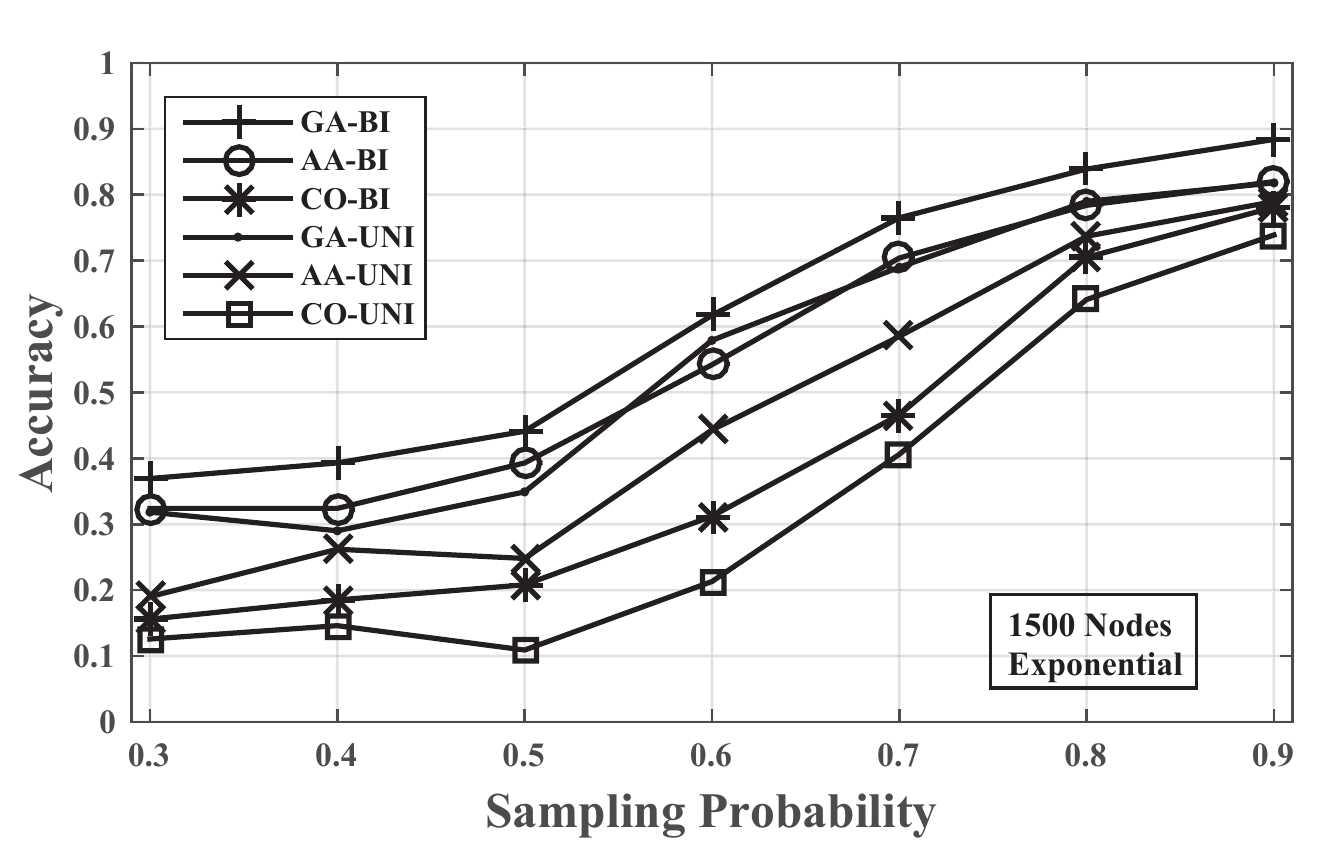}
			\vspace{-4mm}
			\label{fig:1500exp}
		\end{minipage}%
		
	}
	\subfigure[]{
		\begin{minipage}[]{0.235\linewidth}
			\centering
			\vspace{-4mm}
			\includegraphics[width=1.02\linewidth]{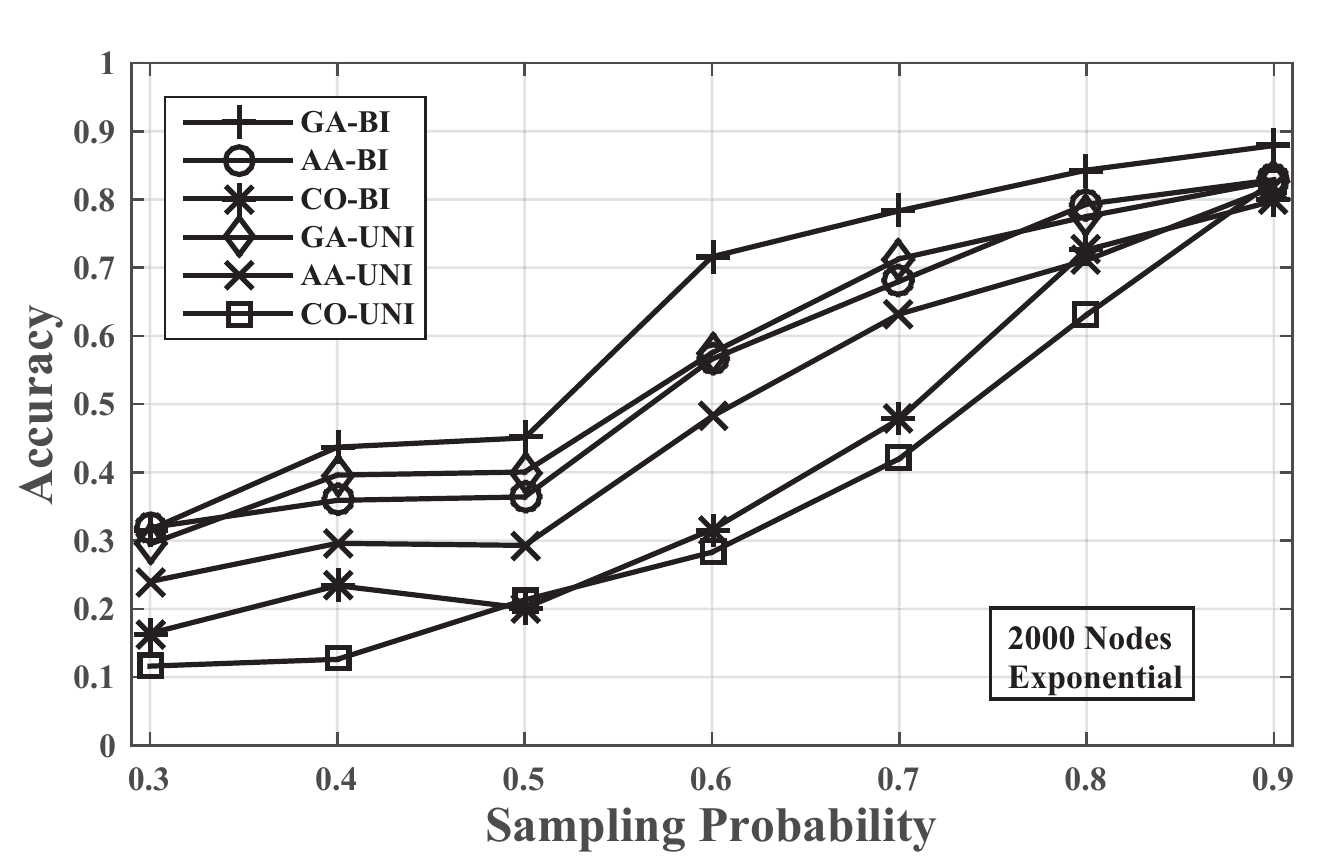}
			\vspace{-4mm}
			\label{fig:2000exp}
		\end{minipage}%
		
	}
	
	\vspace{-4mm}
	\caption{\bf The accuracy of the algorithms on synthetic datasets with different degree distributions.}
	\vspace{-5.7mm}
	\label{fig:syn_accuracy}		
\end{figure*}
	
\vspace{-1.5mm}
\section{Experiments}\label{sec:experiments}
In this section, we present our experimental validation of our theoretical results and the performances of the proposed algorithms. We first introduce our experimental settings and provide detailed results subsequently.
\vspace{-1mm}
\subsection{Experimental Settings}
\vspace{-2mm}
\subsubsection{Experiment Datasets}
Recall that the two key assumptions made in the modeling are that the underlying social network is generated by the stochastic block model and that the published and the auxiliary networks are sampled from the underlying network. To validate our theoretical findings and meanwhile evaluate the proposed algorithms in real contexts, we conduct experiments on three different types of data sets, with each one closer to the practical situations than the last one by gradually relaxing the assumptions.

	(i) \textbf{Synthetic Dataset:} Following the stochastic block model, we generate three sets of networks with Poisson, power law and exponential expected degree distributions respectively by properly assigning the community affinity values $\{p\}$. The size of each community is determined by adding a slight variation to the average community size, which equals to the number of nodes divided by the number of communities. For each set of networks, we take the sampling probabilities of the published and the auxiliary networks as $s_1=s_2$ ranging from 0.3 to 0.9. As this dataset strictly observes the assumptions of our models, it provides direct validations to our theoretical results.
	
	(ii) \textbf{Sampled Social Networks:} The underlying social networks are extracted from LiveJournal online social network \cite{cite:livejournal}, with the communities following from the ground-truth communities in  LiveJournal and the affinity values assigned to be proportional to the ratio of the edges between the communities over the number nodes in the communities. The published and the auxiliary networks are sampled from the underlying networks, again, with the sampling probabilities $s_1=s_2$ ranging from 0.3 to 0.9.  This ``semi-artificial" dataset lies in the middle of synthetic datasets and true cross-domain networks, which enables us to measure the robustness of our theoretical results against the restrictions imposed on the underlying social network.
	
	(iii) \textbf{Co-authorship Networks:} We extract four co-authorship networks in different areas from Microsoft Academic Graph (MAG) \cite{cite:MAG}. From those, we construct a group of networks with equivalent sets of nodes (2053 nodes in each set) and set up the correspondence of nodes as ground-truth based on the unique identifiers of authors in MAG. The communities are assigned based on the institution information of the authors (the affinity values in this case are assigned as in Sampled Social Networks). The four networks are then combined into six pairs, in which one is set as the published network and the other as the auxiliary network. Without relying on any artificial assumptions of generating the published and auxiliary networks, these procedures enable us to construct most genuine scenarios of de-anonymization from cross-domain social networks, which renders the dataset a touchstone for the applicability of our proposed algorithms.

Note that our empirical results in the first two datasets are respectively obtained by taking the average from 50 repetitive experiments. The statistics of the datasets are summarized in Table \ref{table:summary}.
 \vspace{-1mm}
\subsubsection{Algorithms Involved in Comparisons}
For both bilateral and unilateral de-anonymization, we run genetic algorithm (\textbf{GA-BI,GA-UNI}) in hope of finding exact minimizer of our cost functions, i.e., the optimal solution of BI-MAP-ESTIMATE and UNI-MAP-ESTIMATE problems. In both de-anonymization cases, we also evaluate the performance of our two proposed algorithms: the additive approximation algorithm (\textbf{AA-BI,AA-UNI}) and the convex optimization-based algorithm (\textbf{CO-BI,CO-UNI}).
\vspace{-1mm}
\subsubsection{Performance Metrics}
\vspace{-0.5mm}
The two performance metrics we calculate in the experiments are the \textbf{accuracy} of the mappings yielded by the algorithms and the values of the cost function $\Delta_{\pi}$ of the mappings. The accuracy of a mapping $\pi$ is defined as the portion of the nodes that $\pi$ maps correctly (as the ground-truth correct mapping) over the total number of nodes. Since we are not interested in the absolute values of the cost function of the mapping, we calculate the \textbf{relative value} with respect to the cost function of the mappings produced by \textbf{GA}, i.e., for a mapping $\pi$ and the mapping $\pi_{GA}$ produced by \textbf{GA}, $\pi$'s relative value is computed as $(\Delta_{\pi}-\Delta_{\pi_{GA}})/\Delta_{\pi_{GA}}$. Due to space limitations, we defer all the graphical representations of results on the mappings' cost function to \textbf{Appendix \ref{sec:sim_figure}}.


\vspace{-1.5mm}
\subsection{Experiment Results}
\vspace{-2mm}
\subsubsection{Synthetic Networks}
\vspace{-1mm}
We plot the performance of the aforementioned algorithms on synthetic networks with $\{500,1000,1500,2000\}$ number of nodes in Figures \ref{fig:syn_accuracy} and \ref{fig:syn_cost}, based on which we have the following observations: (i) Both \textbf{GA-BI} and \textbf{GA-UNI} exhibit good performance, achieving a de-anonymization accuracy close to 1 when the sampling probability is large in networks with Poisson and power law degree distribution; (ii) The relative value of the correct mapping (\textbf{TRUE-BI,TRUE-UNI}) is fairly small. Hence, we conclude that,  when the sampling probability is large, the cost function based on MAP estimation is an effective metric in both bilateral and unilateral de-anonymization, and is applicable to a wide range of degree distribution, which justify our theoretical results on the validity of the MAP estimate. However, when the sampling probability is small (e.g. $s=0.3,0.4$) or the expected degree distribution has large variation (exponential distribution), the accuracy of \textbf{GA} degrades substantially, only achieving a value of less than 0.4. This can be attributed to the fact that when the sampling probability becomes small, the published and the auxiliary networks have lower degree of structural similarity and the parameters deviate from the conditions in our theoretical results. 

In terms of the two algorithms we propose, we can see that they obtain good performance with respect to both approximately minimizing the cost function and unraveling the correct mapping, with \textbf{AA} superior than \textbf{CO} especially in low-sampling-probability area. Note that although the relative value of the two algorithms is large in high-sampling-probability area, this does not imply the poor performance of the algorithms but is mainly due to the optimal $\Delta_{\pi}$ becoming considerably small as the similarity of $G_1$ and $G_2$ grows high.


\vspace{-1.5mm}
\subsubsection{Sampled Social Networks}
\vspace{-0.5mm}
Figures \ref{fig:sam_accuracy} and \ref{fig:sam_cost} plot our empirical results on the second datasets where the published and auxiliary networks are sampled from real social networks with the number of nodes set as $\{500, 1000, 1500, 2000\}$.

As demonstrated by Figures \ref{fig:sam_accuracy} and \ref{fig:sam_cost}, although in this case the underlying social networks do not follow the stochastic block model, through minimizing the cost function we can still reveal a large proportion (up to 80\%) of the correct mapping, which demonstrates the robustness of the cost function we proposed. Furthermore, the two algorithms \textbf{AA} and \textbf{CO} still achieve reasonable accuracy of up to 0.7, which is not surprising due to that the cost function they seek to minimize is still effective in this case. However, a little defect is that the accuracy of \textbf{AA} can be higher than \textbf{GA} at some points. This reflects that the deviation of the real life social networks from the stochastic block model more or less influences the quality of the MAP estimate.

	\begin{figure}[htbp]
		\centering
		\subfigure[]{
			\begin{minipage}[]{0.48\linewidth}
				\centering
				\vspace{-4mm}
				\includegraphics[width=1.02\linewidth]{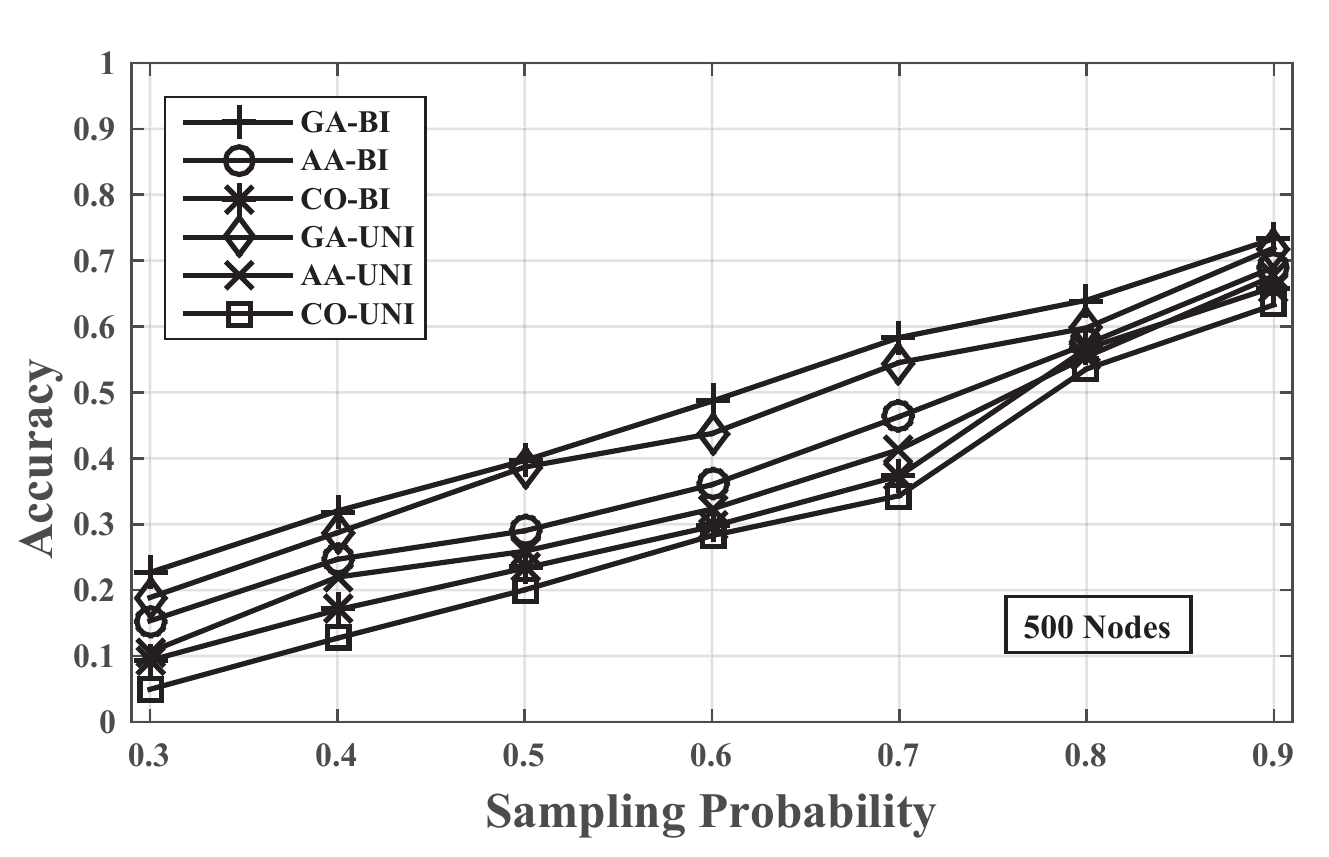}
				\vspace{-4mm}
				\label{fig:500sam}
			\end{minipage}%
			
		}
		\subfigure[]{
			\begin{minipage}[]{0.48\linewidth}
				\centering
				\vspace{-4mm}
				\includegraphics[width=1.02\linewidth]{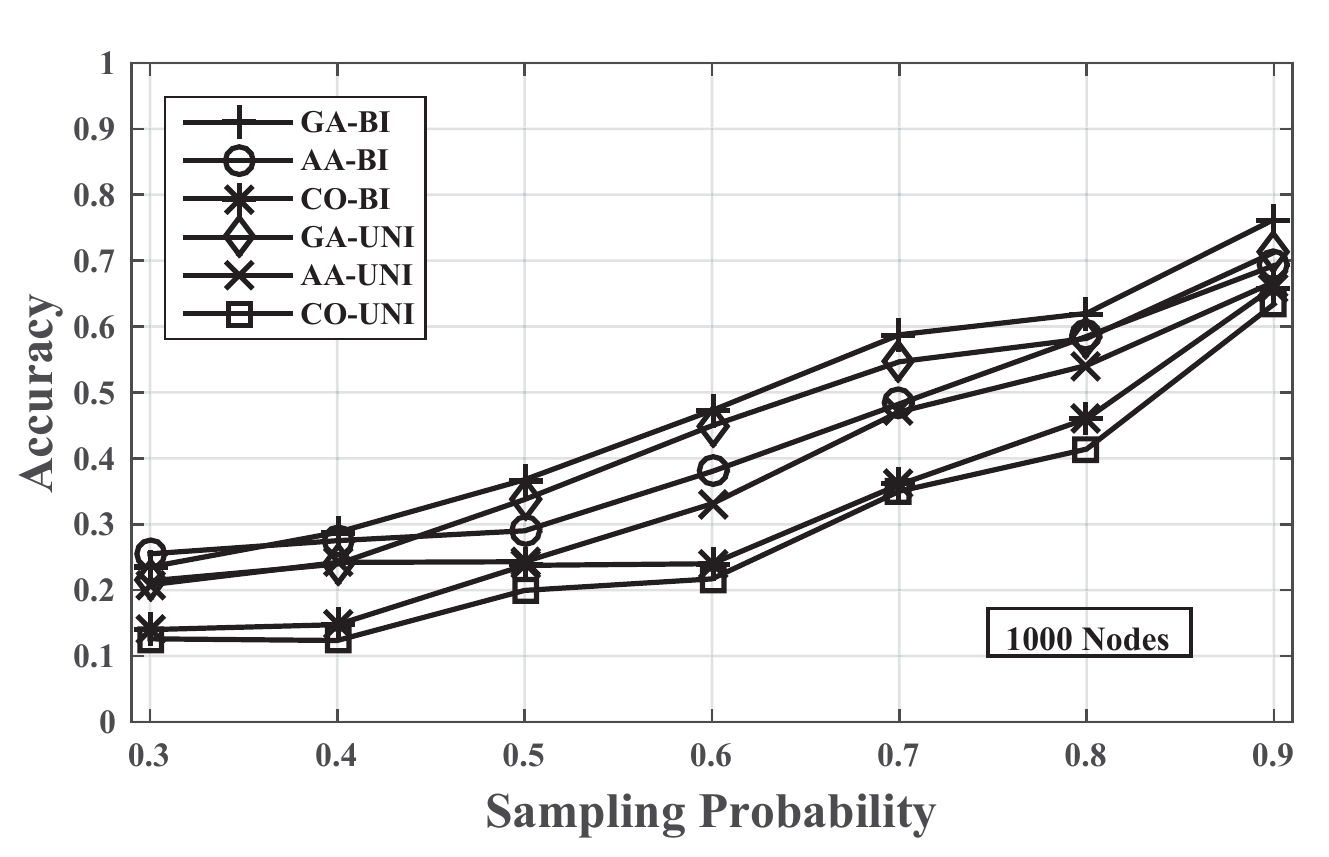}
				\vspace{-4mm}
				\label{fig:1000sam}
			\end{minipage}%
			
		}
		\subfigure[]{
			\begin{minipage}[]{0.48\linewidth}
				\centering
				\vspace{-4mm}
				\includegraphics[width=1.02\linewidth]{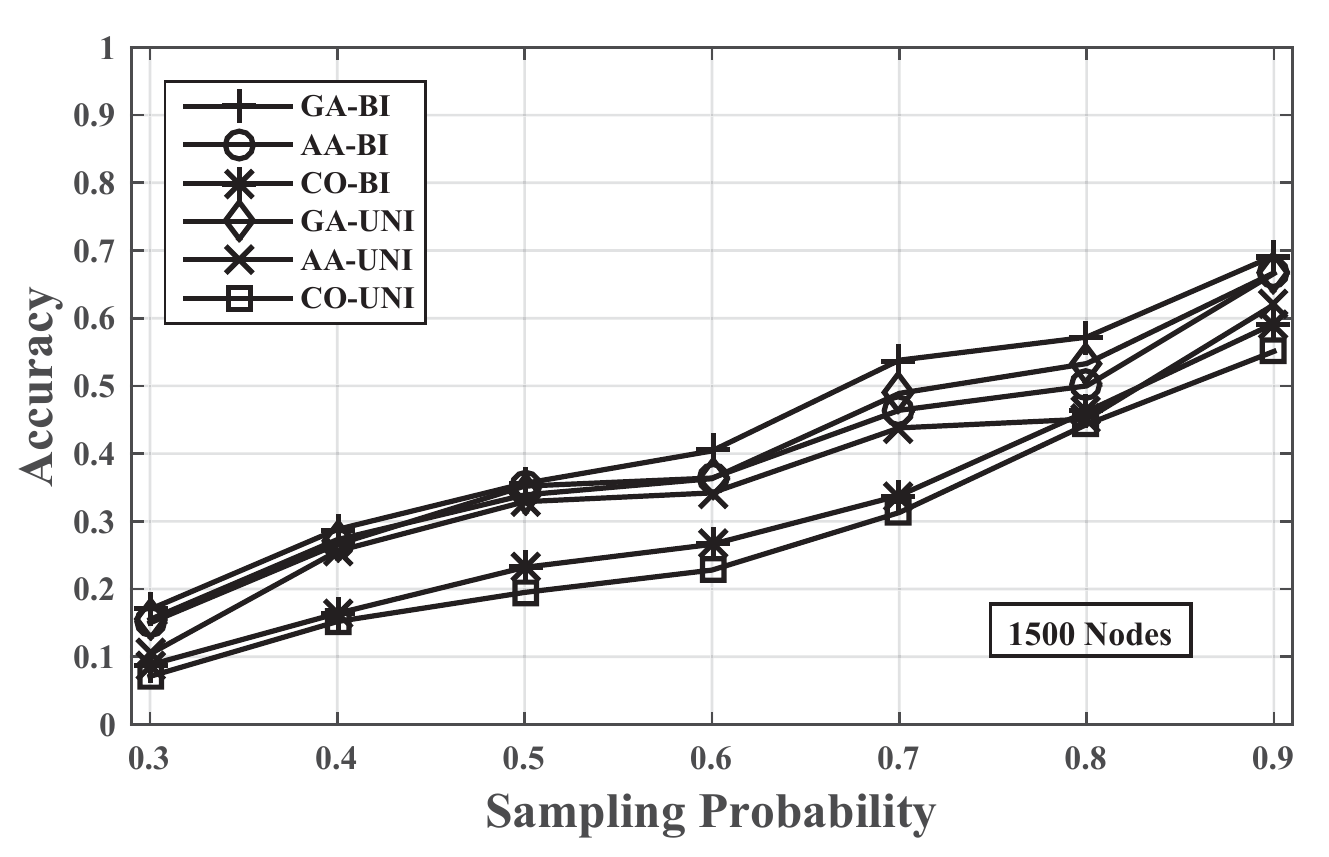}
				\vspace{-4mm}
				\label{fig:1500sam}
			\end{minipage}%
			
		}
		\subfigure[]{
			\begin{minipage}[]{0.48\linewidth}
				\centering
				\vspace{-4mm}
				\includegraphics[width=1.02\linewidth]{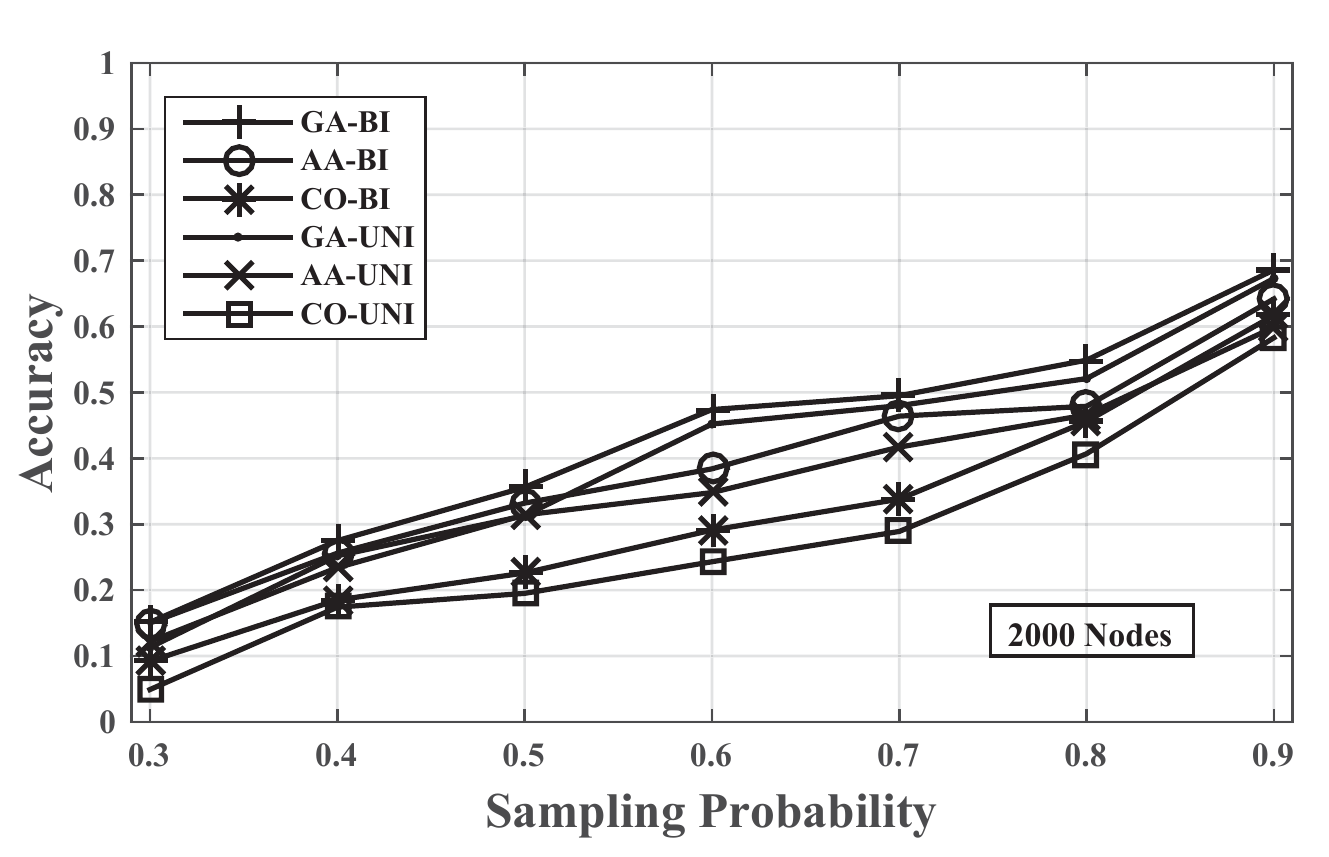}
				\vspace{-4mm}
				\label{fig:2000sam}
			\end{minipage}%
			
		}

		\vspace{-4mm}
		\caption{\bf The accuracy of the algorithms on Sampled Social Networks}
		\label{fig:sam_accuracy}
		\vspace{-5mm}		
	\end{figure}

\vspace{-1mm}
\subsubsection{Cross-domain Co-authorship Networks}
\vspace{-0.5mm}
As stated in experimental setup, we extract four groups of cross-domain co-authorship networks named as Networks A, B, C, D and thus construct six scenarios for social network de-anonymization\footnote{We do not distinguish the interchange of the published and auxiliary networks as different scenarios.}. We evaluate the performance of the algorithms on the six scenarios and show the results in Figures \ref{fig:coauthor_accuracy} and \ref{fig:coauthor_cost}. The figures present several observations and implications: (i) the proposed cost functions still serve as meaningful media for recovering the correct mapping even in realistic scenarios as the  relative value of the correct mapping is close to zero and \textbf{GA} achieves an average accuracy of 67.3\% in bilateral case and 59.0\% in unilateral case; (ii) The two proposed algorithms still enjoy reasonable accuracy, with \textbf{AA} successfully de-anonymizing 60.8\% of nodes in bilateral case and 51.5\% of nodes in unilateral case, and \textbf{CO} successfully de-anonymizing 44.4\% of nodes in bilateral case and 35.9\% of nodes in unilateral case. Therefore, the two algorithms can be qualified as effective methods for seedless social network de-anonymization, which implies that the privacy of current anonymized networks still suffers from attacks of adversaries even when pre-mapped seeds are unavailable; (iii) The performance of \textbf{CO} is most susceptible to the structure of networks among all three algorithms as the standard deviation of its accuracy on the six scenarios are above 3.5\% (3.51\% for \textbf{CO-BI}, 3.81\% for \textbf{CO-UNI}) while the counterparts of the other two algorithms are below 3.0\%.
\begin{figure}[htbp]
	\centering
	\subfigure[]{
		\begin{minipage}[]{0.48\linewidth}
			\centering
			\vspace{-4mm}
			\includegraphics[width=1.0\linewidth]{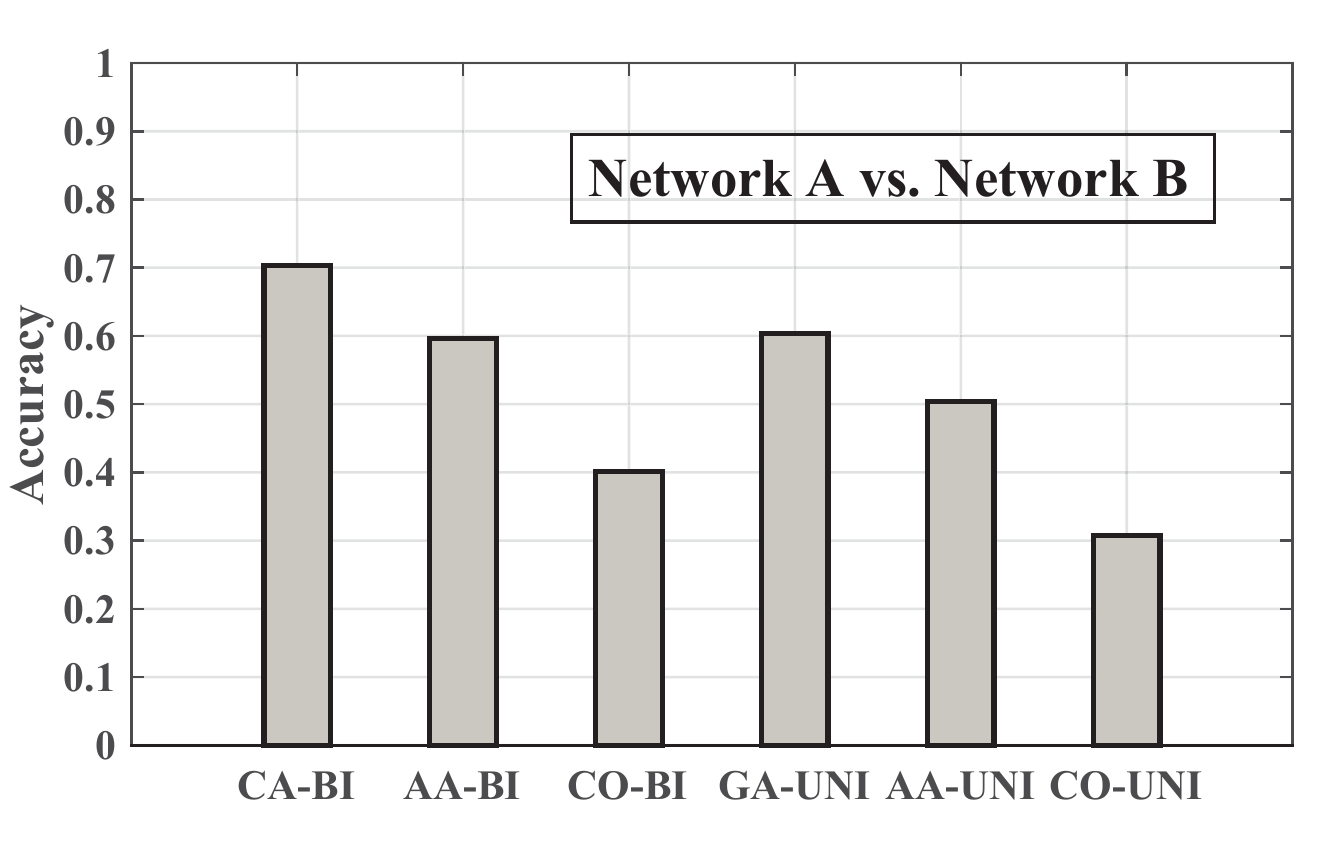}
			\vspace{-5mm}
			\label{fig:COAB}
		\end{minipage}%
		
	}
	\subfigure[]{
		\begin{minipage}[]{0.47\linewidth}
			\centering
			\vspace{-4mm}
			\includegraphics[width=1.0\linewidth]{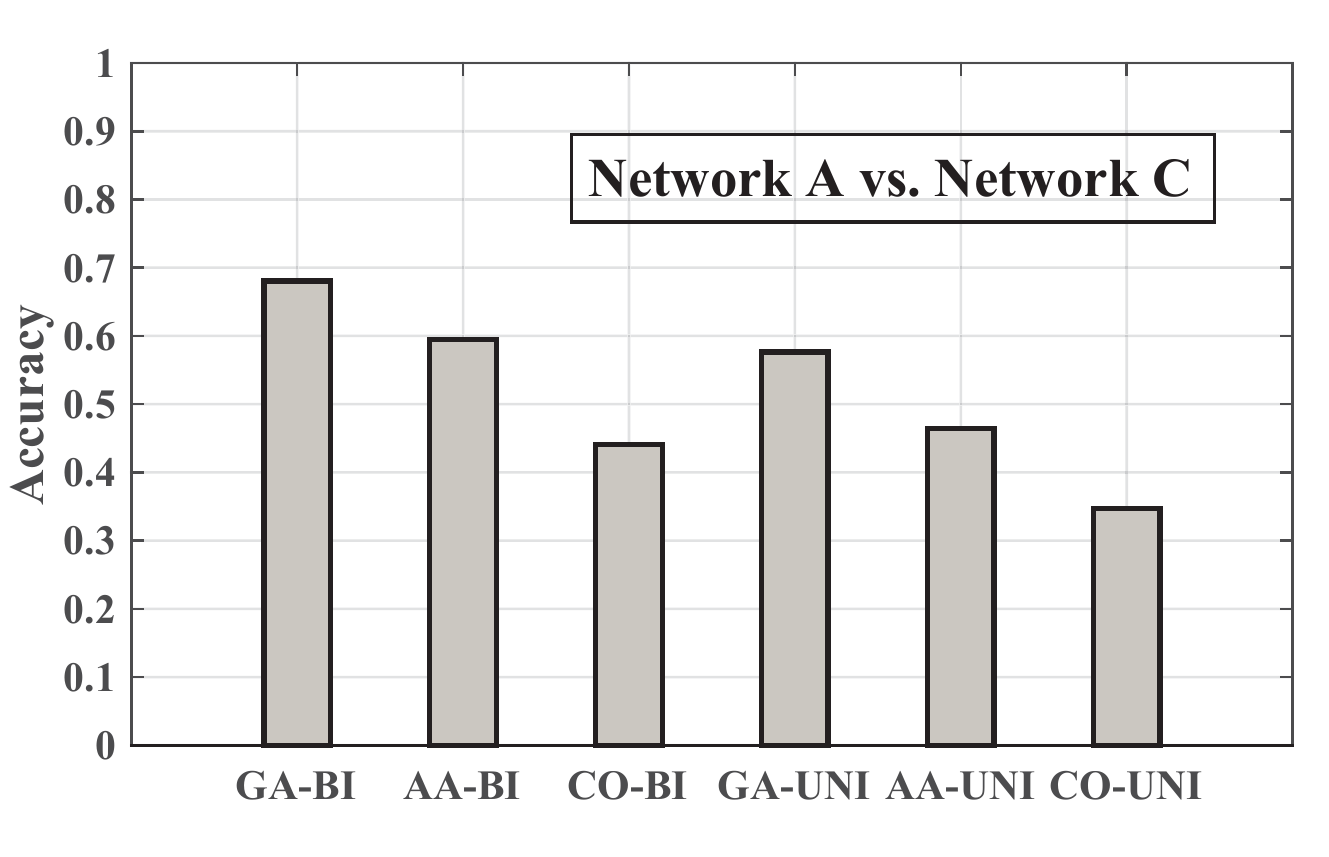}
			\vspace{-5mm}
			\label{fig:COAC}
		\end{minipage}%
		
	}

	\subfigure[]{
		\begin{minipage}[]{0.48\linewidth}
			\centering
			\vspace{-4mm}
			\includegraphics[width=1\linewidth]{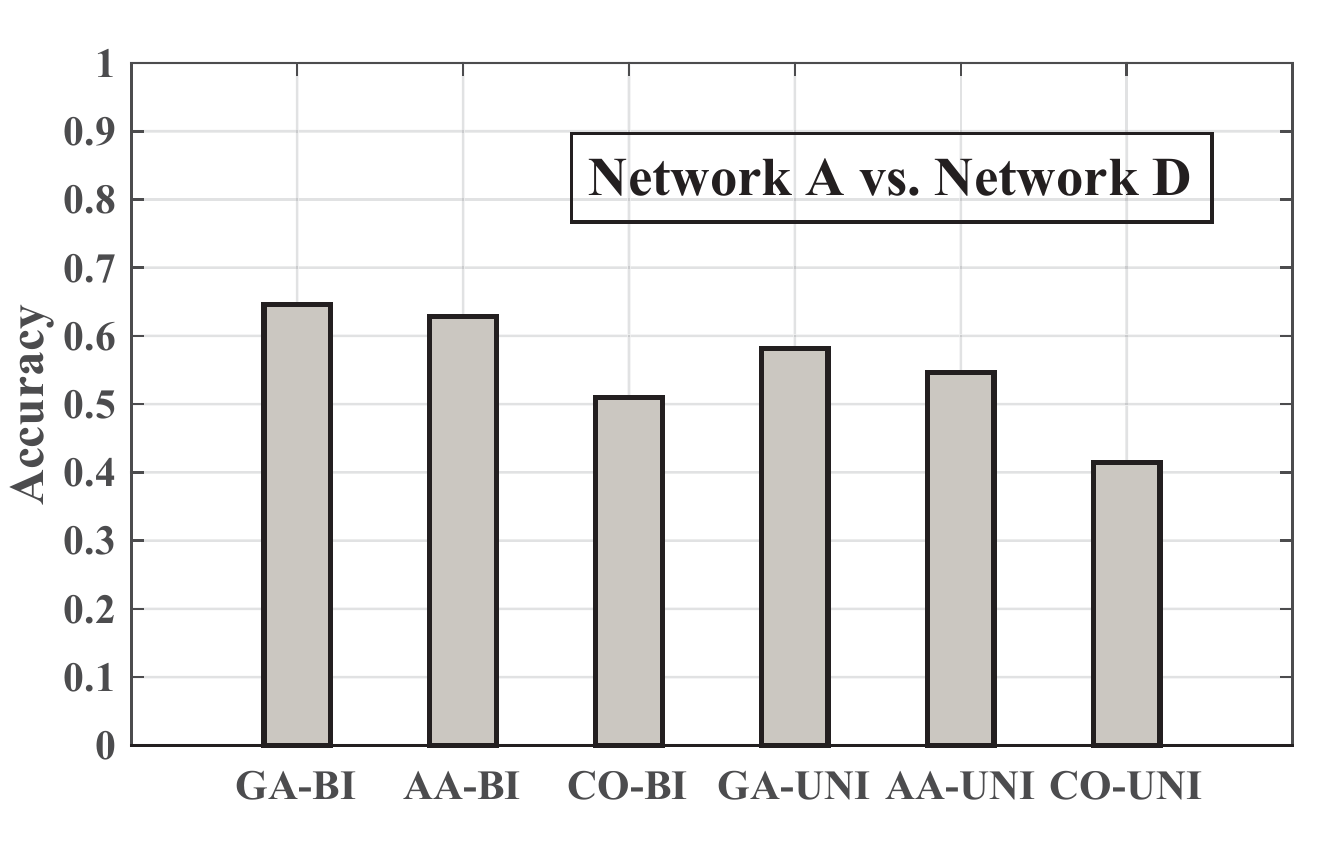}
			\vspace{-5mm}
			\label{fig:COAD}
		\end{minipage}%
		
	}
	\subfigure[]{
		\begin{minipage}[]{0.48\linewidth}
			\centering
			\vspace{-4mm}
			\includegraphics[width=1.0\linewidth]{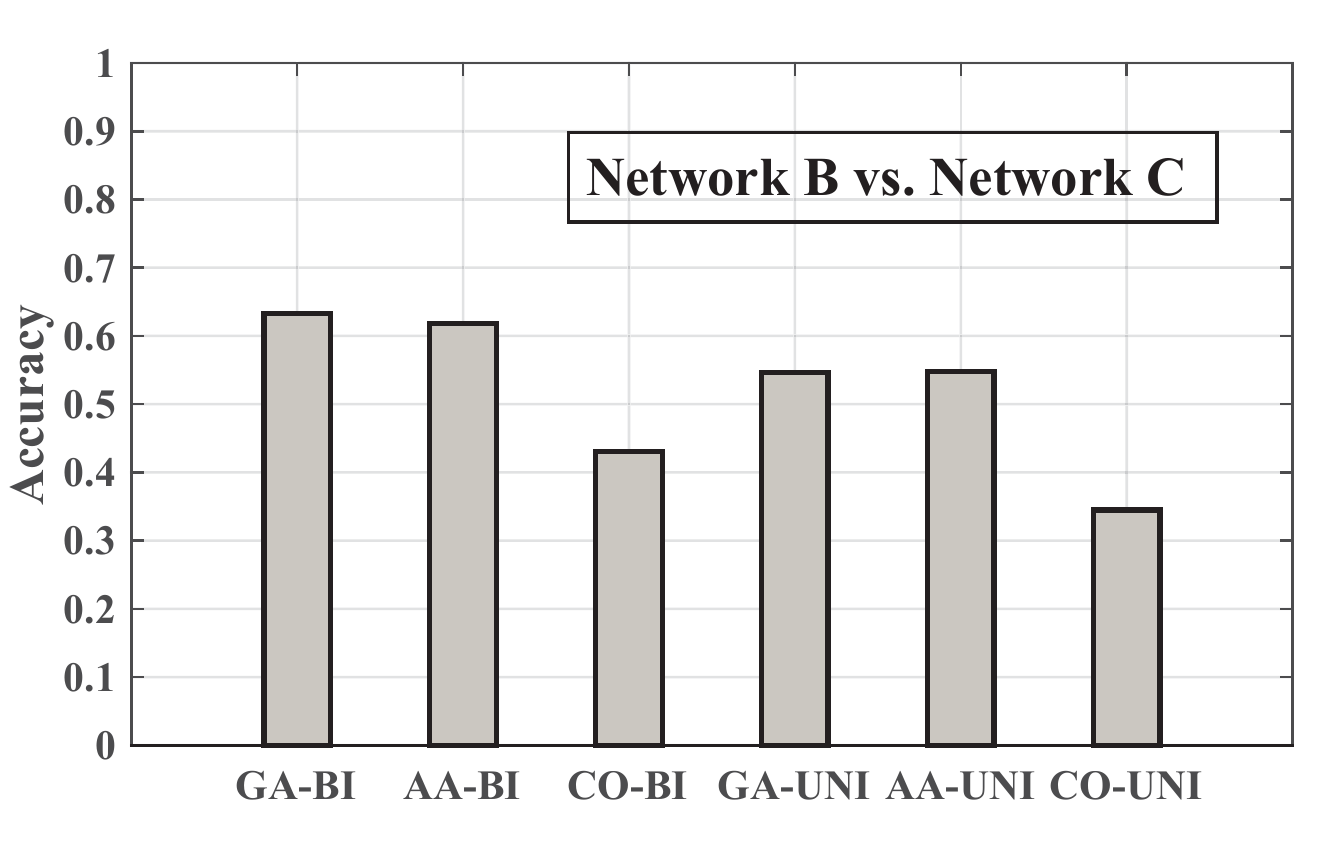}
			\vspace{-5mm}
			\label{fig:COBC}
		\end{minipage}%
		
	}

	\subfigure[]{
		\begin{minipage}[]{0.48\linewidth}
			\centering
			\vspace{-4mm}
			\includegraphics[width=1\linewidth]{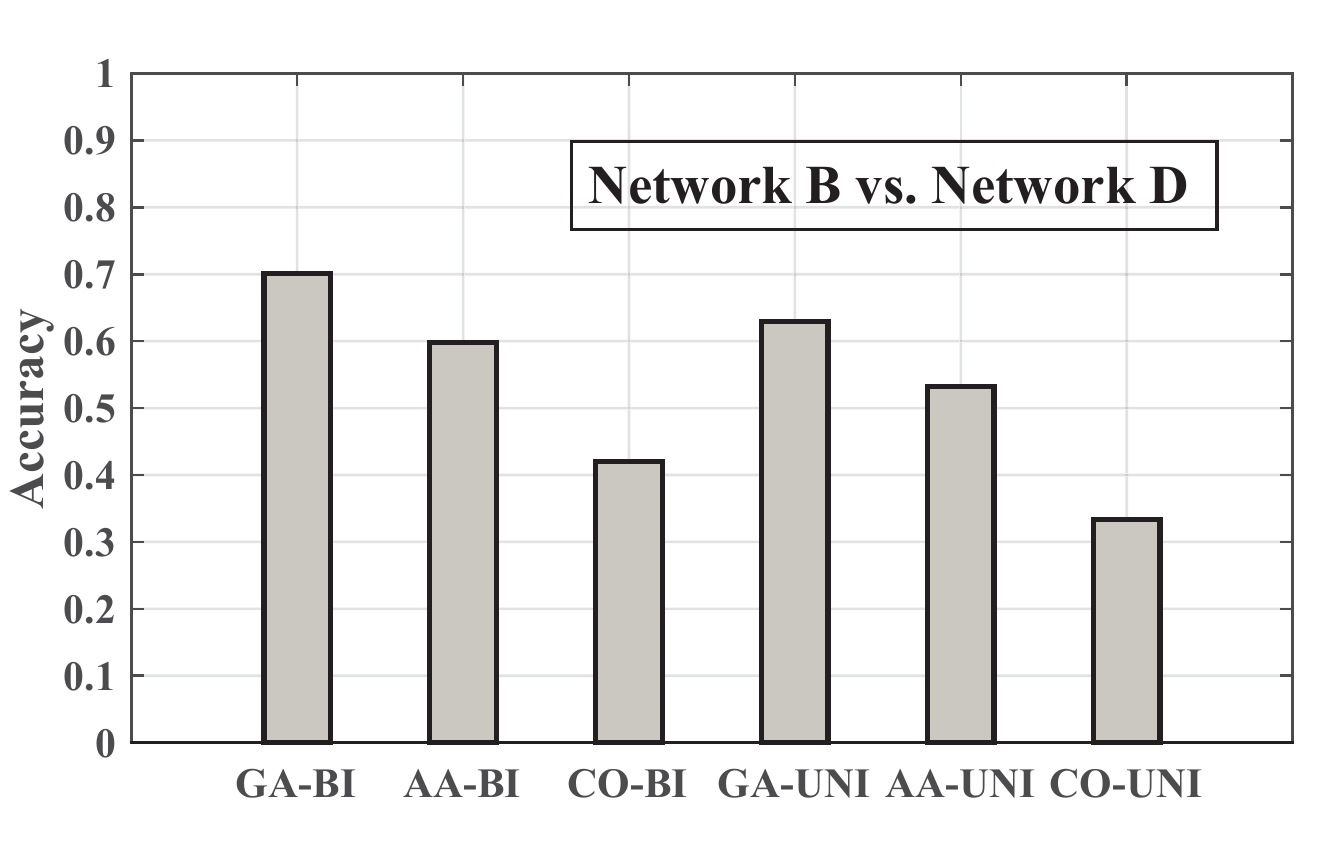}
			\vspace{-5mm}
			\label{fig:COBD}
		\end{minipage}%
		
	}
	\subfigure[]{
		\begin{minipage}[]{0.48\linewidth}
			\centering
			\vspace{-4mm}
			\includegraphics[width=1.0\linewidth]{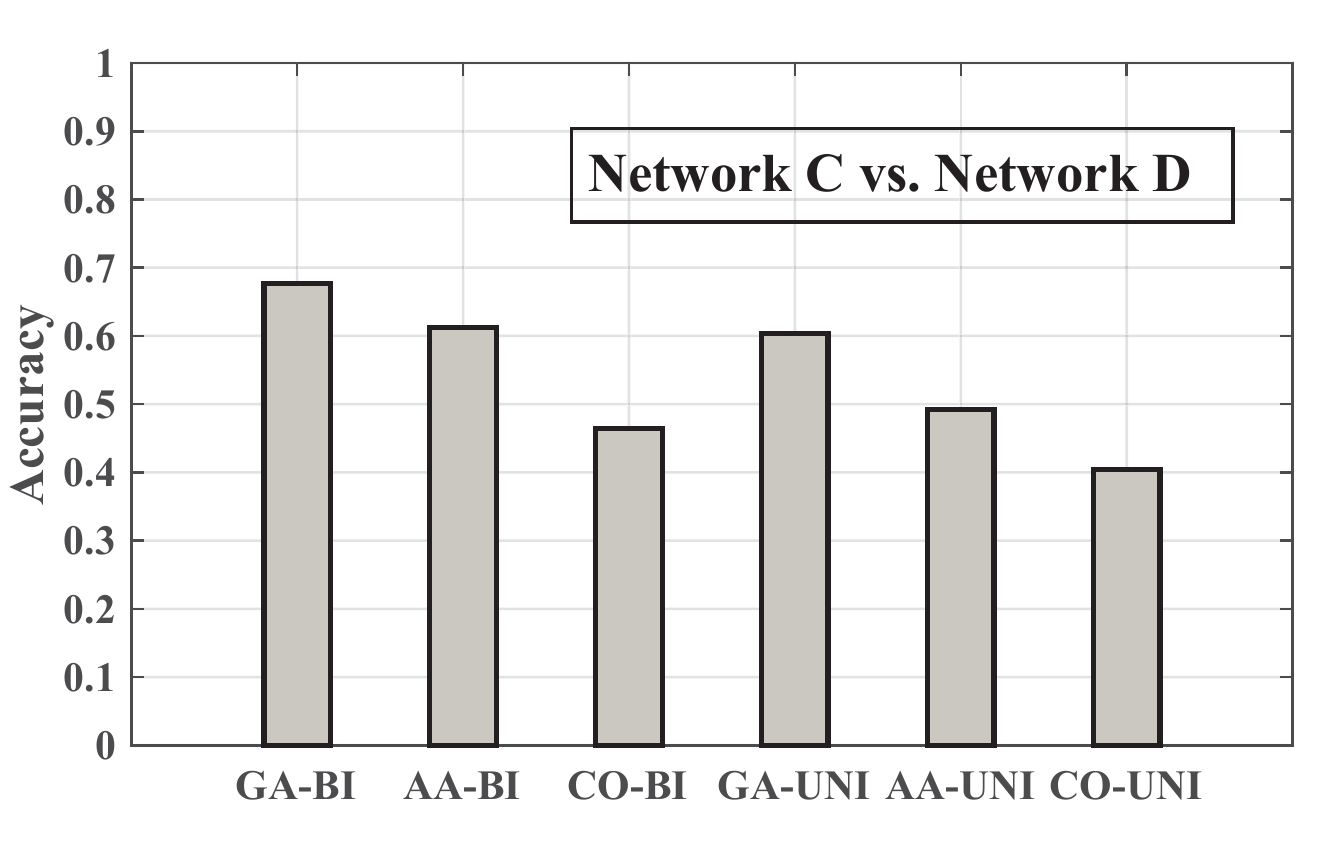}
			\vspace{-5mm}
			\label{fig:COCD}
		\end{minipage}%
		
	}
	\vspace{-5mm}
	\caption{\bf The accuracy of the algorithms on Cross-domain Co-authorship Networks}
	\label{fig:coauthor_accuracy}
	\vspace{-4mm}		
\end{figure}

\subsubsection{Significance of Community Information}
\vspace{-0.5mm}
A notable phenomenon from all the experiments is that the accuracy of the algorithms in bilateral de-anonymization is higher than that in unilateral de-anonymization, especially for \textbf{AA} and \textbf{CO}. According to the experimental results, the gap is at least 3.5\% in each setting and can reach up to 15\% in the worst case. This, from an empirical point of view, demonstrates the importance of the community information on social network de-anonymization.

\vspace{-2mm}	
\section{Conclusion}\label{sec:conclusion}
In this paper, we have presented a comprehensive study of the community-structured social network de-anonymization problem.
Integrating the clustering effect of underlying social network in our models, we have derived a well-justified cost function based on MAP estimation. To further consolidate the validity of such cost function, we have shown that under certain mild conditions, the minimizer of the cost function indeed coincides with the correct mapping. Subsequently, we have investigated the feasibility of the cost function algorithmically by first proving the approximation hardness of the optimization problem induced by the cost function and then proposing two algorithms with their respective performance guarantee by resolving the interweaving of cost function, network topology and candidate mappings through relaxation techniques. All our theoretical findings have been empirically validated through both synthetic and real datasets, with a notable dataset being a set of rare true cross-domain networks that reconstruct a genuine context of social network de-anonymization.

\vspace{-1.5mm}
\appendix\vspace{-1mm}
\section{Proof of Theorem 4.1}\label{sec:MAP-Estimate}\vspace{-1mm}
	 The method we use here is similar to that in \cite{cite:seedless}. Recall that for a mapping $\pi$, we define $\Delta_{\pi}=\sum_{i\le j}^{n}w_{ij}$$|\mathbbm{1}\{(i,j)\in E_1\}-$$\mathbbm{1}\{\pi(i),\pi(j)\in E_2\}|$. Then the proof can be briefly divided into two major steps. The first one is to derive an upper bound for the expectation of the number of (incorrect) mappings $\pi$'s with $\Delta_\pi \le \Delta_{\pi_0}$. The second one is to show that the derived upper bound converges to 0 under the conditions stated in the theorem, as $n\rightarrow\infty$. Based on that, the proof can be concluded as the number of $\pi$'s with $\Delta_{\pi}\le\Delta_{\pi_0}$ goes to 0, i.e., the correct mapping $\pi_0$ is th unique minimizer for $\Delta_{\pi}$ as $n\rightarrow\infty$. Now we turn to the first step as follows:

	1. \textbf{Derivation of the Upper Bound: }We define $\Pi_k$ as the set of all the mappings in $\Pi$ that map $k$ nodes incorrectly. Obviously, $\Pi_0=\{\pi_0\}$. Now we have
	$
	|\Pi_k|\le \binom{n}{k}\left(\frac{k!}{2}\right)\le n^k.
	$
	We subsequently define $S_k$ as a random variable representing the number of incorrect mappings in $\Pi_k$ whose value of cost function is no larger than $\Delta_{\pi_0}$. Formally, $S_k$ is given by
	$
	S_k=\sum_{\pi\in\Pi_k}\mathbbm{1}\{\Delta_\pi\le\Delta_{\pi_0}\}.
	$
	Summing over all $k$, we denote $S=\sum_{k=2}^nS_k$ as the total number of incorrect mappings that induce no larger cost function than the correct mapping $\pi_0$. The mean of $S$ can be calculated as:
	\begin{align}
	\mathbb{E}[S]&=\sum_{k=2}^n\mathbb{E}[S_k]\nonumber=\sum_{k=2}^n\sum_{\pi\in\Pi_k}\mathbb{E}[\mathbbm{1}\{\Delta_\pi\le\Delta_{\pi_0}\}]\nonumber\\&=\sum_{k=2}^n\sum_{\pi\in\Pi_k}Pr\{\Delta_\pi-\Delta_{\pi_0}\le 0\}\nonumber\\
	&\le \sum_{k=2}^{n}n^k\max_{\pi\in\Pi_k}Pr\{\Delta_\pi-\Delta_{\pi_0}\le 0\}.\label{eq:expectation}
	\end{align}
	
	For a mapping $\pi$, let $V_\pi$ be the set of vertices that it maps incorrectly. Then, we define $E_\pi=V_\pi\times V$, i.e., the set of node pairs with one or two vertices mapped incorrectly under $\pi$. For a $\pi\in\Pi_k$, we have 
	$
	|E_\pi|=nk-\frac{k^2}{2}-\frac{k}{2}.
	$
	As every node pair in $V\times V-E_\pi$ is mapped identically in $\pi$ and $\pi_0$, they contribute equally to $\Delta_{\pi_0}$ and $\Delta_\pi$ respectively. Next, we define two random variables for $\pi$ as 
	\begin{align*}
	X_\pi&=\sum_{(i,j)\in E_{\pi}}w_{ij}|\mathbbm{1}\{(i,j)\in E_1\}-\mathbbm{1}\{(\pi(i),\pi(j)\in E_2\}|,\\
	Y_\pi&=\sum_{(i,j)\in E_{\pi}}w_{ij}|\mathbbm{1}\{(i,j)\in E_1\}-\mathbbm{1}\{(i,j)\in E_2\}|.
	\end{align*}
	It is easy to verify that $\Delta_\pi-\Delta_{\pi_0}=X_\pi-Y_\pi$ for all $\pi$, where $Y_\pi$ is the value of cost function contributed by node pairs in $E_{\pi}$ under the correct permutation. For a node pair $(i,j)$, the probability that it contributes to $Y_{\pi}$ equals to $p_{c(i)c(j)}(s_1+s_2-2s_1s_2)$. Therefore, $Y_\pi$ is the weighted sum of independent Bernoulli random variables.
	
	For $X_\pi$, assume that $\pi$ has $\phi\ge 0$ transpositions\footnote{If a mapping $\pi$ has a transposition on $i,j$, it means that $\pi(i)=j$ and $\pi(j)=i$.}, then each transposition induces one invariant node pair in $E_{\pi}$. The remaining node pairs are not invariant under $\pi$, i.e., they are mapped incorrectly under $\pi$. Each node pair $(i,j)$ contributes $w_{ij}$ to $X_\pi$ if $(i,j)\in E_1$ and $(\pi(i),\pi(j))\notin E_2$ or vice versa. This happens with probability $p_{c(i)c(j)}(s_1+s_2-2p_{c(i)c(j)}s_1s_2)$. Note that the random variable for each node pair is not independent. As in \cite{cite:seedless}, we conservatively ignore the positive correlation and get a lower bound of $X_\pi$, which is the weighted sum of independent random Bernoulli variables. Also, since transpositions in $\pi$ can only occur in nodes in $V_\pi$, we have that $\phi\le k/2$. Now, denote $X_{ij}$ as a Bernoulli random variable with mean $p_{c(i)c(j)}(s_1+s_2-2p_{c(i)c(j)}s_1s_2)$ and $Y_{ij}$ as a Bernoulli random variable with mean $p_{c(i)c(j)}(s_1+s_2-2s_1s_2)$ as $Y_{ij}$. Based on the above manipulations, we can get a lower bound of $X_\pi$ and an upper bound of $Y_\pi$ as follows:
	\begin{equation*}
	\begin{aligned}
	X_\pi \overset{(\text{stoch.})\footnotemark}{\ge}&\sum_{(i,j)\in E_{\pi}\backslash \phi}w_{ij}X_{ij}\triangleq X_{\pi}'\\
	Y_\pi \overset{(\text{stoch.})}{\le}&\sum_{(i,j)\in E_{\pi}}w_{ij}Y_{ij}\triangleq Y_{\pi}'.
	\end{aligned}
	\end{equation*}
	\footnotetext{$\overset{(stoch.)}{\ge}$ denotes stochastic domination}

	Therefore, we can use the probability of event $\{X_{\pi'}-Y_{\pi'}\}\le 0$ to upper bound the probability of event $\{X_{\pi}-Y_{\pi}\}\le 0$. Denoting $\lambda_X$ as the expectation of $X_{\pi}'$ and $\lambda_Y$ as the expectation of $Y_{\pi}'$, the bound we use for $Pr\{X_\pi-Y_\pi\le0\}$ is summarized in the following lemma.
	
	\begin{lemma}\label{lemma:upperbound}
		For all mapping $\pi$, random variables $X_\pi$ and $Y_\pi$ satisfy that
		\begin{align}
		Pr\{X_\pi- Y_\pi\le 0 \}\le 2\exp\left(\frac{-(\lambda_X-\lambda_Y)^2}{12(\lambda_X+\lambda_Y)}\right)
		\label{eq:lemma}
		\end{align}
	\end{lemma}
	\begin{proof}
		First, we have that for all $\pi$
		\begin{align*}
		Pr\{X_\pi-& Y_\pi\le 0 \}\le Pr\{X_\pi'-Y_\pi'\le 0 \}\\&\le Pr\Big\{ Y_\pi'\ge \frac{\lambda_X+\lambda_Y}{2}\Big\}+Pr\Big\{ X_\pi'\le \frac{\lambda_X+\lambda_Y}{2}\Big\}
		\end{align*}
		
		Then we invoke Lemma \ref{lemma:difference} (Theorems 1 and 2 in \cite{cite:chernoffbound}), which presents Chernoff-type bounds for weighted sum of independent Bernoulli variables.
		
		\begin{lemma}\label{lemma:difference}(Theorems 1 and 2 in \cite{cite:chernoffbound})
			\ Let $a_1,a_2,\ldots,a_r$ be positive real numbers and let $X_1,\ldots,X_n$ be independent Bernoulli trials with $\mathbb{E}[X_j]=p_j$. Defining random variable $\Psi=\sum_{j=1}^ra_jX_j$ with $\mathbb{E}[\Psi]=\sum_{j=1}^ra_jp_j=m$, we have
			\begin{align*}
			Pr\{\Psi\ge(1+\delta)m\}&\le \exp\left( -m\delta^2/3 \right),\\ 
			Pr\{\Psi\le(1-\delta)m\}&\le \exp\left( -m\delta^2/2 \right). 
			\end{align*}
		\end{lemma}
		
		Using Lemma \ref{lemma:difference} by treating $X_{\pi}'$ and $Y_{\pi}'$ as the weighted ($w_{ij}$) sum of random variables $X_{ij}$ ($Y_{ij}$), we obtain that
		\begin{align*}
		Pr\Big\{Y_\pi'\ge \frac{\lambda_X+\lambda_Y}{2}\Big\}&\le \exp\left( -(\lambda_X-\lambda_Y)^2/12(\lambda_X+\lambda_Y)\right), \\
		Pr\Big\{X_\pi'\le \frac{\lambda_X+\lambda_Y}{2}\Big\}&\le \exp\left( -(\lambda_X-\lambda_Y)^2/8(\lambda_X+\lambda_Y)\right).
		\end{align*}
		
		Hence, we have
		\[
		Pr\{X_\pi-Y_\pi \le 0\}\le 2\exp\left(\frac{-(\lambda_X-\lambda_Y)^2}{12(\lambda_X+\lambda_Y)}\right).
		\qed\]
	\end{proof}

	%

	We now proceed to derive lower bound for the numerator and upper bound for the denominator in the exponent of the RHS of Inequality (\ref{eq:lemma}) to obtain the upper bound of the RHS. By standard calculation, we have 
	\begin{align*}
	&(\lambda_X-\lambda_Y)^2\\\ge& \left(2\displaystyle{\sum_{(i,j)\in E_\pi\backslash \phi}}w_{ij}p_{c(i)c(j)}(1-p_{c(i)c(j)})s_1s_2-\frac{k\overline{w}\beta(s_1+s_2-2s_1s_2)}{2}\right)^2\\
	\ge &
	\frac{k^2}{4}\left[4\left(n-\frac{k}{2}-1\right)\underline{w}\alpha(1-\beta)s_1s_2-\overline{w}\beta(s_1+s_2-2s_1s_2) \right]^2, 
	\end{align*}
	\normalsize
	and
	\begin{align*}
	&\lambda_X+\lambda_Y\\\le& \displaystyle{\sum_{(i,j)\in E_\pi}}[w_{ij}p_{c(i)c(j)}(s_1+s_2-2s_1s_2)\\&\quad+w_{ij}p_{c(i)c(j)}(s_1+s_2-2p_{c(i)c(j)}s_1s_2)]\\
	\le& 2\displaystyle{\sum_{(i,j)\in E_\pi}}w_{ij}p_{c(i)c(j)}(s_1+s_2)\\
	\le& 2\left(nk-\frac{k^2}{2}-k\right)\overline{w}\alpha(s_1+s_2).\\
	\end{align*}
	
	Therefore, by Lemma \ref{lemma:upperbound}, $Pr\{X_{\pi}-Y_{\pi}\le 0\}$ can be upper bounded by 
	\begin{small}
		\begin{align}
		& Pr\{X_{\pi}-Y_{\pi}\le 0\}\le 2\exp\left[-(\lambda_X-\lambda_Y)^2/12(\lambda_X+\lambda_Y) \right]\nonumber\\[-2pt]
		\le& 2\exp\Bigg\{-\frac{k^2\left[4(\frac{2n-k+2}{2})\underline{w}\alpha(1-\beta)s_1s_2-\overline{w}\beta(s_1+s_2-2s_1s_2) \right]^2}{96(nk-\frac{k^2}{2}-k)\overline{w}\alpha(s_1+s_2)} \Bigg\} \label{eq:upperbound}\\[-2pt]
		\le&\exp\Bigg\{-\frac{k^2\left[(n-\frac{k}{2}-1)\underline{w}\alpha(1-\beta)s_1s_2\right]^2}{6(nk-\frac{k^2}{2}-k)\overline{w}\alpha(s_1+s_2)}\Bigg\}, \label{eq:upperbound1}
		\end{align}
	\end{small}
	where Inequality (\ref{eq:upperbound1}) follows from the conditions stated in the theorem.
	%
	%

	2. \textbf{Convergence of the Upper Bound: }Now, we further show that the derived upper bound converges to 0 as $n\rightarrow \infty$. Due to the monotonicity of $w_{ij}$ with respect to $p_{c(i)c(j)}$,  we easily obtain that $\overline{w}=\log\left( \frac{1-\alpha(s_1+s_2-2s_1s_2)}{\alpha(1-s_1)(1-s_2)}\right) $ and $\underline{w}=\log\left( \frac{1-\beta(s_1+s_2-2s_1s_2)}{\beta(1-s_1)(1-s_2)}\right) $. Hence, $\overline{w}$ and $\underline{w}$ can be determined by $\alpha,\beta,s_1,s_2$.   
	
	Plugging Inequality (\ref{eq:upperbound1}) into  Inequality (\ref{eq:expectation}), we have 
	\begin{small}
		\begin{align*}
		\mathbb{E}[S]&\le 2\sum_{k=2}^n n^k\cdot \exp\left(-\frac{k^2\left[(n-\frac{k}{2}-1)\underline{w}\alpha(1-\beta)s_1s_2\right]^2}{6(nk-\frac{k^2}{2}-k)\overline{w}\alpha(s_1+s_2)}\right) \\
		&\le \sum_{k=2}^{\infty}\exp \Bigg\{ k\left(-\frac{\left[(n-\frac{k}{2}-1)\underline{w}\alpha(1-\beta)s_1s_2\right]^2}{6(n-\frac{k}{2}-1)\overline{w}\alpha(s_1+s_2)}+\log n\right)  \Bigg\}\\
		&\le \sum_{k=2}^{\infty}\exp \Bigg\{ k\left(-\frac{\left[(n-\frac{k}{2}-1)\underline{w}^2\alpha^2(1-\beta)^2s_1^2s_2^2\right]}{6\overline{w}\alpha(s_1+s_2)}+\log n\right)  \Bigg\}
		\end{align*}
	\end{small}
	Since $\alpha,\beta\rightarrow 0,\frac{\log \alpha}{\log \beta}\le \gamma$, we also have $\frac{ \overline{w}}{ \underline{w}}\le \gamma'=\Theta(\gamma)$ and $\overline{w}=\Theta(\log \frac{1}{\alpha})$ where $\gamma'$ may be a function of $\gamma$.
	Hence, we have for some constant $C$,
	
	\begin{align*}
	\mathbb{E}[S]&\le \sum_{k=2}^{\infty}\exp\Bigg\{ k\left(-\frac{\left[Cn\alpha^2(1-\beta)^2s_1^2s_2^2\log\frac{1}{\alpha}\right]}{\gamma'^2\alpha(s_1+s_2)}+\log n\right)  \Bigg\}.
	\end{align*}
	
	Therefore, if $\frac{\alpha(1-\beta)^2s_1^2s_2^2\log(1/\alpha)}{s_1+s_2}=\Omega({\frac{\gamma\log^2 n}{n}})+\omega(\frac{1}{n})$, the sum of the above geometric series goes to zero as $n$ goes to infinity.
	Therefore, $\mathbb{E}[S]\rightarrow 0$. Hence, with the above conditions in Theorem \ref{theorem:MAP} satisfied, the MAP estimate $\hat{\pi}$ coincides with the correct mapping $\pi_0$ with probability goes to 1 as $n$ goes to infinity.
	\vspace{-2mm}
\section{\text{Superiority of Our Cost Function}}\label{app:Superiority}
In this section, we compare our cost functions over previous ones proposed in the literature. Specifically, we demonstrate the superiority of our cost function in bilateral case over the most similar previous cost function proposed by Pedarsani et al. \cite{cite:seedless}. Recall that the cost function derived in \cite{cite:seedless}, which we denoted as $\Delta'_{\pi}$, is
 \[
 \Delta'_{\pi}=\sum_{i\le j}^{n}\left|\mathbbm{1}\{(i,j)\in E_1\}-\mathbbm{1}\{(\pi(i),\pi(j))\in E_2\}\right|.
 \]
 
 The advantages of our cost function is two-fold. First, $\Delta'_{\pi}$, as an unweighted version of our proposed $\Delta_{\pi}$, corresponds to the MAP estimator in bilateral de-anonymization when the underlying social network is an Erd\H{o}s-R\'{e}nyi graph. Therefore, our cost function in a sense, subsumes the cost function in \cite{cite:seedless} as a special case in bilateral de-anonymization, and has more generality when the underlying network is non-uniform or the adversary only possesses unilateral community information. Second, we show that in certain cases, the correct mapping $\pi_0$ is the unique minimizer of $\Delta_{\pi}$, while it is not the unique minimizer of $\Delta'_{\pi}$. Indeed, when the underlying social network is as shown in Figure \ref{fig:append2_1}, and the sampling probabilities $s_1=s_2\le \frac{\gamma'}{2}$, with $\gamma'$ defined as in the proof of Theorem \ref{theorem:MAP}, we have that the unique minimizer of $\Delta_{\pi}$ asymptotically almost surely coincides with $\pi_0$ by Theorem \ref{theorem:MAP}. However, as $\Delta'_{\pi}$ does not count the weight of node pairs, in each realization of $G_1$ and $G_2$, there exists a mapping $\pi'$ that permutes $\pi_0=\arg\min_{\pi\in\Pi}\Delta_{\pi}$ on some nodes in $C_3$ with $\Delta'_{\pi'}\le\Delta_{\pi_0}$. Therefore, in this case, the minimizer of $\Delta'_{\pi}$ does not equals to $\pi_0$, which demonstrates that $\Delta_{\pi}$ has wider application.

 \begin{figure}
\centering
\includegraphics[width=0.95\linewidth]{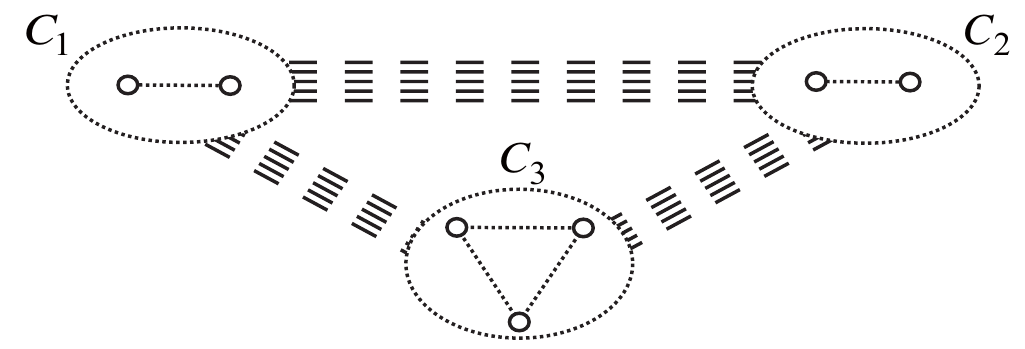}
\vspace{-4mm}
\caption{\small\bf An example demonstrating the superiority of our cost function: the sizes of the communities $|C_1|,|C_2|$ equal to some constant $C$ and $|C_3|=n-2C$, the affinity values $p_{11}=p_{22}=p_{33}=p_{12}=p_{23}={5\log n}/{n}$, $p_{13}={\log n}/{\sqrt{n}}$, the sampling probabilities $s_1=s_2={2}/{3}$}
\label{fig:append2_1}
\vspace{-8mm}
\end{figure}

\vspace{-4mm}
\section{Upper Bound of Inequality (19)}\label{app:upperbound}
	To present the upper bound of Inequality (\ref{eq:matrixbound}), we begin with bounding $\|\mathbf{D}^{-1}\|$ and $\|\mathbf{N}\|$.
	First, by the special block-diagonal structure of $\mathbf{D}$, we readily have that $\mathbf{D}^{-1}$ is also block diagonal with each $n\times n$ diagonal block as $\mathbf{D}_i^{-1}$, which is the identity matrix with the $i$th row replaced by $\frac{1}{\mathbf{v}_i}(-\mathbf{v}_1,\ldots,-\mathbf{v}_{i-1},1,-\mathbf{v}_{i+1},\ldots,-\mathbf{v}_n)$. We have
	\begin{align*}
	\|\mathbf{D}_i^{-1}\|\le 1+\frac{\sqrt{n}\epsilon_1}{\epsilon_2}, \quad\text{for all }i. 
	\end{align*}
	Hence we have, 
	\begin{align}
	\|\mathbf{D}^{-1}\|\le\max_{i=1\ldots n}\|\mathbf{D}_i^{-1}\|\le 1+\frac{\sqrt{n}\epsilon_1}{\epsilon_2}\label{eq:upperboundD}.
	\end{align}
	Similarly, we obtain 
	\begin{align*}
	\|\mathbf{N}\|^2&\le \max_{i,j=1\ldots n}\|\mathbf{N}_i\|_{\mathrm{F}}^2=\max_{i=1\ldots n}\sum_{jk}(s_{jk}^i+t_{jk}^i+w_{jk}^i+\frac{r_{ij}}{n})^2\\
	&\le 4\left(\max_{i=1\ldots n}\sum_k \left(s_{jk}^i\right)^2+\max_{i,j=1\ldots n}\sum_{jk} \left(t_{jk}^i\right)^2 \right.\\&\left.\quad+\max_{i,j=1\ldots n}\sum_k \left(w_{jk}^i\right)^2 + \max_{i,j=1\ldots n}\sum_k \left(\frac{r_{ij}}{n}\right)^2\right)^2.
	\end{align*}
	Next, we bound these the terms $s_{jk}^i$, $t_{jk}^i$, $w_{jk}^i$ and $r_{ij}$ one by one in the following inequalities.
	\begin{align*}
	&\max_{i=1\ldots n}\sum_{jk} \left(s_{jk}^i\right)^2\\=&\max_{i=1\ldots n}\sum_k\frac{1}{(\lambda_i-\lambda_j)^4}\\&\qquad\qquad\cdot\left(\mathbf{E}_{kj}(\lambda_j+\lambda_k-2\lambda_i)-\frac{v_j}{v_i}\mathbf{E}_{ki}(\lambda_k-\lambda_i)\right)^2
	\\\le &\max_{i=1\ldots n}\frac{1}{\delta^4}\left(4\sigma\sum_{jk}|\mathbf{E}_{kj}|+2\sigma\frac{\epsilon_1}{\epsilon_2}\sum_{kj}|\mathbf{E}_{kj}|\right)^2
	\\\le&\frac{4\sigma^2}{\delta^4}\left(1+2\frac{\epsilon_1}{\epsilon_2}\right)^2\xi^2
	\end{align*}
	\begin{align*}
	&\max_{i=1\ldots n}\sum_{jk} \left(t_{jk}^i\right)^2
	\\=&\max_{i=1\ldots n}\sum_{jk}\frac{1}{(\lambda_i-\lambda_j)^4}\left(\mathbf{G}_{kj}-\frac{v_j}{v_i}\mathbf{G}_{ki} \right)^2
	\\\le&\sum_{jk}\frac{1}{\delta^4}\left(\mathbf{G}_{kj}+2\frac{\epsilon_1}{\epsilon_2}\mathbf{G}_{ki}\right)^2 
	\\\le&\frac{1}{\delta^4}\left(1+2\frac{\epsilon_1}{\epsilon_2}\right)^2\|\mathbf{G}\|_\mathrm{F}^2
	\\\le&\frac{1}{\delta^4}\left(1+2\frac{\epsilon_1}{\epsilon_2}\right)^2\xi^4.
	\end{align*}
	\begin{align*}
	&\max_{i=1\ldots n}\sum_{jk} \left(w_{jk}^i\right)^2
	\\=&\max_{i=1\ldots n}\sum_{jk}\frac{\mu^2}{(\lambda_i-\lambda_j)^4}\left(\mathbf{M}'_{kj}-\frac{\mathbf{v}_j}{\mathbf{v}_i}\mathbf{M}'_{ki} \right)^2
	\\\le&\frac{\mu^2}{\delta^4}\max_{i=1\ldots n}\sum_{jk}\left(\sum_k\mathbf{M}'_{kj}+2\frac{\epsilon_1}{\epsilon_2}\sum_k\mathbf{M}'_{ki}\right)^2
	\\\le&\frac{\mu^2}{\delta^4}\left(1+2\frac{\epsilon_1}{\epsilon_2}\right)^2\|\mathbf{M'}\|_{\mathrm{F}}^2
	\\\le&\frac{\mu^2}{\delta^4}\left(1+2\frac{\epsilon_1}{\epsilon_2}\right)M^2.\quad\mbox{(by the orthonomality of $\mathbf{U}$)}
	\end{align*}
	\begin{align*}
	&\max_{i,j=1\ldots n}\sum_{k} \left(\frac{r_{ij}}{n}\right)^2
	\\=&\max_{i,j=1\ldots n}\frac{\mu^2}{n(\lambda_i-\lambda_j)^4}\left(\mathbf{v}_j\mathbf{M}'_{ii}-\mathbf{v}_i\mathbf{M}'_{ij} \right)^2
	\\\le&\frac{4\epsilon_1^2\mu^2}{n\delta^4}\|\mathbf{M'}\|_\mathrm{F}^2
	\\\le&\frac{4\epsilon_1^2\mu^2}{n\delta^4}{M}^2. 
	\end{align*}
	
	From the above manipulations, we have 
	\begin{small}
		\begin{align}
		\|\mathbf{N}\|^2&\le  4\left[\left(1+2\frac{\epsilon_1}{\epsilon_2}\right)^2\left(\frac{\sigma^2}{\delta^4}\xi^2+\frac{1}{\delta^4}\xi^4+\frac{\mu^2}{\delta^4}M^2\right) +\frac{4\epsilon_1^2\mu^2M^2}{n\delta^4} \right]\nonumber\\
		&\le 5\left[\left(1+2\frac{\epsilon_1}{\epsilon_2}\right)^2\left(\frac{\sigma^2}{\delta^4}\xi^2+\frac{1}{\delta^4}\xi^4+\frac{\mu^2}{\delta^4}M^2\right) \right]\label{eq:upperboundN},
		\end{align}
	\end{small}
	for sufficiently large $n$.
	Substituting Inequalities (\ref{eq:upperboundD}) and \ref{eq:upperboundN} into (\ref{eq:matrixbound}), it follows that
	\begin{align*}
	&\|\mathbf{F}-\mathbf{I}\|_{\mathrm{F}}=\|\mathbf{f-f}_0\|\le\\ &\sqrt{n}\frac{1-\left(1+\frac{\sqrt{n}\epsilon_1}{\epsilon_2}\right)\sqrt{\frac{5}{\delta^4}\left[\left(1+2\frac{\epsilon_1}{\epsilon_2}\right)^2\left({\sigma^2}\xi^2+\xi^4+{\mu^2}M^2\right) \right]}}{\left(1+\frac{\sqrt{n}\epsilon_1}{\epsilon_2}\right)\sqrt{\frac{5}{\delta^4}\left[\left(1+2\frac{\epsilon_1}{\epsilon_2}\right)^2\left({\sigma^2}\xi^2+\xi^4+{\mu^2}M^2\right) \right]}}.
	\end{align*}
\vspace{-5mm}
\section{MAP estimation of Unilateral De-anonymization}\label{app:MAP}
In this section, we derive the MAP estimator for unilateral de-anonymization. Recall that given $G_1$, $G_2$, $c$, $\bm{\theta}$, the MAP estimate $\hat{\pi}$ of the correct mapping $\mathbf{\pi}_0$ is defined as follows
\begin{align}
\hat{\pi}= \arg\max_{\pi\in \Pi}Pr(\pi_0=\pi\mid G_1,G_2,c,\bm{\theta}),\label{eq:mapestimator2}
\end{align}

The MAP estimator can be further written as:
\begin{align}
\hat{\pi}=\arg\max_{\pi\in\Pi}\sum_{G\in \mathcal{G}_\pi}p(G,\pi\mid G_1,G_2,c,\bm{\theta}),
\end{align}
where $\mathcal{G}_\pi$ is the set of all realizations of the underlying social network that are consistent with $G_1$, $G_2$ and $\pi$. By Bayesian rule, we have
\begin{align*}
&\arg\max_{\pi\in\Pi}\sum_{G\in\mathcal{G}_\pi}p(G,\pi\mid G_1,G_2,c,\bm{\theta})\\
=&\arg\max_{\pi\in\Pi}\sum_{G\in\mathcal{G}_\pi}\frac{p(G_1,G_2\mid G,\pi)p(G,\pi)}{p(G_1,G_2)}\\
=&\arg\max_{\pi\in\Pi}\sum_{G\in\mathcal{G}_\pi}p(G_1,G_2\mid G,\pi)p(G
)p(\pi)\\
=&\arg\max_{\pi\in\Pi}\sum_{G\in\mathcal{G}_\pi}p(G_1\mid G)p(G_2\mid G,\pi)p(G).
\end{align*}
Note that we drop parameters $c$ and $\bm{\theta}$ for brevity since their values are fixed. From the definitions of the models, we have:
\begin{small}
\begin{align*}
&\arg\max_{\pi\in\Pi}\sum_{G\in\mathcal{G}_\pi}p(G_1\mid G=g)p(G_2\mid G,\pi)p(G)\\
=&\arg\max_{\pi\in\Pi}\sum_{G\in\mathcal{G}_\pi}\prod_{i<j}^n(1-s_1)^{|E^{ij}|-|E_1^{ij}|}s_1^{|E_1^{ij}|}
\\&\cdot\prod_{i<j}^n(1-s_2)^{|E^{ij}|-|E_2^{\pi(i)\pi(j)}|}s_2^{|E_2^{\pi(i)\pi(j)}|}\cdot\prod_{i<j}^np_{c(i)c(j)}^{|E^{ij}|}(1-p_{c(i)c(j)})^{1-|E^{ij}|}
\\=&\arg\max_{\pi\in\Pi}\left(\prod_{i<j}^n\left(\frac{s_1}{1-s_1}\right)^{|E_1^{ij}|}\left(\frac{s_2}{1-s_2}\right)^{|E_2^{\pi(i)\pi(j)}|} \right)\\&\cdot\left(\sum_{g\in\mathcal{G}_\pi}\prod_{i<j}^k\left( \frac{p_{c(i)c(j)}(1-s_1)(1-s_2)}{1-p_{c(i)c(j)}}\right)^{|E^{ij}|}  \right) 
\\=&\arg\max_{\pi\in\Pi}\prod_{i<j}^n\left(\frac{s_2}{1-s_2}\right)^{|E_2^{\pi(i)\pi(j)}|}\\&\cdot \sum_{g\in\mathcal{G}_\pi}\prod_{i<j}^k\left( \frac{p_{c(i)c(j)}(1-s_1)(1-s_2)}{1-p_{c(i)c(j)}}\right)^{|E^{ij}|}\\
=& \sum_{g\in\mathcal{G}_\pi}\prod_{i<j}^k\left( \frac{p_{c(i)c(j)}(1-s_1)(1-s_2)}{1-p_{c(i)c(j)}}\right)^{|E^{ij}|},
\end{align*}
\end{small}
where $|E^{ij}|,|E_1|^{ij},|E_2^{ij}|$ take value 0 or 1 indicating whether there exists an edge between nodes $i$ and $j$ in $G,G_1,G_2$ respectively.
Note that in the above manipulations, we frequently eliminate the terms that do not depend on $\pi$. Particularly, in the last step, although the term \small{$\left(\frac{s_2}{1-s_2}\right)^{|E_2^{\pi(i)\pi(j)}|}$}\normalsize\\ depends on $\pi$, the value of the whole product \small$\prod_{i<j}^n\left(\frac{s_2}{1-s_2}\right)^{|E_2^{\pi(i)\pi(j)}|}$\normalsize\\ is independent of $\pi$ itself since it is a bijective mapping.

\begin{figure}
\centering
\includegraphics[width=0.95\linewidth]{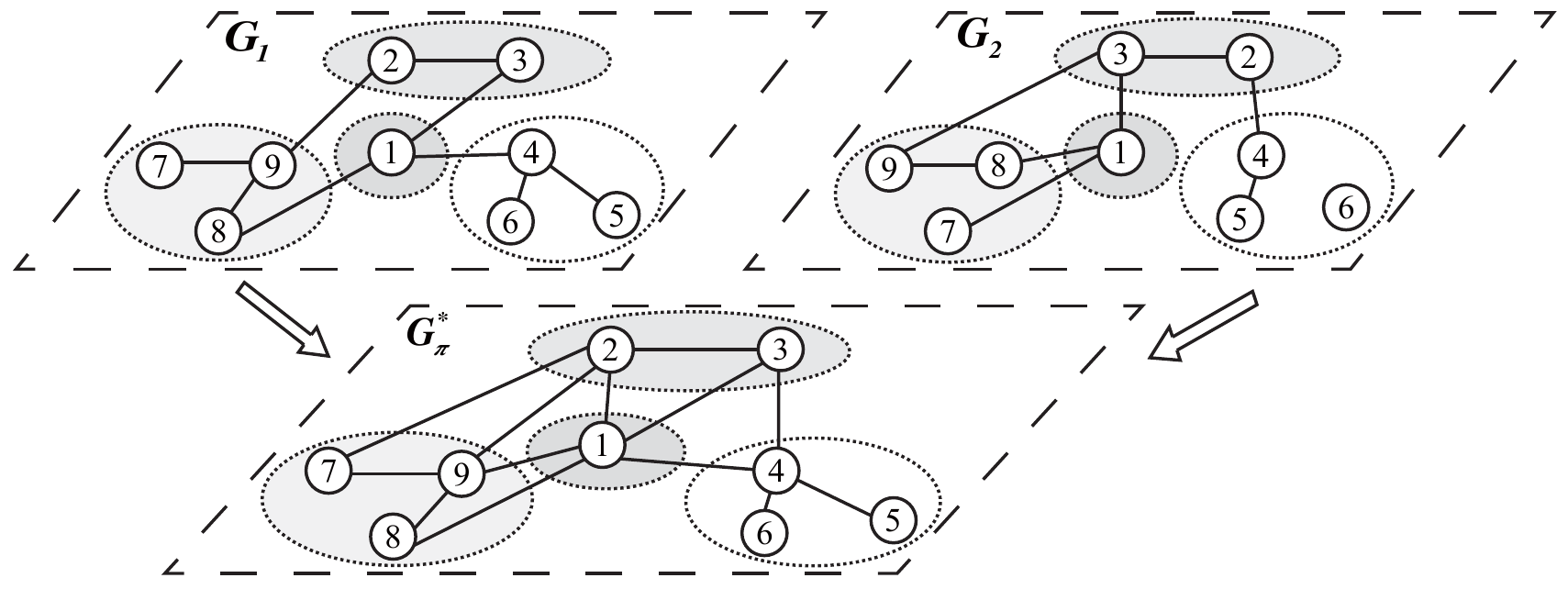}
\vspace{-4mm}
\caption{\small\bf An example of $G_{\pi}^*$ that has the minimum number of edges in $\mathcal{G}_{\pi}$, which is the set of all realizations of $G$ that are consistent with $G_1,G_2,\pi$. In this case $\pi=\pi_0$.}
\label{fig:append1}
\vspace{-7mm}
\end{figure}

\begin{figure*}[!tb]
	\centering
	\subfigure[]{
		\begin{minipage}[]{0.235\linewidth}
			\centering
			\vspace{-3mm}
			\includegraphics[width=1.0\linewidth]{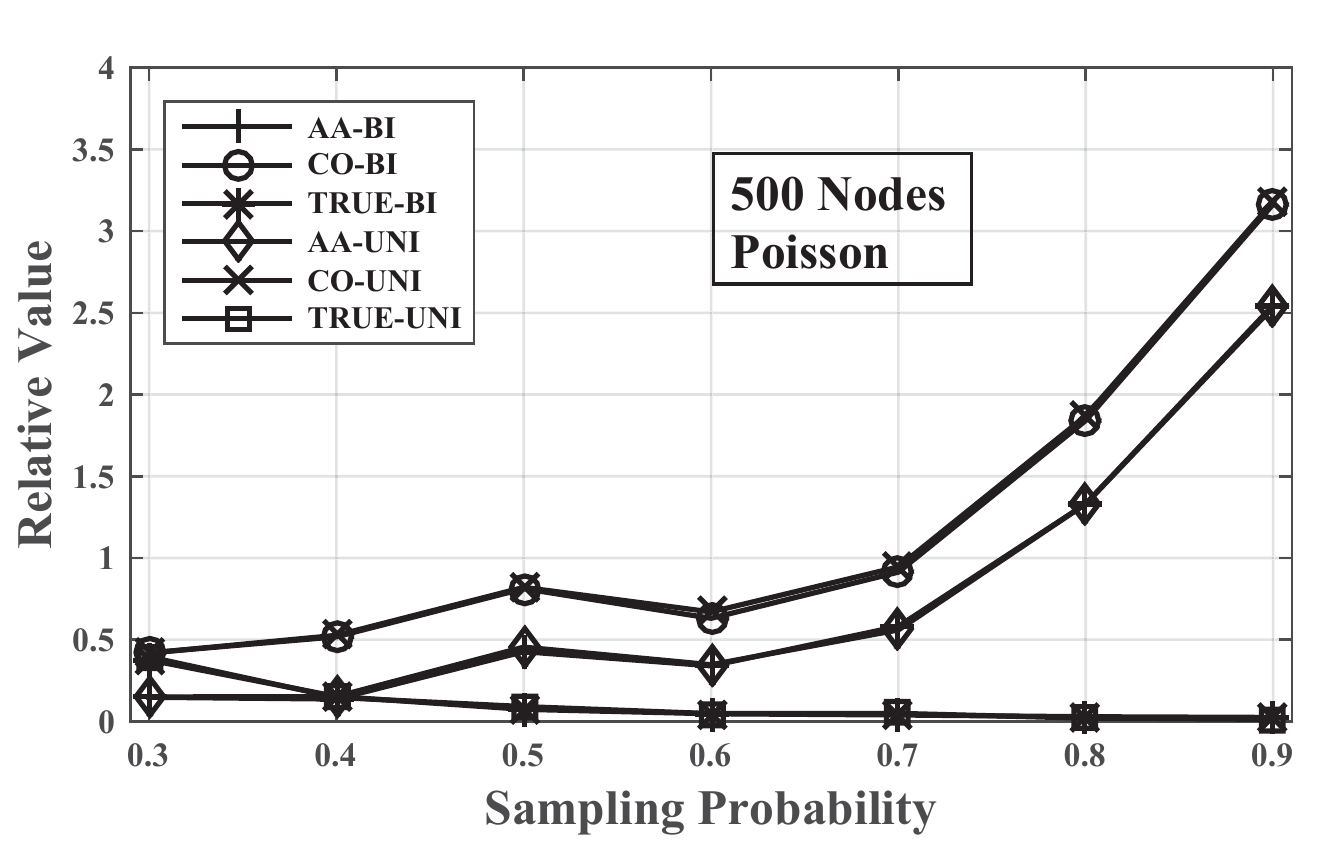}
			\vspace{-4mm}
			\label{fig:500poi1}
		\end{minipage}%
		
	}
	\subfigure[]{
		\begin{minipage}[]{0.235\linewidth}
			\centering
			\vspace{-3mm}
			\includegraphics[width=1.0\linewidth]{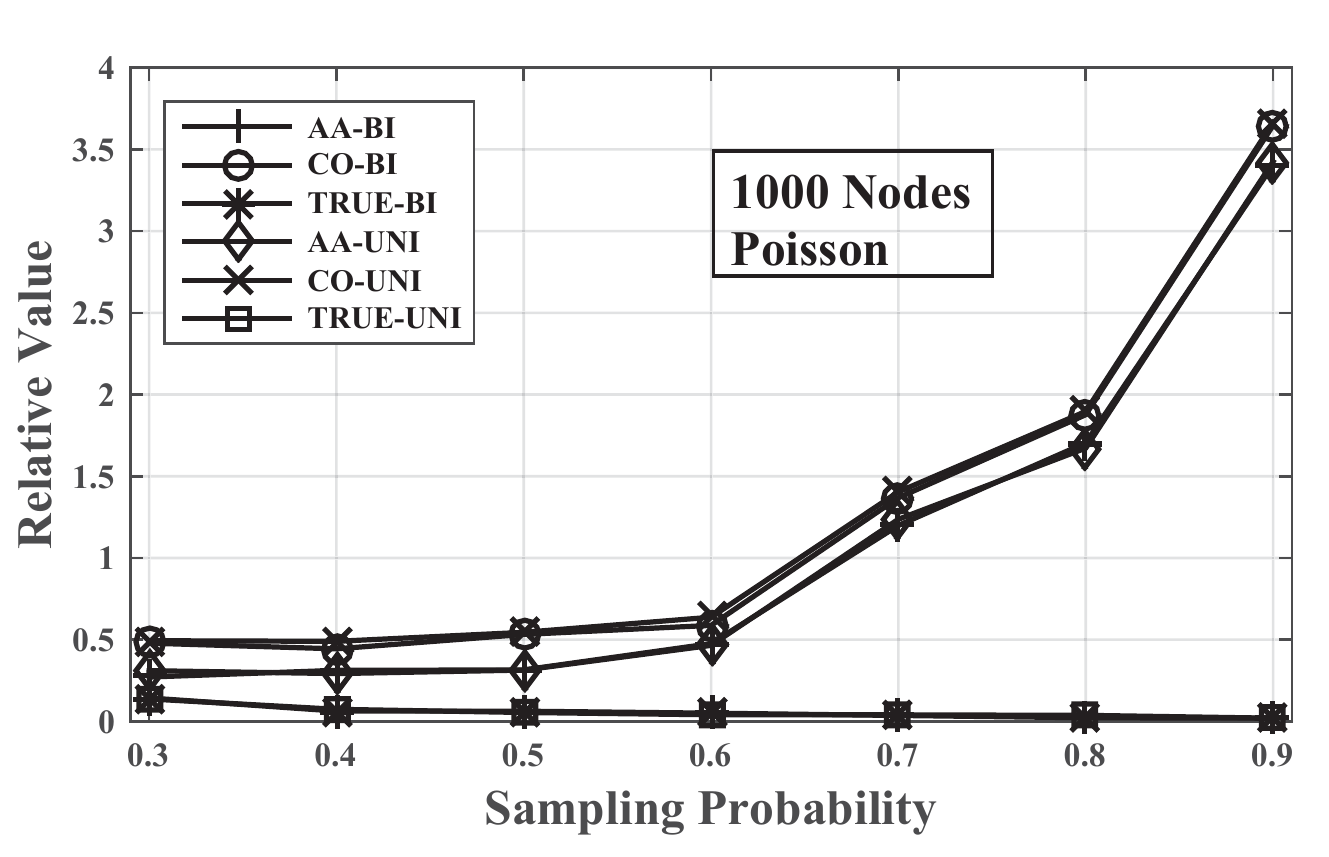}
			\vspace{-4mm}
			\label{fig:1000poi1}
		\end{minipage}%
		
	}
	\subfigure[]{
		\begin{minipage}[]{0.235\linewidth}
			\centering
			\vspace{-3mm}
			\includegraphics[width=1.0\linewidth]{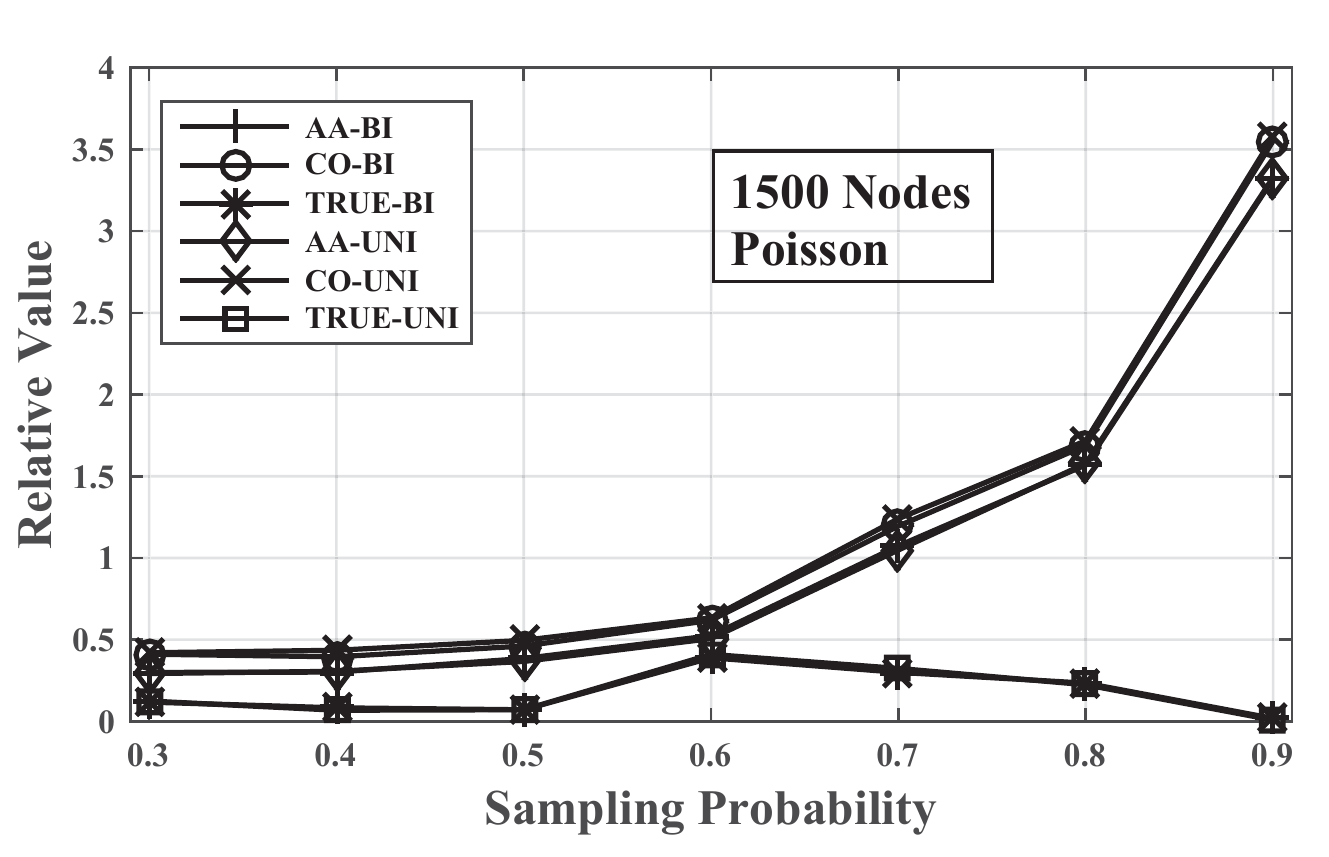}
			\vspace{-4mm}
			\label{fig:1500poi1}
		\end{minipage}%
		
	}
	\subfigure[]{
		\begin{minipage}[]{0.235\linewidth}
			\centering
			\vspace{-3mm}
			\includegraphics[width=1.0\linewidth]{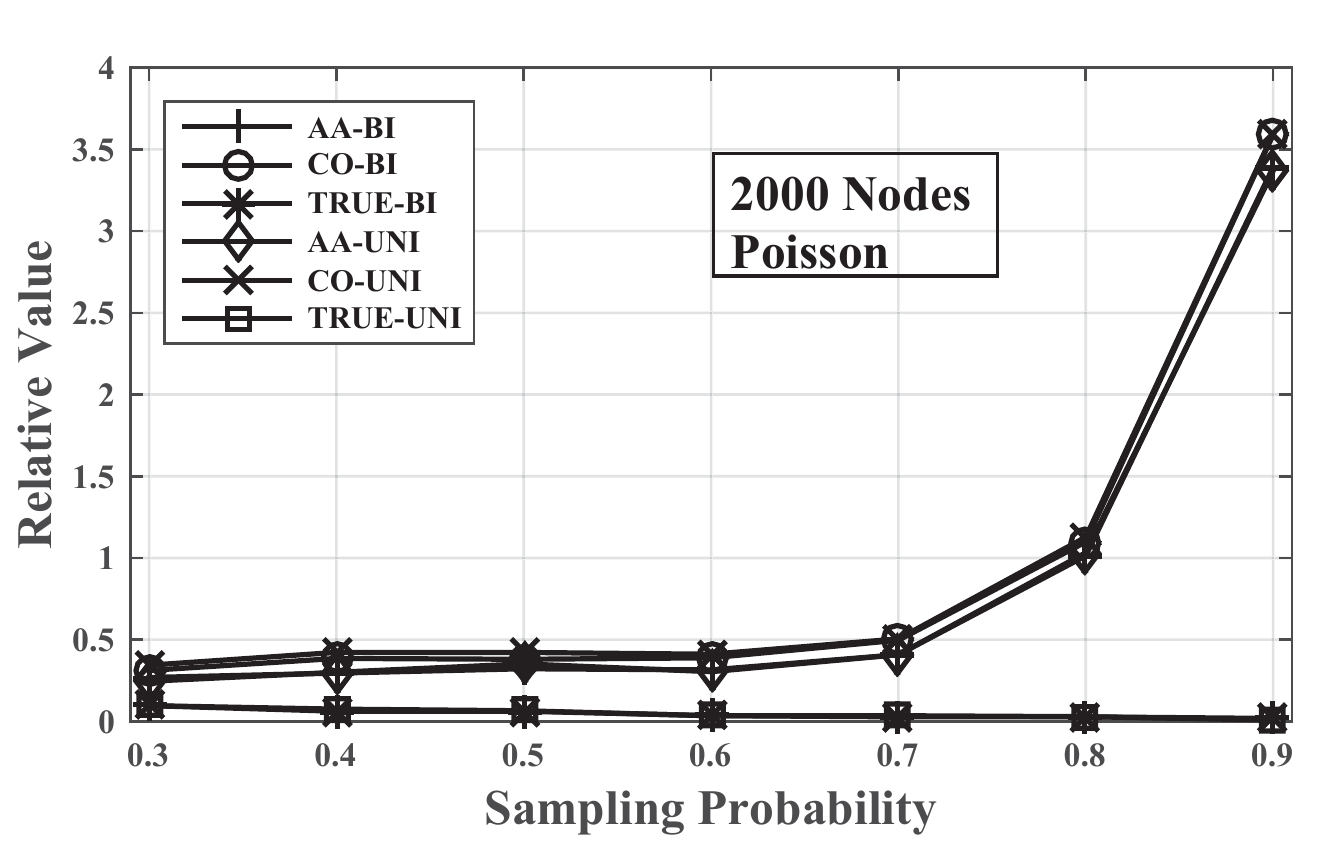}
			\vspace{-4mm}
			\label{fig:2000poi1}
		\end{minipage}%
		
	}

	\subfigure[]{
		\begin{minipage}[]{0.235\linewidth}
			\centering
			\vspace{-3mm}
			\includegraphics[width=1.0\linewidth]{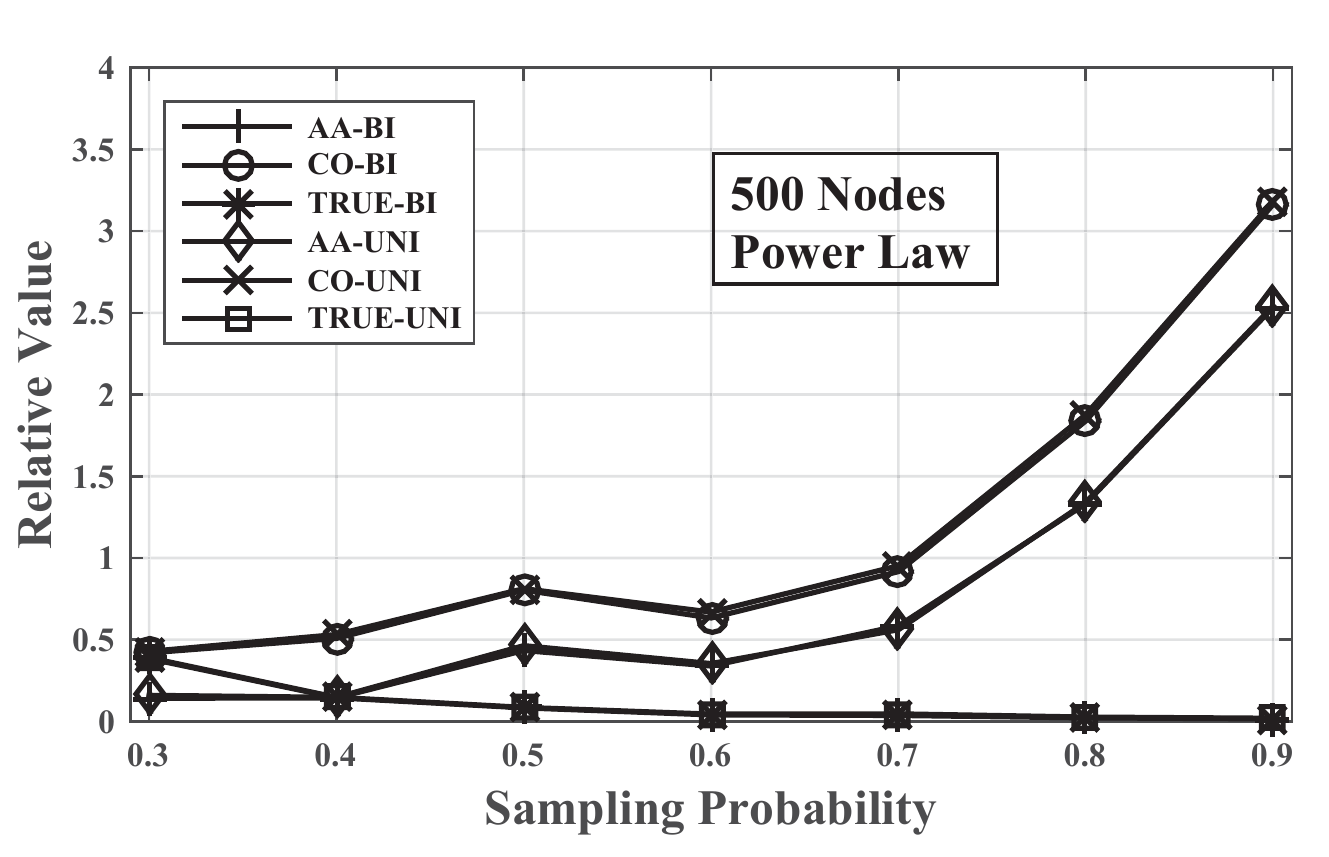}
			\vspace{-4mm}
			\label{fig:500pow1}
		\end{minipage}%
		
	}
	\subfigure[]{
		\begin{minipage}[]{0.235\linewidth}
			\centering
			\vspace{-3mm}
			\includegraphics[width=1.0\linewidth]{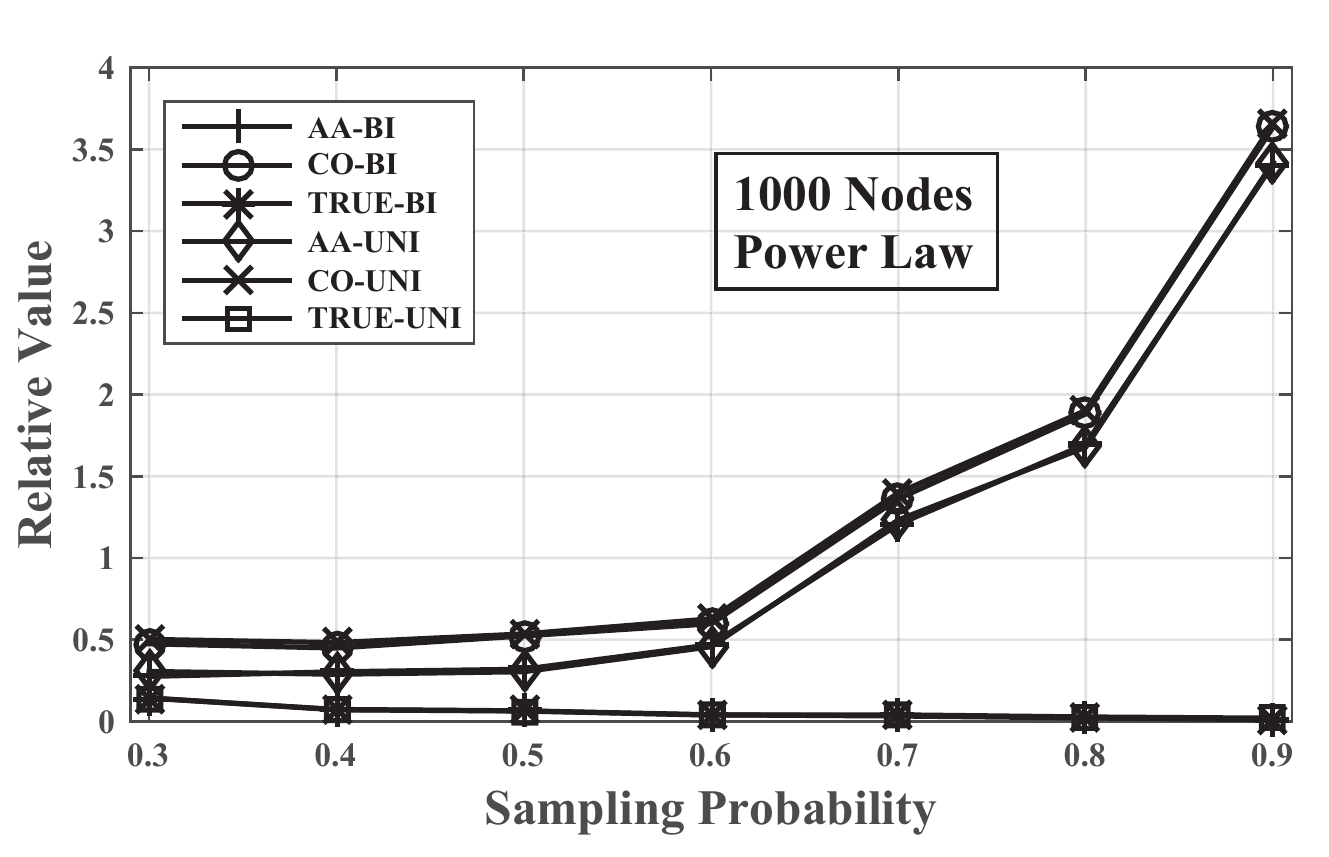}
			\vspace{-4mm}
			\label{fig:1000pow1}
		\end{minipage}%
		
	}
	\subfigure[]{
		\begin{minipage}[]{0.235\linewidth}
			\centering
			\vspace{-3mm}
			\includegraphics[width=1.0\linewidth]{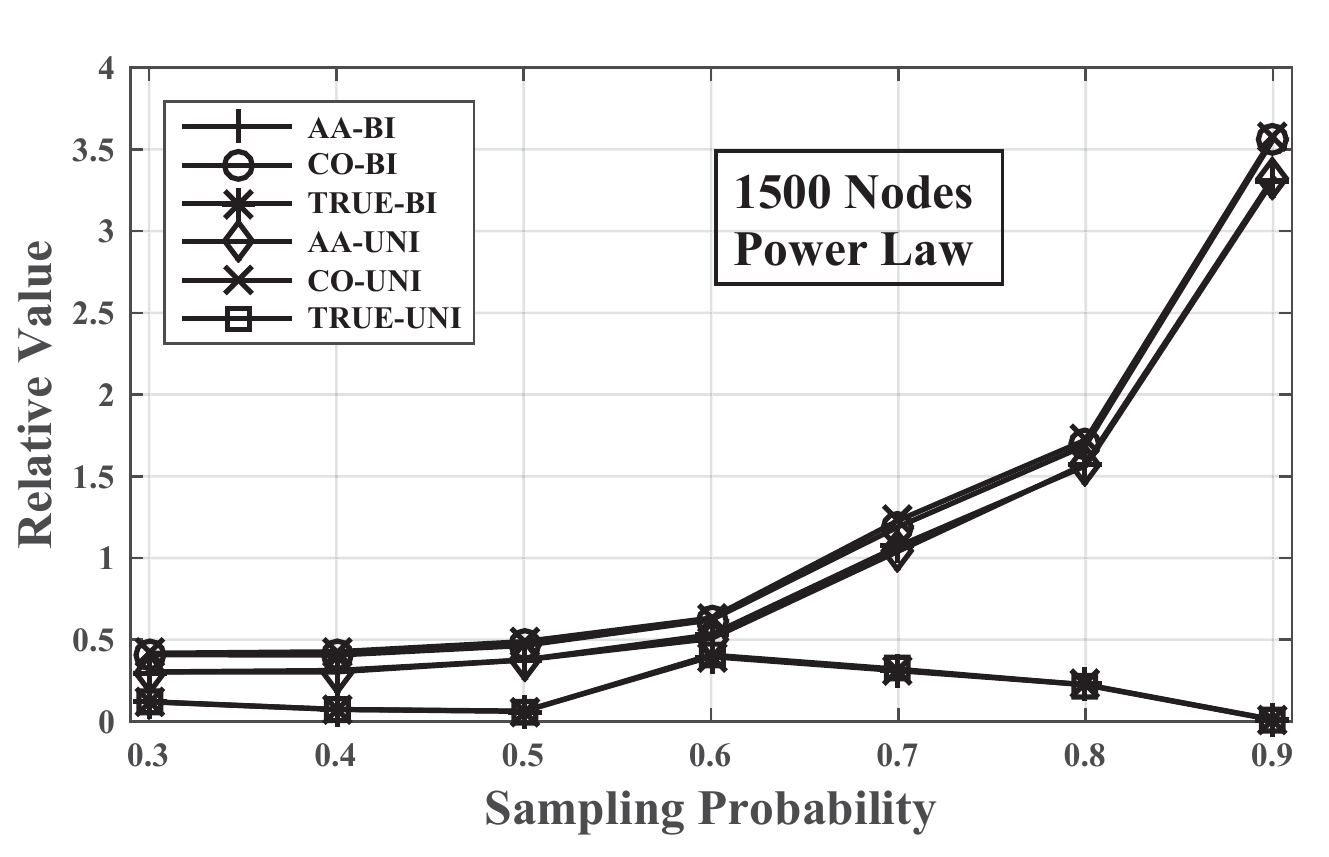}
			\vspace{-4mm}
			\label{fig:1500pow1}
		\end{minipage}%
		
	}
	\subfigure[]{
		\begin{minipage}[]{0.235\linewidth}
			\centering
			\vspace{-3mm}
			\includegraphics[width=1.0\linewidth]{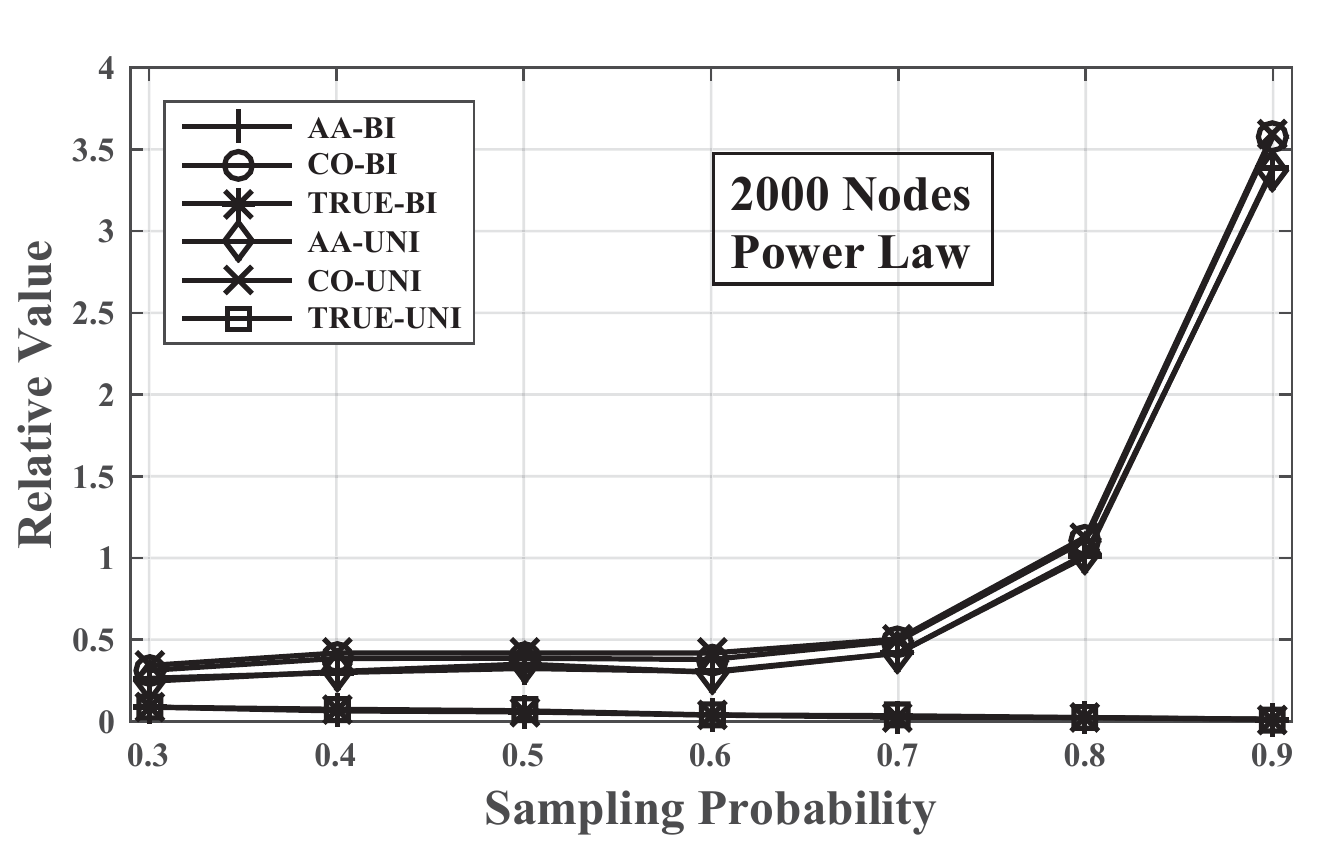}
			\vspace{-4mm}
			\label{fig:2000pow1}
		\end{minipage}%
		
	}

	\subfigure[]{
		\begin{minipage}[]{0.235\linewidth}
			\centering
			\vspace{-3mm}
			\includegraphics[width=1.0\linewidth]{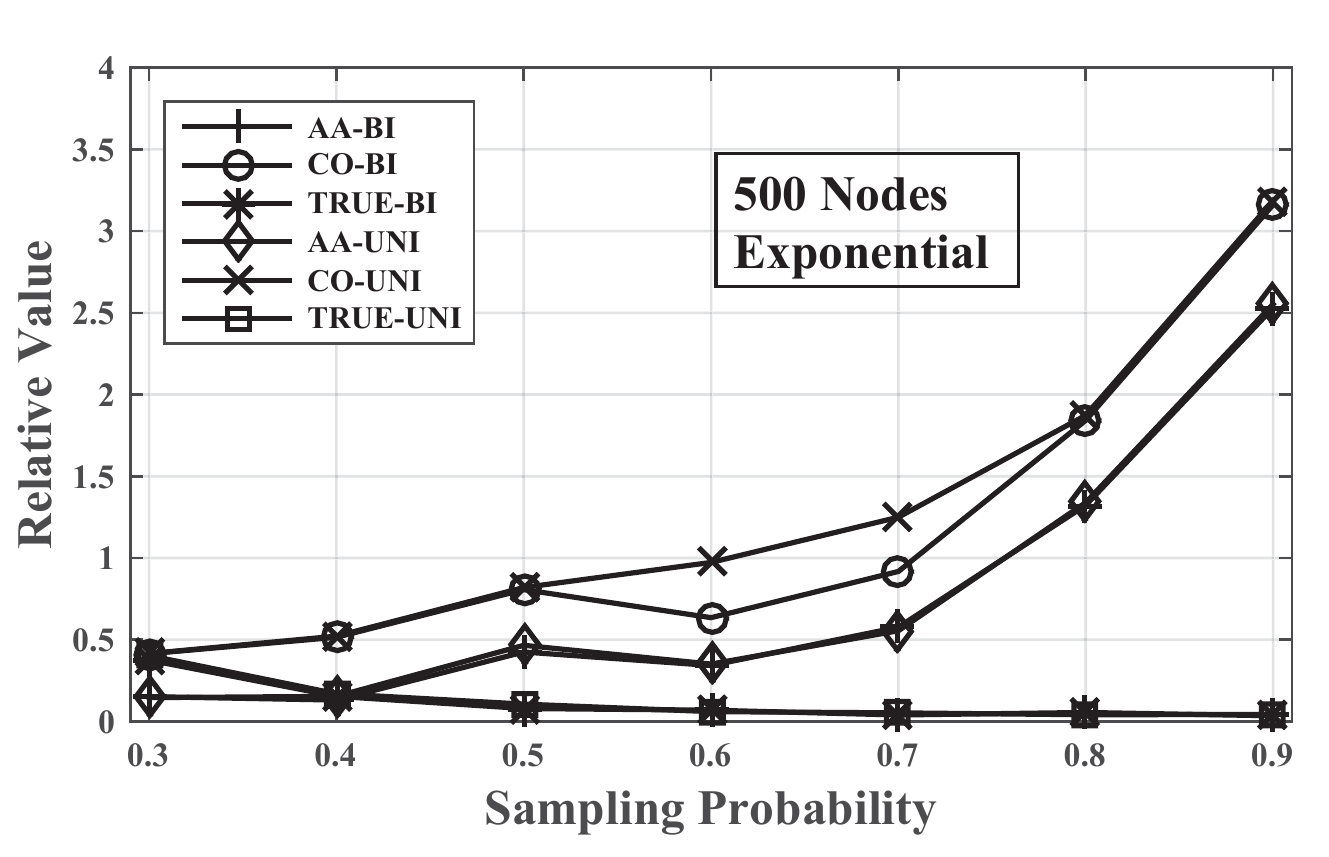}
			\vspace{-4mm}
			\label{fig:500exp1}
		\end{minipage}%
		
	}
	\subfigure[]{
		\begin{minipage}[]{0.235\linewidth}
			\centering
			\vspace{-3mm}
			\includegraphics[width=1.0\linewidth]{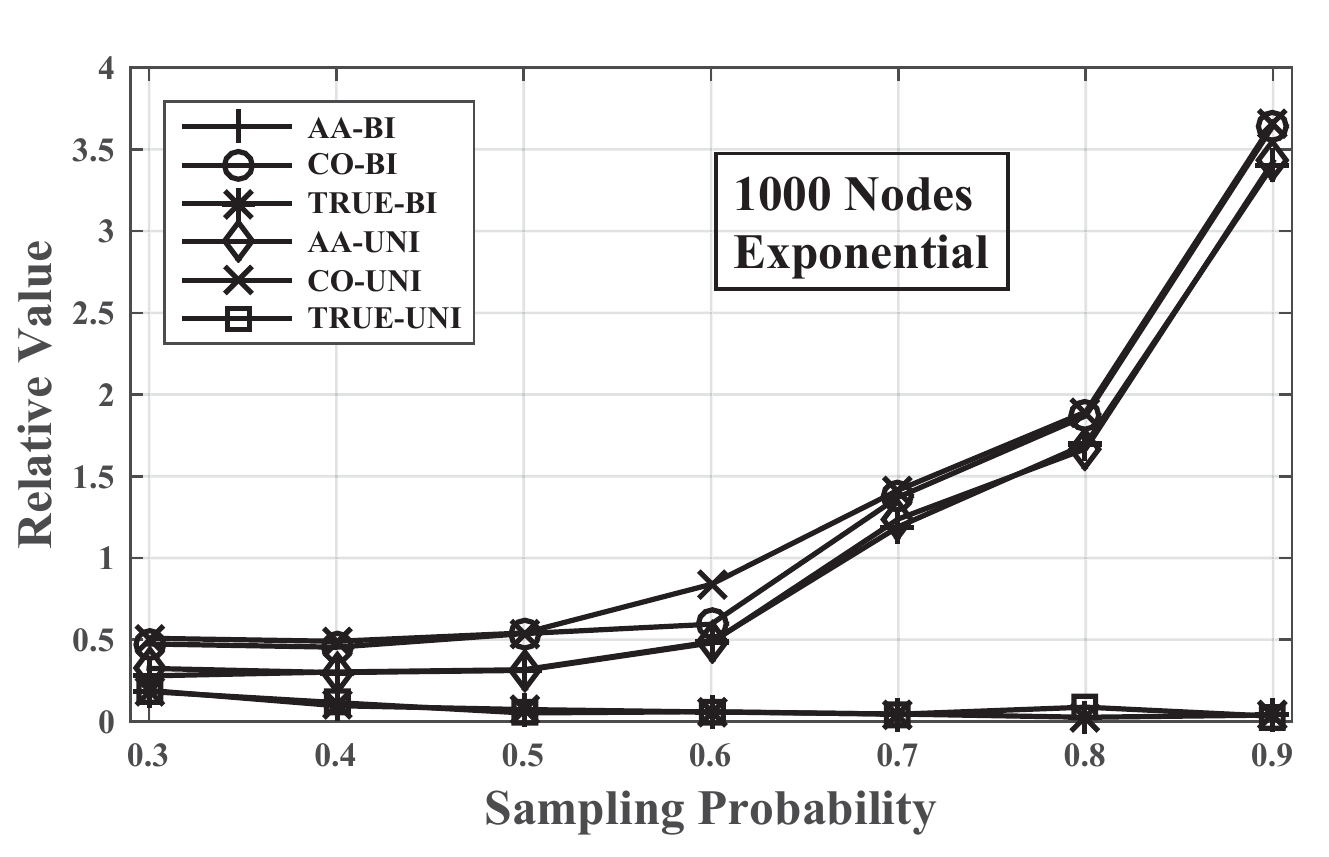}
			\vspace{-4mm}
			\label{fig:1000exp1}
		\end{minipage}%
		
	}
	\subfigure[]{
		\begin{minipage}[]{0.235\linewidth}
			\centering
			\vspace{-3mm}
			\includegraphics[width=1.0\linewidth]{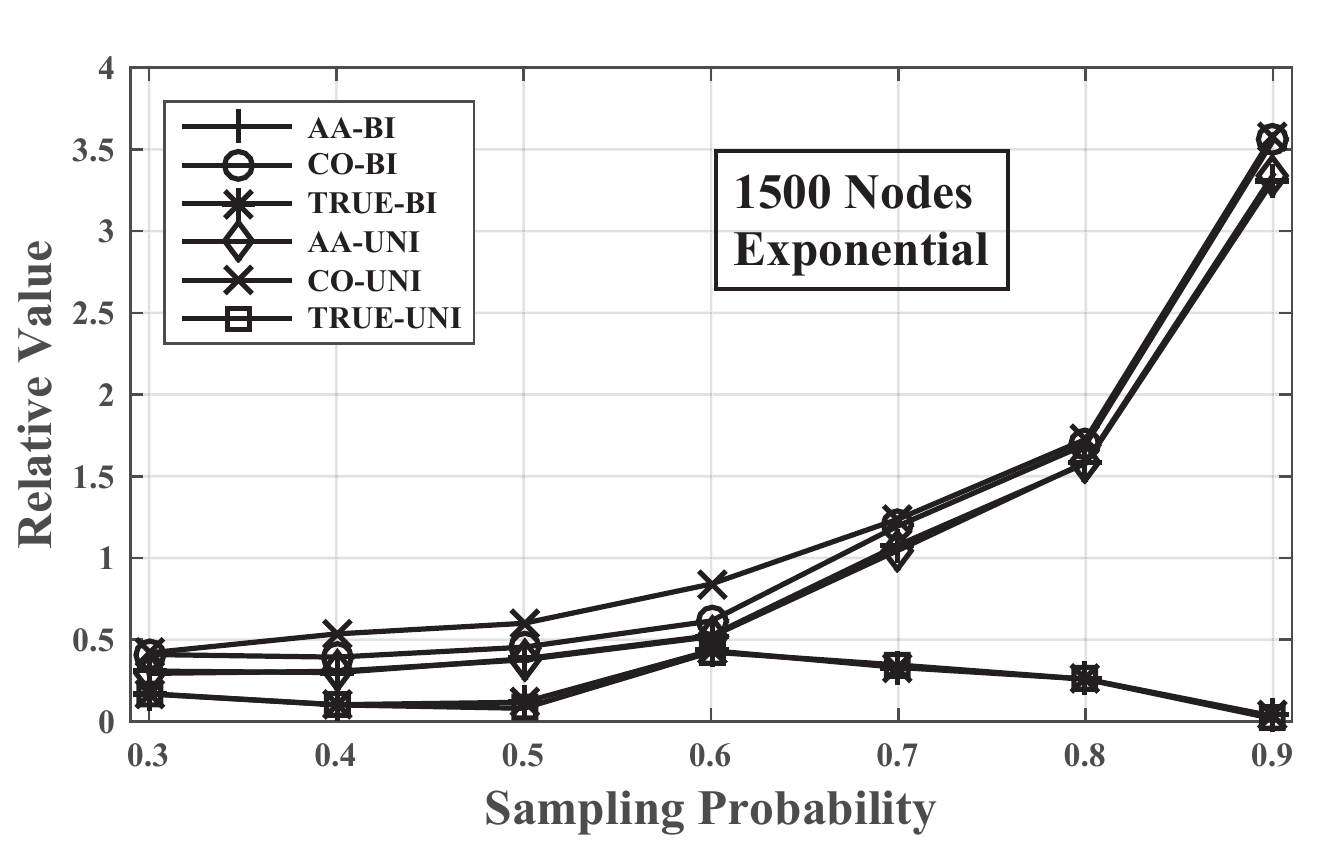}
			\vspace{-4mm}
			\label{fig:1500exp1}
		\end{minipage}%
		
	}
	\subfigure[]{
		\begin{minipage}[]{0.235\linewidth}
			\centering
			\vspace{-3mm}
			\includegraphics[width=1.0\linewidth]{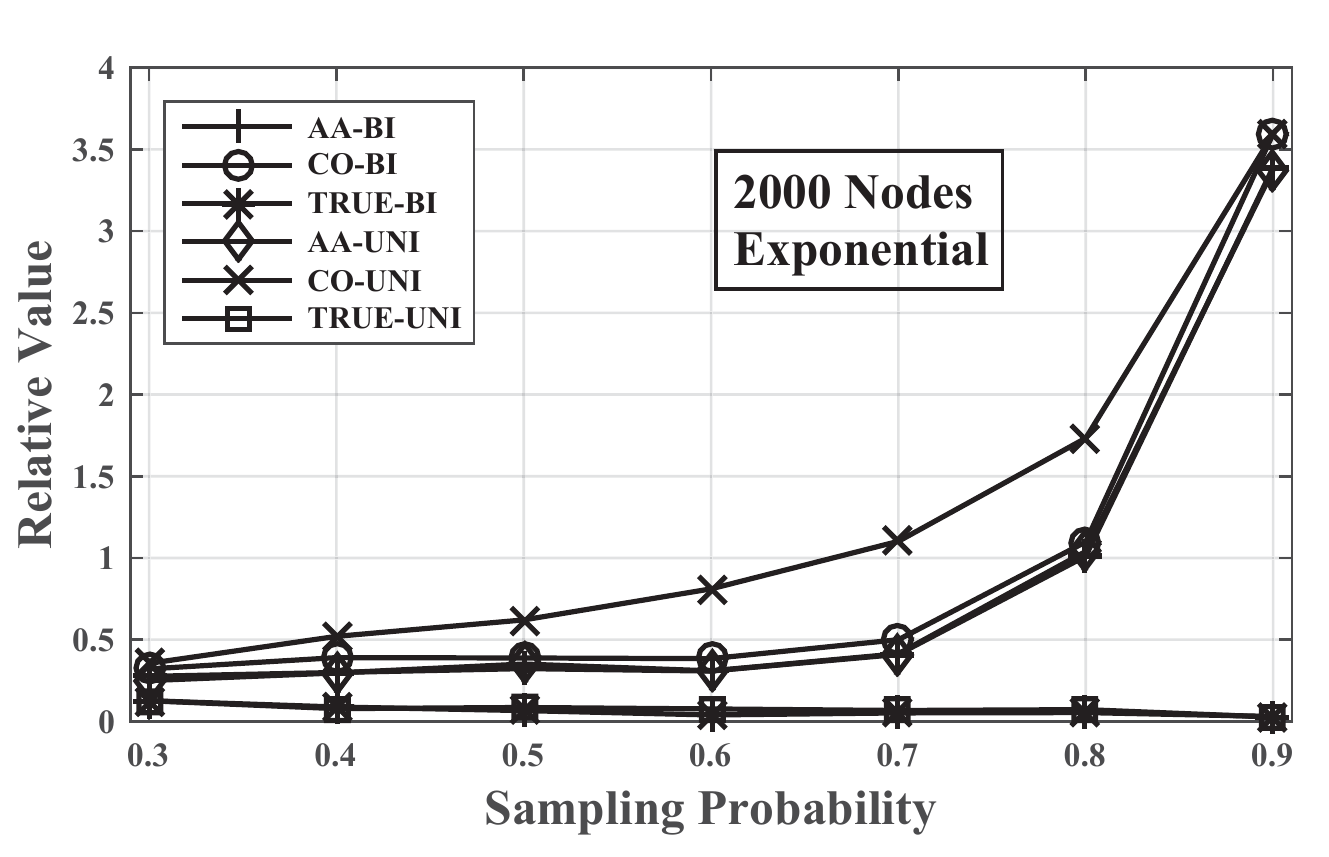}
			\vspace{-4mm}
			\label{fig:2000exp1}
		\end{minipage}%
		
	}
	
	\vspace{-5mm}
	\caption{The relative value of the cost function of the mappings produced by the algorithms on synthetic datasets with different degree distributions.}
	\label{fig:syn_cost}
	\vspace{-4mm}		
\end{figure*}

Now, let $G_{\pi}^*$ be the graph having the smallest number of edges in $\mathcal{G}_{\pi}$, which is equivalent to that $G_{\pi}^*=(V,E_1\cup\pi(E_1))$. An illustration of $G_{\pi}^*$ is provided in Figure \ref{fig:append1}. Denote the set of edges in $G_{\pi^*}$ as $E_{\pi^*}$, with $|E_{\pi^*}^{ij}|$ indicating the number of edges between $i$ and $j$. By the definition we have that in $\mathcal{G}_\pi$, all the graphs have edge sets that are supersets of $G_{\pi}^*$. By summing over all the graphs in $\mathcal{G}_\pi$, we have that
\begin{align*}
\hat{\pi}&=\arg\max_{\pi\in\Pi} \prod_{i<j}^n\left(\frac{p_{c(i)c(j)}(1-s_1)(1-s_2)}{1-p_{c(i)c(j)}}\right)^{|E^{ij}_{\pi^*}|}\\
&\cdot\prod_{i<j}^n\left(1+\left(\frac{p_{c(i)c(j)}(1-s_1)(1-s_2)}{1-p_{c(i)c(j)}}\right)\right)^{1-|E^{ij}_{\pi^*}|},
\end{align*}
where the above equality follows from that 
\begin{align*}
&\sum_{g\in\mathcal{G}_\pi}\prod_{i<j}^n\left( \frac{p_{c(i)c(j)}(1-s_1)(1-s_2)}{1-p_{c(i)c(j)}}\right)^{|E^{ij}|-|E^{ij}_{\pi^*}|}\\
=&\sum_{0\le k_{ij}\le 1-|E_{\pi^*}^{ij}|}\prod_{i<j}^n\left( \frac{p_{c(i)c(j)}(1-s_1)(1-s_2)}{1-p_{c(i)c(j)}}\right)^{k_{ij}}\\
=&\prod_{i<j}^n\left(1+\left(\frac{p_{c(i)c(j)}(1-s_1)(1-s_2)}{1-p_{c(i)c(j)}}\right)\right)^{1-|E^{ij}_{\pi^*}|}.
\end{align*}
Then, from the above equation we can further write the MAP estimator as:
\begin{align*}
&\arg\max_{\pi\in\Pi}\prod_{i<j}^n\left(\frac{p_{c(i)c(j)}(1-s_1)(1-s_2)}{1-p_{c(i)c(j)}(s_1+s_2-s_1s_2)}\right)^{|E_{\pi^*}^{ij}|}\\
=&\arg\min_{\pi\in\Pi}\prod_{i<j}^n\left(\frac{1-p_{c(i)c(j)}(s_1+s_2-s_1s_2)}{p_{c(i)c(j)}(1-s_1)(1-s_2)}\right)^{|E_{\pi^*}^{ij}|}\\
=&\arg\min_{\pi\in\Pi}\left[|E_{\pi^*}^{ij}|\log\left(\frac{1-p_{c(i)c(j)}(s_1+s_2-s_1s_2)}{p_{c(i)c(j)}(1-s_1)(1-s_2)}\right)\right].
\end{align*}

Next, by the definition of $g_{\pi^*}$, we notice that
\[
|E^{ij}_{\pi^*}|=\lceil\frac{(|E_1^{ij}|+|E_2^{\pi(i)\pi(j)}|)}{2}\rceil.
\]
Hence, by setting $w_{ij}=\log\left(\frac{1-p_{c(i)c(j)}(s_1+s_2-s_1s_2)}{p_{c(i)c(j)}(1-s_1)(1-s_2)}\right)$  we have

\begin{align*}
\hat{\pi}=&\arg\min_{\pi\in\Pi}\left( \sum_{i<j}^{n}w_{ij}(\mathbbm{1}\{(i,j)\notin E_1,(\pi(i),\pi(j))\in E_2\}\right)
\end{align*}

Note that the MAP estimator is not symmetric with regard to $G_1$ and $G_2$. This stems from the fact that the adversary in this case only has knowledge on the community assignment function of $G_1$.

\section{Convexity of the Relaxed UNI-MAP-ESTIMATE}\label{app:convex}
In this section, we prove that the relaxed matrix formulation of the optimization problem UNI-MAP-ESTIMATE is convex. The relaxed formulation is presented as follows:
\begin{align}
\text{mininize }  \|\mathbf{W}\circ(\mathbf{\Pi}\mathbf{A}-&\mathbf{B}\mathbf{\Pi})\|_\mathrm{\lfloor F\rfloor}^2\nonumber\\
\text{\textbf{s.t. }}  \forall i,\ \sum_{i}\mathbf{\Pi}_{ij}&=1\label{P3:constraint1app}\\
\forall j,\ \sum_{j}\mathbf{\Pi}_{ij}&=1 \label{P3:constraint2app}
\end{align}
Obviously, the set of feasible solutions defined by Constraints (\ref{P3:constraint1app}) and (\ref{P3:constraint2app}) is a convex set. Then, for the objective function $\|\mathbf{W}\circ(\mathbf{\Pi}\mathbf{A}-\mathbf{B}\mathbf{\Pi})\|_\mathrm{\lfloor F\rfloor}^2$, according to the definition of operator $\|\cdot\|_{\lfloor\mathrm{F}\rfloor}$ it can be interpreted as weighted summation of truncated quadratic functions of each element of $\Pi$ with the weights being positive real numbers. Each truncated function is equivalent to the square of a linear function of an element of $\Pi$ with the part where the elements take positive values truncated. Therefore, each truncated function is convex. It follows that the whole objective function, being a weighted combination of convex functions, is convex. Thus, we conclude that the relaxed UNI-MAP-ESTIMATE is a convex optimization problem, the global optima of which can be found in $O(n^6)$ time using the same algorithm as in the bilateral case.

\section{Differences Between Bilateral and Unilateral De-anonymization}\label{sec:difference}
In this subsection, we summarize, from a higher level, the differences existing in the essence of bilateral and unilateral de-anonymization and the results we obtain for the two problems.
\begin{itemize}
	\item The extra knowledge on the community assignment function in bilateral de-anonymization enables us to restrict the feasible mappings to the ones that observe the community assignment, thus decreases the number of possible candidates and makes the problem intuitively easier than unilateral one.
	\item The community assignment as side information is the main reason behind the difference of the posterior distribution of the optimal mapping, which leads to different MAP estimates, and thus different cost functions in the two cases. Note that the cost function for bilateral de-anonymization cannot be calculated in unilateral case since we have no knowledge on the community assignment of $G_2$.
	\item Although under similar conditions, minimizing the cost function asymptotically almost surely recovers the correct mapping in both cases, the lack of community assignment in unilateral de-anonymization impose asymmetry in its cost function and render the cost function harder to (approximately) minimize, as justified by our stronger complexity-theoretic result.
	\item In terms of the proposed algorithms, the additive approximation algorithms for both bilateral and unilateral de-anonymization share the same guarantee. However, the convex optimization-based algorithm has been shown to yield conditionally yield optimal solutions only for bilateral de-anonymization.
	\item The empirical results demonstrate that in all the contexts, our algorithms successfully de-anonymize larger portion of users when provided with bilateral community information. 
\end{itemize}
\vspace{-2mm}
\section{Graphical Results on Relative Value of Cost Function}\label{sec:sim_figure}
In this section we present graphical results on the relative value of the cost function of the mappings produced by the algorithms. Recall that for a mapping $\pi$ and the mapping $\pi_{GA}$ produced by \textbf{GA} algorithm, the relative value of the cost function of $\pi$ equals to $(\Delta_{\pi}-\Delta_{\pi_{GA}})/\Delta_{\pi_{GA}}$.

\begin{figure}[htbp]
	\centering
	\subfigure[]{
		\begin{minipage}[]{0.48\linewidth}
			\centering
			\vspace{-4mm}
			\includegraphics[width=1.0\linewidth]{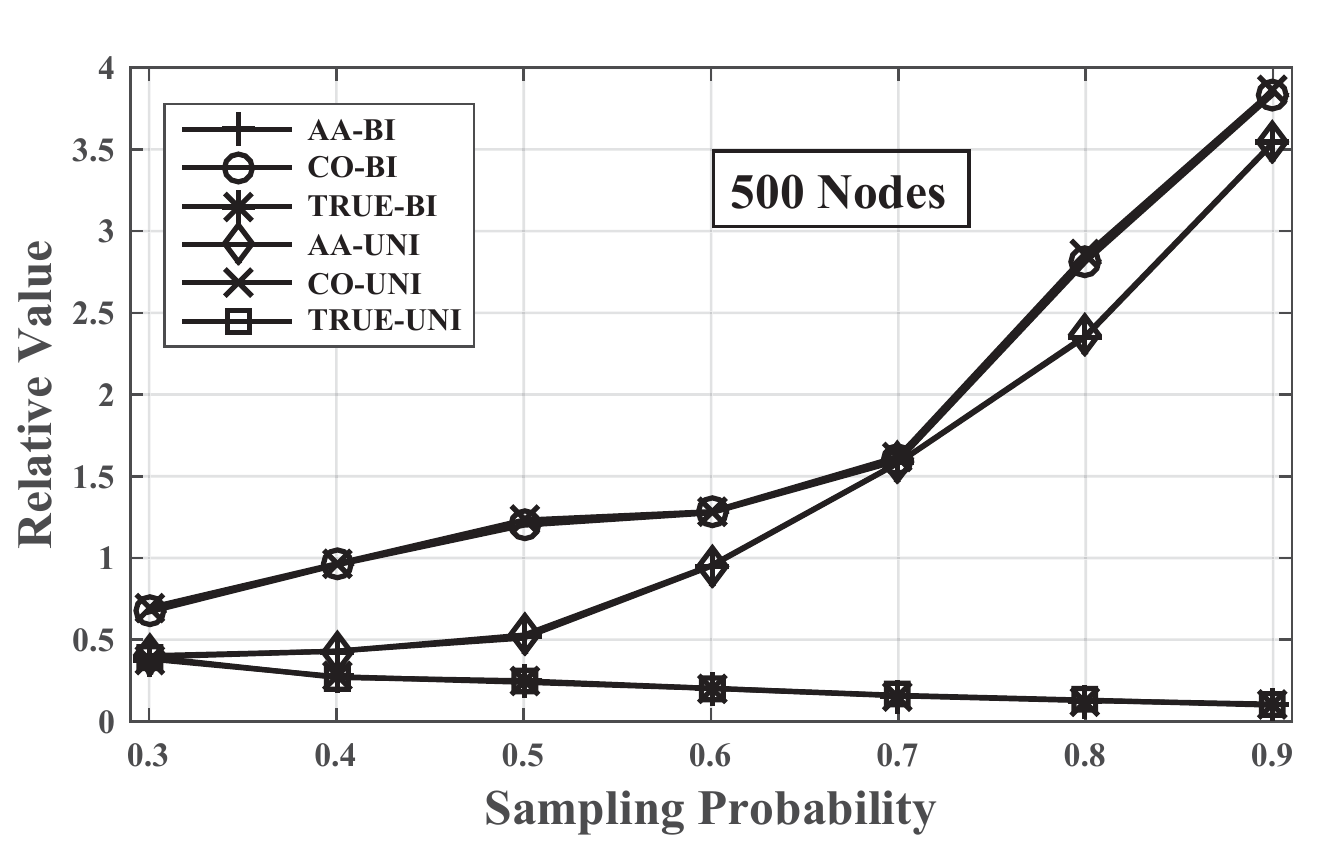}
			\vspace{-4mm}
			\label{fig:500sam1}
		\end{minipage}%
		
	}
	\subfigure[]{
		\begin{minipage}[]{0.48\linewidth}
			\centering
			\vspace{-4mm}
			\includegraphics[width=1.0\linewidth]{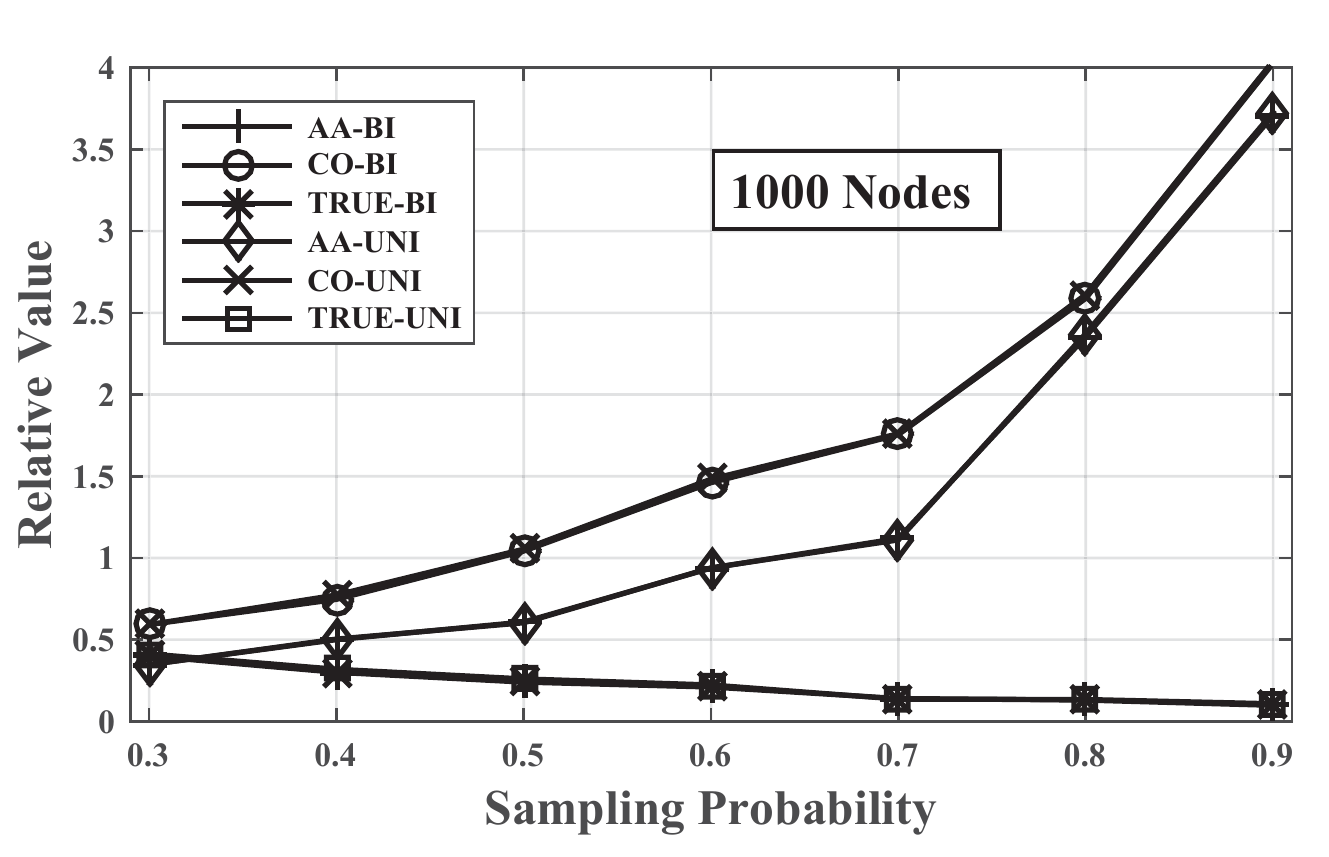}
			\vspace{-4mm}
			\label{fig:1000sam1}
		\end{minipage}%
		
	}
	\subfigure[]{
		\begin{minipage}[]{0.48\linewidth}
			\centering
			\vspace{-4mm}
			\includegraphics[width=0.98\linewidth]{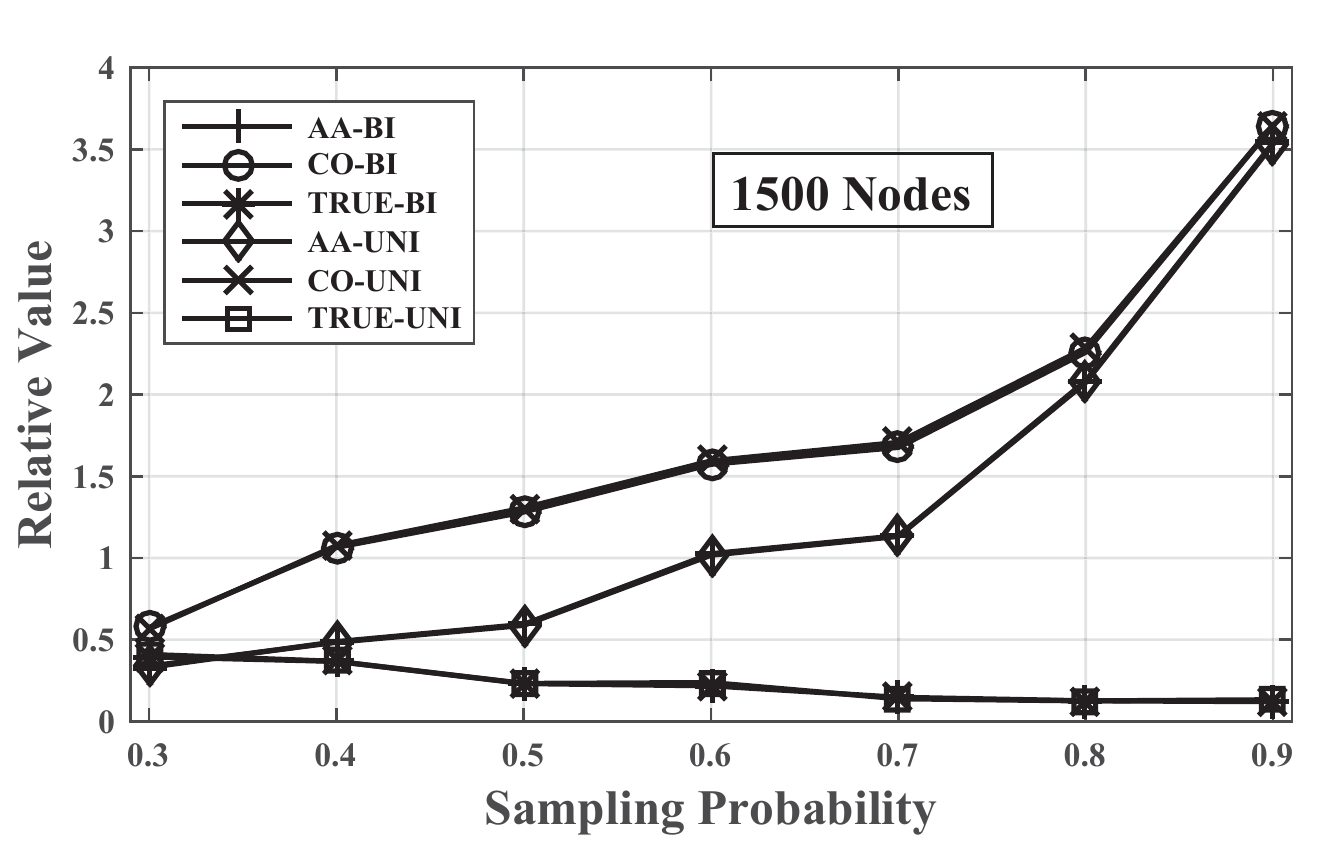}
			\vspace{-4mm}
			\label{fig:1500sam1}
		\end{minipage}%
		
	}
	\subfigure[]{
		\begin{minipage}[]{0.48\linewidth}
			\centering
			\vspace{-4mm}
			\includegraphics[width=1.0\linewidth]{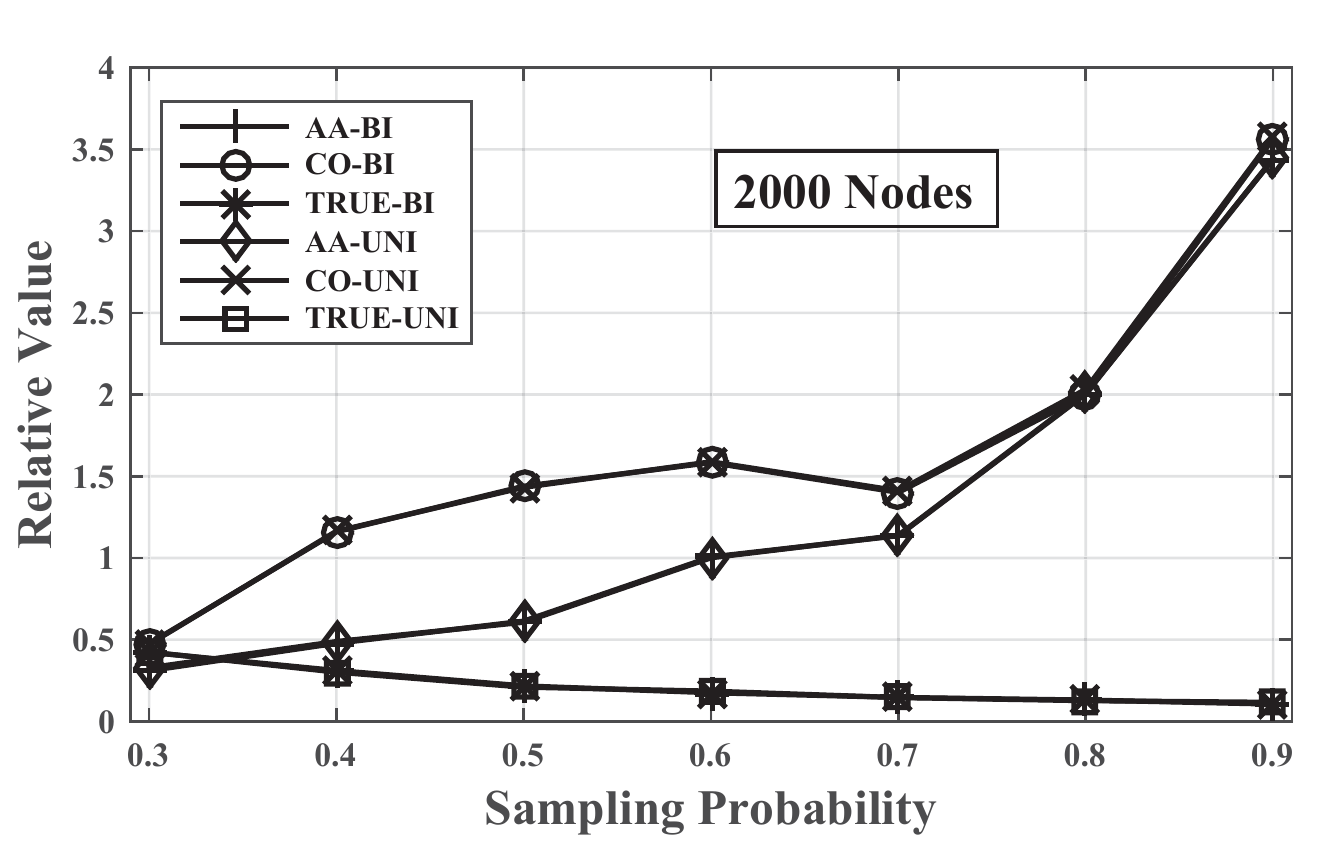}
			\vspace{-4mm}
			\label{fig:2000sam1}
		\end{minipage}%
		
	}

	\vspace{-5mm}
	\caption{\bf\small The relative value of the cost function of the mappings produced by the algorithms on Sampled Social Networks}
	\label{fig:sam_cost}		
\end{figure}

\begin{figure}[htbp]
	\centering
	\subfigure[]{
		\begin{minipage}[]{0.48\linewidth}
			\centering
			\vspace{-3mm}
			\includegraphics[width=1.0\linewidth]{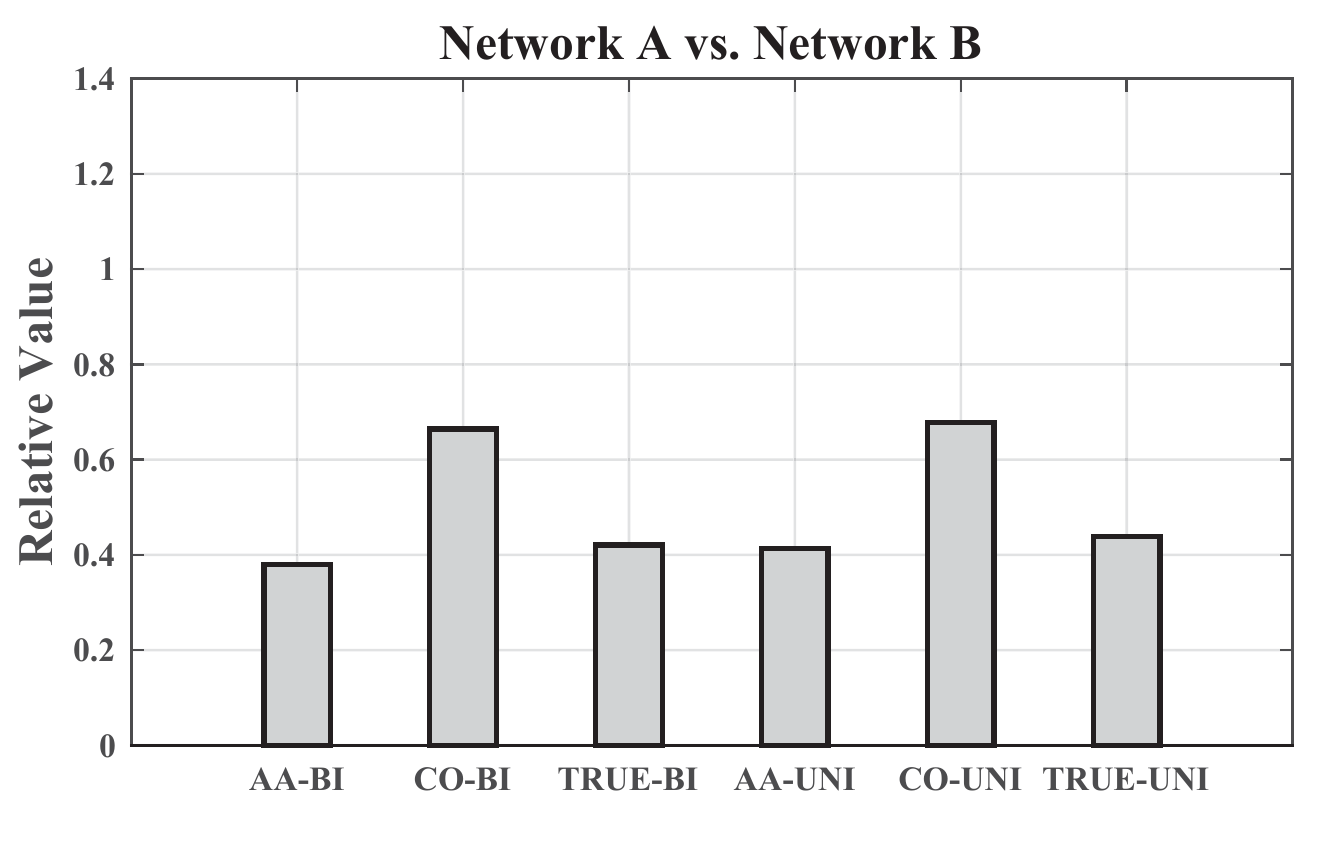}
			\vspace{-3mm}
			\label{fig:COAB1}
		\end{minipage}%
		
	}
	\subfigure[]{
		\begin{minipage}[]{0.48\linewidth}
			\centering
			\vspace{-3mm}
			\includegraphics[width=1.0\linewidth]{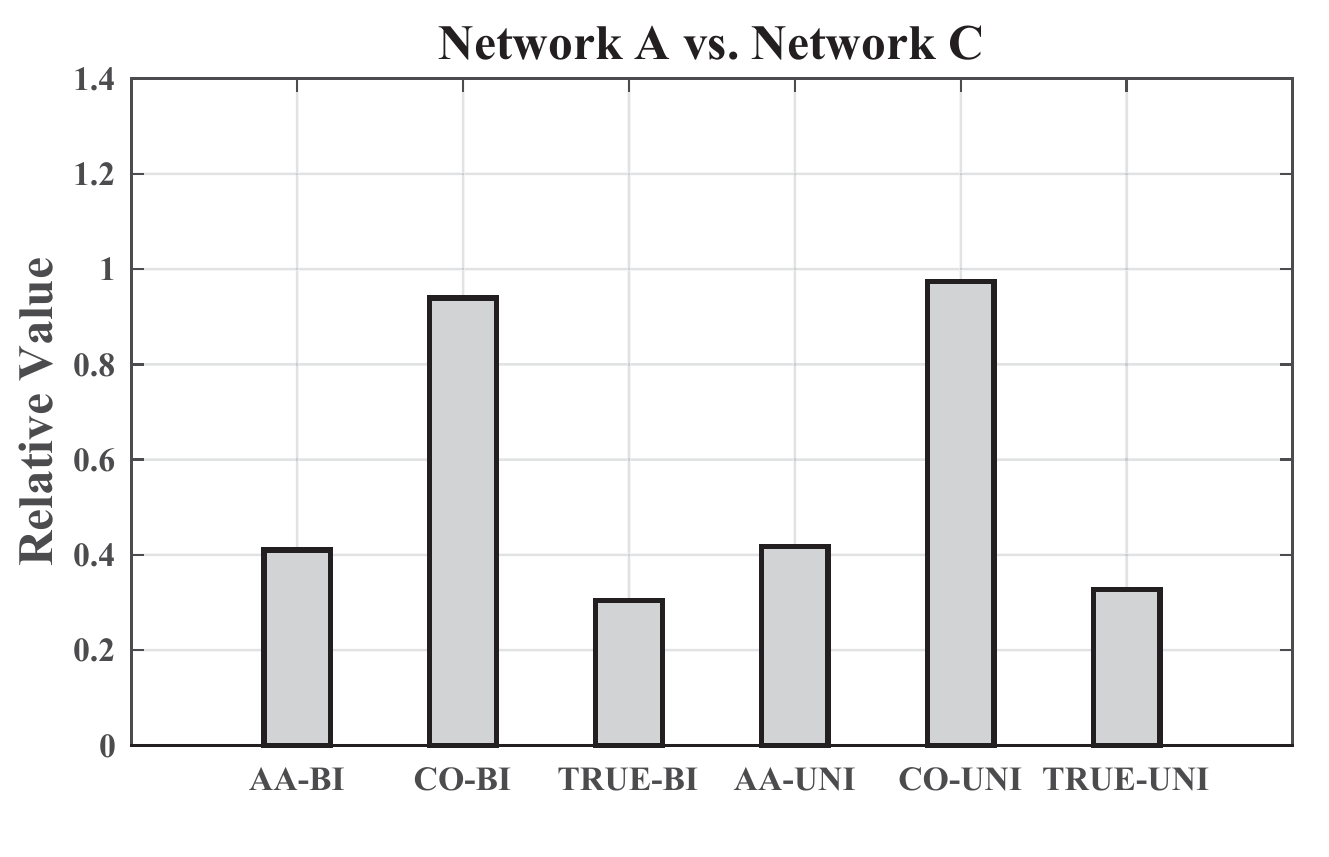}
			\vspace{-3mm}
			\label{fig:COAC1}
		\end{minipage}%
		
	}

	\subfigure[]{
		\begin{minipage}[]{0.48\linewidth}
			\centering
			\vspace{-2.7mm}
			\includegraphics[width=1\linewidth]{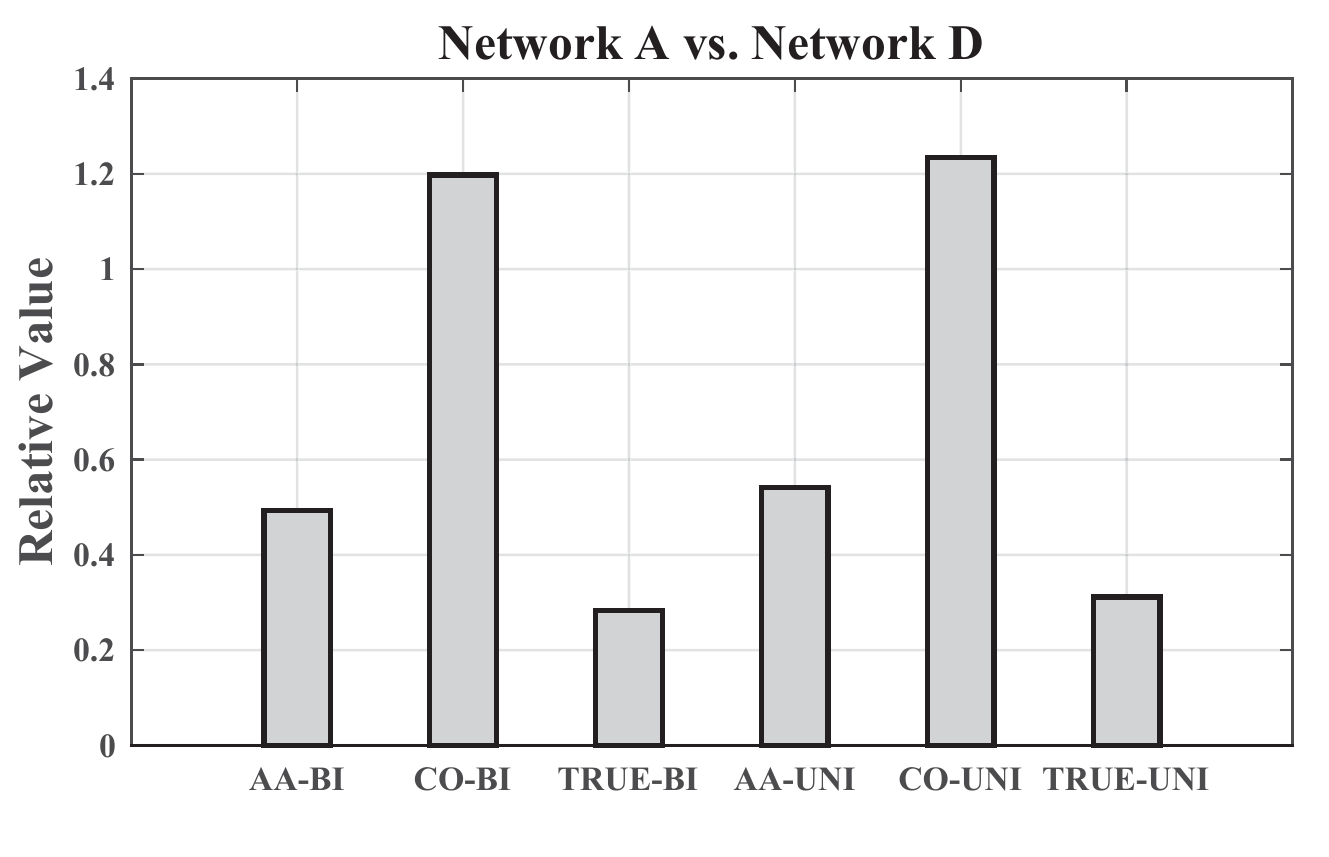}
			\vspace{-3mm}
			\label{fig:COAD1}
		\end{minipage}%
		
	}
	\subfigure[]{
		\begin{minipage}[]{0.48\linewidth}
			\centering
			\vspace{-3mm}
			\includegraphics[width=1.0\linewidth]{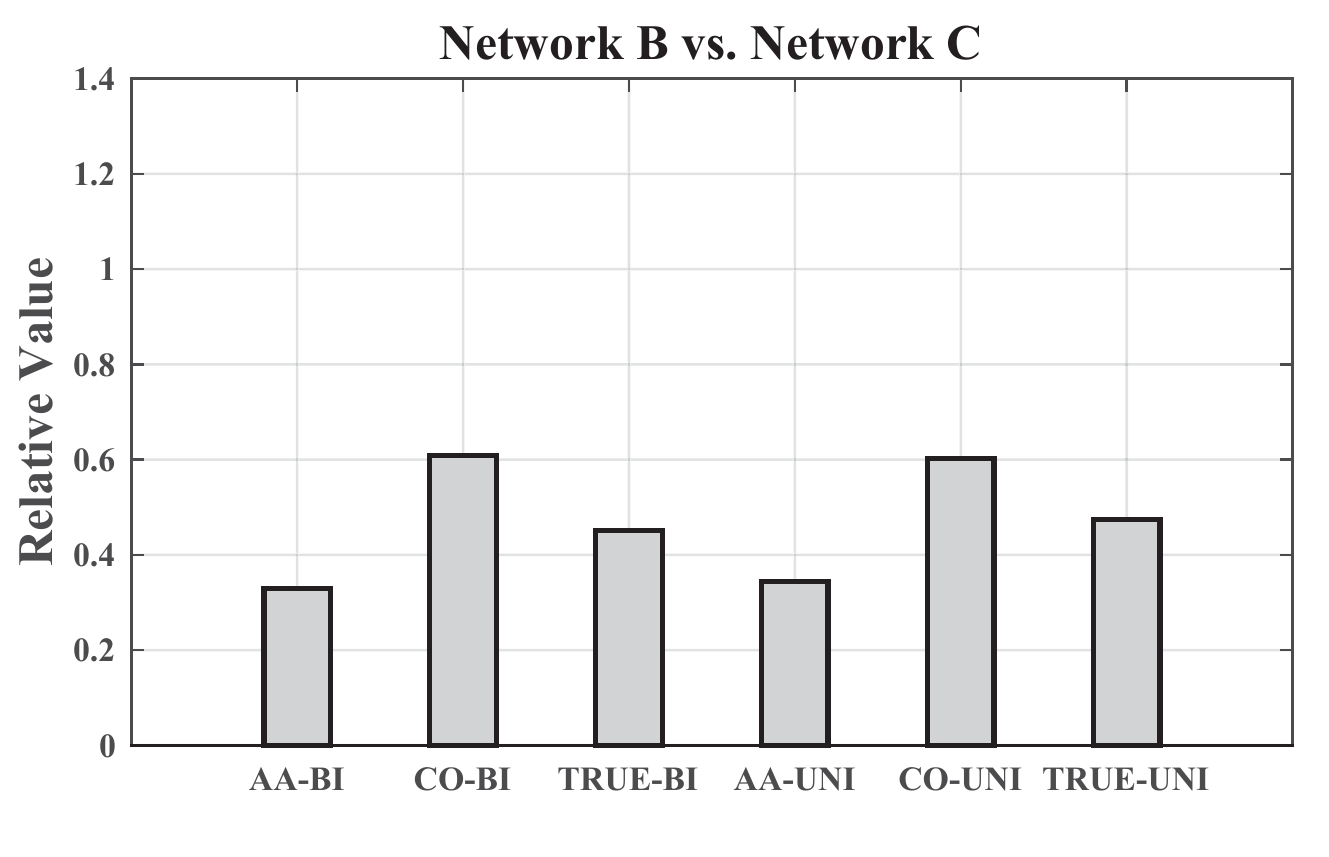}
			\vspace{-3mm}
			\label{fig:COBC1}
		\end{minipage}%
		
	}

	\subfigure[]{
		\begin{minipage}[]{0.48\linewidth}
			\centering
			\vspace{-2.7mm}
			\includegraphics[width=0.98\linewidth]{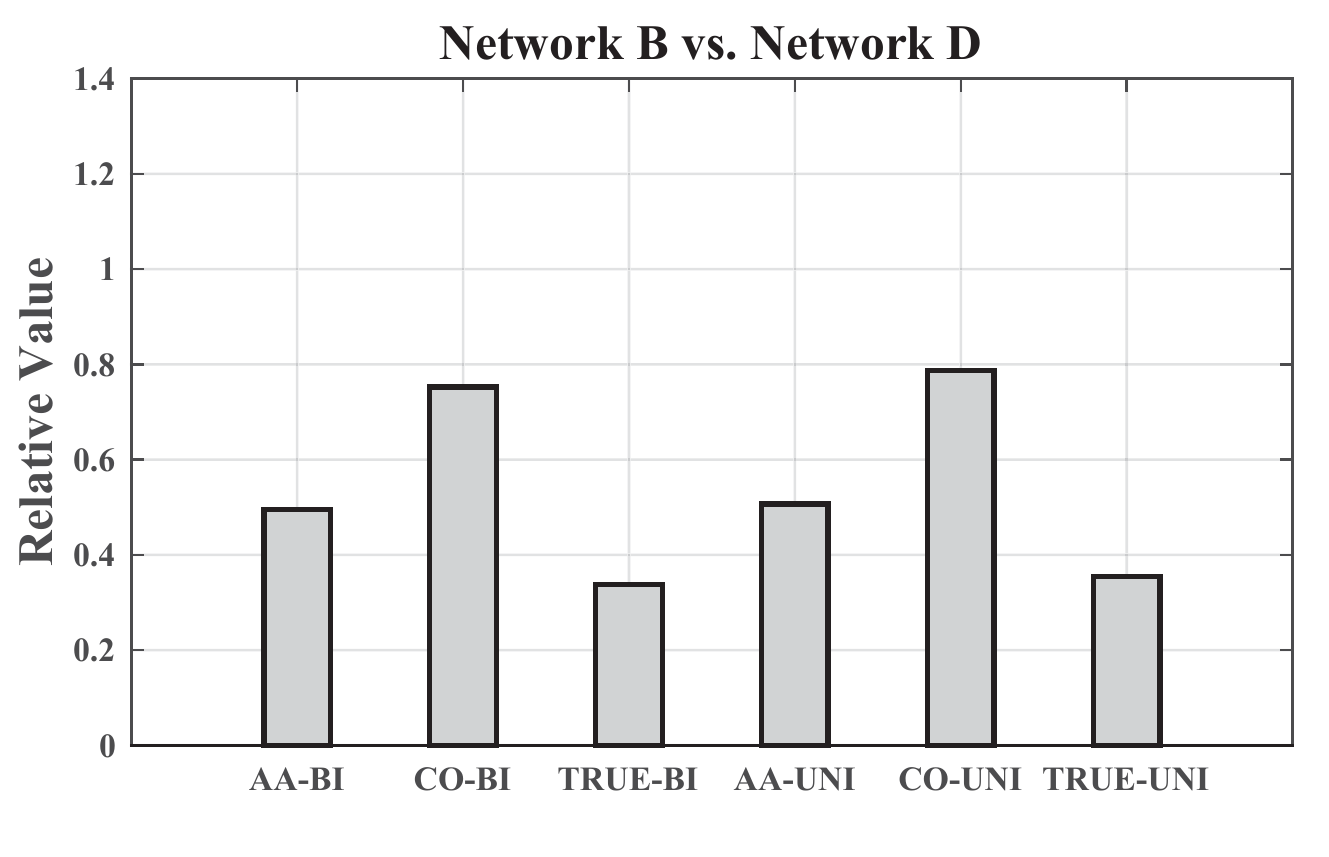}
			\vspace{-3mm}
			\label{fig:COBD1}
		\end{minipage}%
		
	}
	\subfigure[]{
		\begin{minipage}[]{0.48\linewidth}
			\centering
			\vspace{-3mm}
			\includegraphics[width=1.0\linewidth]{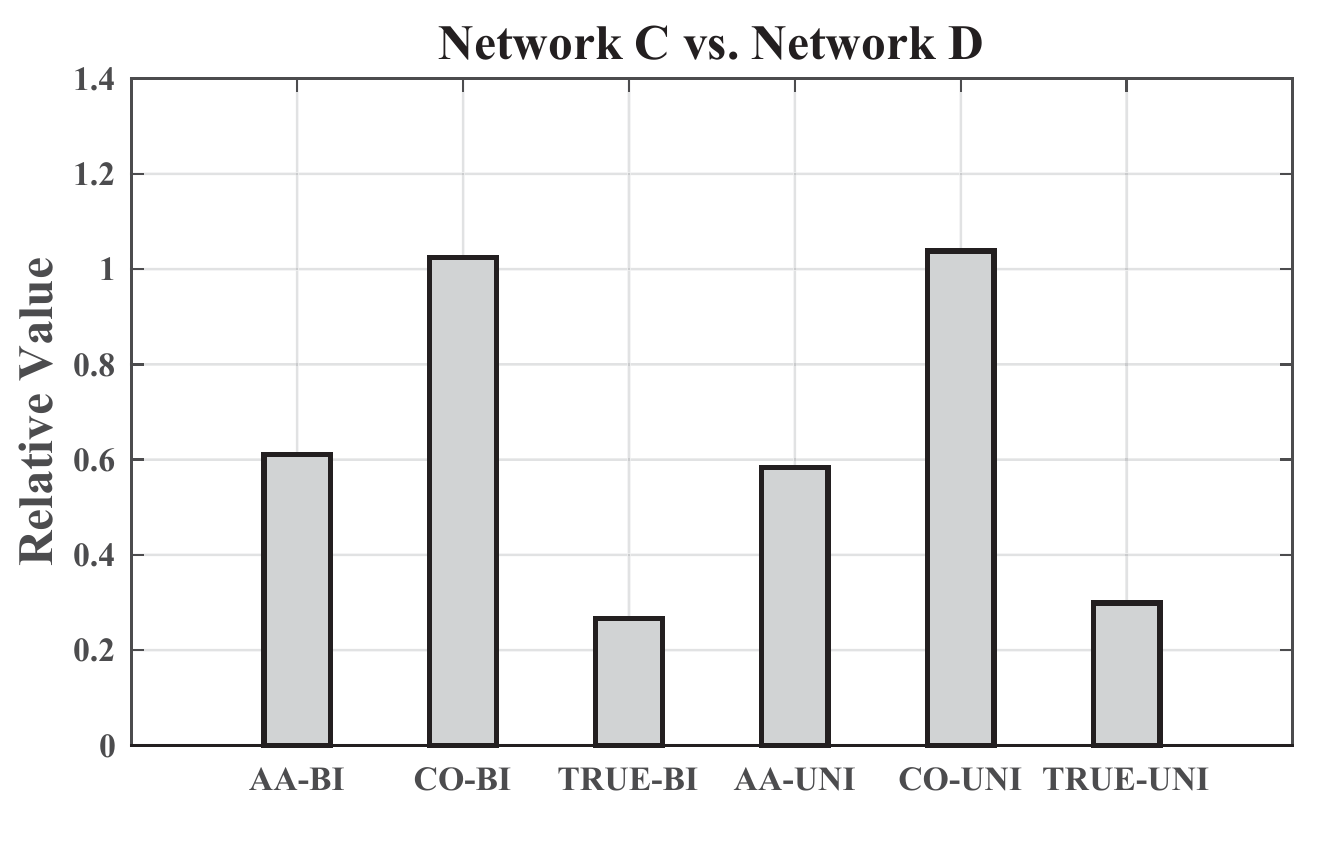}
			\vspace{-3mm}
			\label{fig:COCD1}
		\end{minipage}%
		
	}
	\vspace{-5mm}
	\caption{\small\bf The relative value of the cost function of the mappings produced by the algorithms on Cross-domain Co-authorship Networks}
	\label{fig:coauthor_cost}		
\end{figure}


%
%
%

\end{document}